%% file: arxiv.tex
\documentclass[prd,aps,amsfonts,superscriptaddress,nofootinbib,longbibliography,notitlepage]{revtex4-1}

\include{preamble}

\usepackage{dsfont}

\begin{document}

\title{Quantum memory at nonzero temperature in a thermodynamically trivial system}

\author{Yifan Hong}
\affiliation{Department of Physics and Center for Theory of Quantum Matter, University of Colorado, Boulder CO 80309, USA}

\author{Jinkang Guo}
\affiliation{Department of Physics and Center for Theory of Quantum Matter, University of Colorado, Boulder CO 80309, USA}

\author{Andrew Lucas}
\email{andrew.j.lucas@colorado.edu}
\affiliation{Department of Physics and Center for Theory of Quantum Matter, University of Colorado, Boulder CO 80309, USA}

\date{\today}

\begin{abstract}
Passive error correction protects logical information forever in the thermodynamic limit by updating the system based only on local information and few-body interactions. A paradigmatic example is the classical two-dimensional Ising model: a Metropolis-style Gibbs sampler retains the sign of the initial magnetization (a logical bit) for thermodynamically long times in the low-temperature phase.  Known models of passive quantum error correction similarly exhibit thermodynamic phase transitions to a low-temperature phase wherein logical qubits are protected by thermally stable topological order.  Here, in contrast, we show that certain families of constant-rate classical and quantum low-density parity check codes have no \emph{thermodynamic} phase transitions at nonzero temperature, but nonetheless exhibit \emph{ergodicity-breaking} dynamical transitions: below a critical nonzero temperature, the mixing time of local Gibbs sampling diverges in the thermodynamic limit.  Slow Gibbs sampling of such codes enables fault-tolerant passive quantum error correction using finite-depth circuits. This strategy is well suited to measurement-free quantum error correction, and may present a desirable experimental alternative to conventional quantum error correction based on syndrome measurements and active feedback. 
\end{abstract}

\maketitle

\section{Introduction}

There is a deep relationship between error correction and thermodynamics.  Data storage in magnetic devices (e.g. hard disk drives) relies on the fact that there are two stable magnetic configurations in a ferromagnet, such as the Ising model, in $d\ge 2$ spatial dimensions.  Whether the spins point up or down determines whether a classical bit of information is 0 or 1.  Crucially, the ferromagnetic phase is robust to thermal noise: at low temperatures, errors are \emph{passively} corrected by the environment itself! The time until the stored data becomes corrupted is exponentially long in the system size \cite{Thomas1989}.

The simplest analogy to the above story for a quantum error-correcting code is the four-dimensional toric code \cite{Dennis_2002}.  Intuitively, it corresponds to ``two copies of the 2D Ising model": since a quantum code must protect against both $X$ (bit-flip) and $Z$ (phase-flip) errors, half of the dimensions of the 4D toric code protect against each.  Similar to how the 2D Ising model exhibits ferromagnetic order below a nonzero critical temperature \cite{Onsager_1944}, the 4D toric code also has a nonzero critical temperature below which its thermal state contains \emph{topological order} \cite{alicki2008thermal}.  This topological order can protect logical quantum information forever in the thermodynamic limit, even in the presence of thermal noise.  In other words, the 4D toric code is a passive quantum memory.

In all known examples of passive quantum error correction, there is a thermodynamic phase transition at nonzero temperature to a topologically ordered phase \cite{alicki2008thermal, Yoshida_2011, Hastings_2011,Brown_2016,lieu}.  Such transitions are only known to exist in at least four spatial dimensions. Unfortunately, we live in three spatial dimensions.  Therefore, many existing attempts to realize quantum error correction in experiment involve \emph{active} decoding, where projective measurements and feedback are necessary to protect quantum information.  These active decoders typically require \emph{global} information about measurement outcomes, and so the cost of performing such classical computations will increase as one builds ever-larger quantum computers.  In contrast, a passively-decodable system corrects its own errors simply by thermalizing with a cold environment -- far simpler, at least in principle, to implement experimentally.

Motivated by the above challenges, we address the question of whether quantum memory can exist without any thermodynamic transition to a topologically ordered phase. An answer to this question is important, both because it will constrain the landscape of passively-correctable quantum codes, and because it is important to understand whether the limited nonlocality \cite{xu2023,Hong:2023trf} that can be realized in quantum hardware will enable passive quantum error correction.  

We will prove that a thermodynamic phase transition, at any nonzero temperature, is not necessary for finite-temperature self-correction in classical codes, and that a thermodynamic phase transition to a topologically ordered phase is not necessary for self-correction in quantum codes. Our analysis is of classical and quantum low-density parity-check (LDPC) codes \cite{Gallager_1962, Breuckmann_2021}.  ``Parity check" refers to the multi-bit generalizations of the ferromagnetic interactions in the Ising model; like in the Ising model, the interactions are frustration-free -- they can all be satisfied simultaneously if all bits are in the 0 state (all spins up).  Furthermore, the Hamiltonian of an LDPC code is $\mathsf{k}$-local -- each term in the Hamiltonian only couples $\mathsf{k}=\mathrm{O}(1)$ degrees of freedom in the thermodynamic limit. However, there is no constraint on the physical range of these interactions when the interacting degrees of freedom are embedded in physical (three-dimensional) space. In fact, a typical LDPC code cannot be locally embedded in any finite dimension.   The parity checks are ``low-density" because they are $\mathsf{k}$-local, and each bit only participates in a finite number of parity checks.  In this respect, these models are similarly ``few-body" to the simpler Ising model, but can avoid stringent constraints on quantum error correction in low spatial dimensions \cite{Bravyi_2010}.

The ``infinite-dimensionality" of LDPC codes naturally leads to two properties, each of which are critical for our story.  Firstly, LDPC codes can be non-redundant -- every single parity check is an \emph{independent} constraint from the others.  This is \emph{not} true in the two-dimensional Ising model, where the product of interactions around a square plaquette is always $+1$.  This non-redundancy ensures thermodynamic triviality of the random LDPC codes that we study.  Secondly, a LDPC codes can possess \emph{rapidly growing energy barriers} to small perturbations \cite{Sipser_1996}, implying that the ground states of its corresponding Hamiltonian reside in very deep local minima in the energy landscape.  The large energy barriers to escape these local minima ensure that Gibbs sampling \cite{Metropolis_1953, Hastings_1970}, a specific kind of passive error correction, is capable of trapping the system near its starting state, which enables the eventual decodability of the information.

Previous works on random, classical LDPC code ensembles have demonstrated both the thermodynamic triviality \cite{Montanari_2006} as well as the existence of a dynamical transition reminiscent to those in spin glasses \cite{Kirkpatrick_1987,Montanari_2006_2}. For completeness and pedagogy, we will derive these results explicitly.  We can then more directly explain the main result of this paper, which is the separation between quantum self-correction and thermodynamic phase transitions, which is qualitatively similar (but quantitatively different) to the classical story.


\section{Classical codes}

We first describe random classical LDPC codes, demonstrating both the simultaneous thermodynamic triviality and the self-correcting behavior.    Informally, a classical LDPC code is a collection of $n$ physical bits $Z_i = \pm 1$ and $m$ parity checks of the form $S = \lbrace s_1,\ldots, s_{\Delta_C}\rbrace$, which correspond to subsets of physical bits of size $\Delta_C$ in which the bits will interact via the following Hamiltonian: \begin{equation}
    \mathcal{H} = -\sum_{S=1}^m \prod_{j=1}^{\Delta_C} Z_{S_j} \, .
\end{equation}
Here $S_j$ denotes the $j^{\mathrm{th}}$ bit in $S$.  The $2^k \geq 2^{n-m}$ ground states of the above Hamiltonian are called logical codewords, generated by $k$ logical bits. Notice that one codeword ($\mathcal{H}=-m$) is the state $Z_i=1$ for all $i$.
We restrict our focus to codes in which each bit $Z_i$ is contained in exactly $\Delta_B$ checks, and take $\Delta_B<\Delta_C.$  It is conventional in the coding literature to organize the data above as follows: letting $\mathbb{F}_2=\lbrace 0,1\rbrace$, we consider an $\mathbb{F}_2$-valued $m\times n$ matrix $H$ such that
\begin{equation}\label{eq:parity check}
    H_{Si} = \left\lbrace\begin{array}{ll} 1 &\ \text{parity check $S$ contains bit $i$} \\ 0 &\ \text{otherwise}\end{array}\right..
\end{equation}
Defining vector $\mathbf{z}\in\mathbb{F}_2^n$ as $z_i = \frac{1}{2}(1-Z_i)$, \begin{equation} 
    \mathcal{H} = 2|H\mathbf{z}|-m.
\end{equation}
where $|\cdots |$ denotes the Hamming weight, or number of 1s, of any $\mathbb{F}_2$-valued vector. The code distance
\begin{equation}
    d = \min_{\text{codeword }\textbf{x}\ne\mathbf{0}} |\mathbf{x}|
\end{equation}
can be interpreted as the minimum number of bit flips to transition between codewords. These parameters are often conveniently packaged using the bracketed notation $[n,k,d]$.

Let us now review some key facts about \emph{randomly chosen} LDPC codes $H$ subject to the above constraints (formal statements and proofs are relegated to Appendix \ref{app:classical LDPC}) \cite{Sipser_1996,ModernCodingTheory}: (\emph{1}) almost surely as $n\rightarrow \infty$, for $\Delta_B\geq 3$, the rate of the code obeys
\begin{equation}
    \frac{k}{n} = \mathfrak{r} = 1-\frac{\Delta_B}{\Delta_C} \, . \label{eq:main_rate}
\end{equation}
In other words, the number of logical bits scales linearly with the system size $n$. (\emph{2}) There exist O(1) numbers $\gamma,\eta > 0$ such that for all  $\mathbf{z}$ obeying $0<|\mathbf{z}| \leq \gamma n$, $|H\mathbf{z}|\ge \eta |\mathbf{z}|$. This property is known as linear confinement, and it implies that a continuous path of bit flips between codewords will necessarily cross an extensive energy barrier. Since this confinement holds up to $\gamma n$, the code distance must also scale linearly with the system size.

Using these facts and following \cite{Montanari_2006, Yoshida_2011, freedman2013, Weinstein_2019}, we now show that this LDPC code has a trivial free energy in the thermodynamic limit.  Observe that since $k=n-m$, $H$ has no left null vectors.  Hence, given any $\mathbf{s}\in \mathbb{F}_2^m$, there are exactly $2^k$ candidate states $\mathbf{z}$ satisfying $\mathbf{s} = H\mathbf{z}$. It is then straightforward to evaluate 
\begin{equation}
    Z(\beta) = \sum_{\mathbf{z}} \mathrm{e}^{-\beta \mathcal{H}(\mathbf{z})}= \sum_{\mathbf{z}} \mathrm{e}^{-\beta (2 |H\mathbf{z}|-m)} = 2^k\sum_{\mathbf{s}} \mathrm{e}^{-\beta(2| \mathbf{s}|-m)} = 2^{\mathfrak{r}n} (2\cosh \beta)^{(1-\mathfrak{r})n}.
\end{equation}
The free energy density $f(\beta) = -n^{-1}\beta^{-1}\log Z(\beta)$ is clearly an analytic function for $0\le \beta < \infty$, which means that there is no thermodynamic phase transition. Appendix \ref{app:KW duality} provides another derivation of this fact based on generalizing classical Kramers-Wannier duality to generic parity-check codes with redundancy; yet another (quantum Kramers-Wannier) perspective on this result is found in \cite{Rakovszky:2023fng}.

Let us now study the dynamics of Gibbs sampling for this LDPC code.  We consider the discrete-time Metropolis algorithm which flips one bit at a time: \begin{equation}
    \mathbf{P}[\mathbf{z}\rightarrow \mathbf{z}^\prime \ne \mathbf{z}] = \frac{1}{n} \min\left(1, \mathrm{e}^{-\beta(\mathcal{H}(\mathbf{z}^\prime)-\mathcal{H}(\mathbf{z}))} \right) \, \mathbb{I}\left(|\mathbf{z}-\mathbf{z}^\prime|=1\right)  \, ,  \label{eq:db}
\end{equation}
where $\mathbb{I}(\cdot)$ is the indicator function which evaluates to 1 if its argument is true and 0 otherwise. Suppose that at time $t=0$, we start in the codeword $\mathbf{0}$.  To bound the time it takes for the Gibbs sampler to destroy information, it suffices to bound the time it takes for the Gibbs sampler to reach a state $\mathbf{z}_t \in E$, where the ``bottleneck" set $E$ consists of all $|\mathbf{z}_t| = \gamma n/2 < d/2$.   In thermal equilibrium, we observe that \begin{equation}
    \frac{\pi[E]}{\pi[\mathbf{0}]} \le \mathrm{e}^{-\beta \eta \gamma n} \,{n \choose \gamma n/2} \lesssim \mathrm{e}^{-(2\beta \eta -  \log(2/\gamma) - 1  )\gamma n/2} \, , \label{eq:piE/pi0}
\end{equation}
where $\pi$ denotes the Gibbs probability distribution: $\pi(\mathbf{x}) = \mathrm{e}^{-\beta \mathcal{H}(\mathbf{x})}/Z(\beta)$.  The Markov chain bottleneck theorem then shows that the Gibbs sampler, starting from codeword $\mathbf{0}$, will not reach the set $E$ before a typical time scale $t\gtrsim \mathrm{e}^{\mathrm{O}(n)}$ at low temperature. Below $E$, all states contain less than $d/2$ errors, and so are correctable in principle (under maximum-likelihood decoding). Specific decoders may require additional constraints; e.g., bit-flip can provably correct errors of weight up to $\gamma\eta\Delta_B^{-1} n$ for $\Delta_B\geq 5$ \cite{Sipser_1996}; these constant offsets do not change the overall behavior and can be accounted for by adjusting the definition of $E$. Since all transitions are equally probable at $T=\infty$ -- the system will fully mix after we consider flipping each spin once -- we deduce there must be a finite temperature $T_{\mathrm{c}}$ such that for $T<T_{\mathrm{c}}$, the system has (exponentially) slow mixing; our calculation provides a lower bound on $T_{\mathrm{c}}$.

\begin{figure*}[t]
\centering
\includegraphics[width=0.75\textwidth]{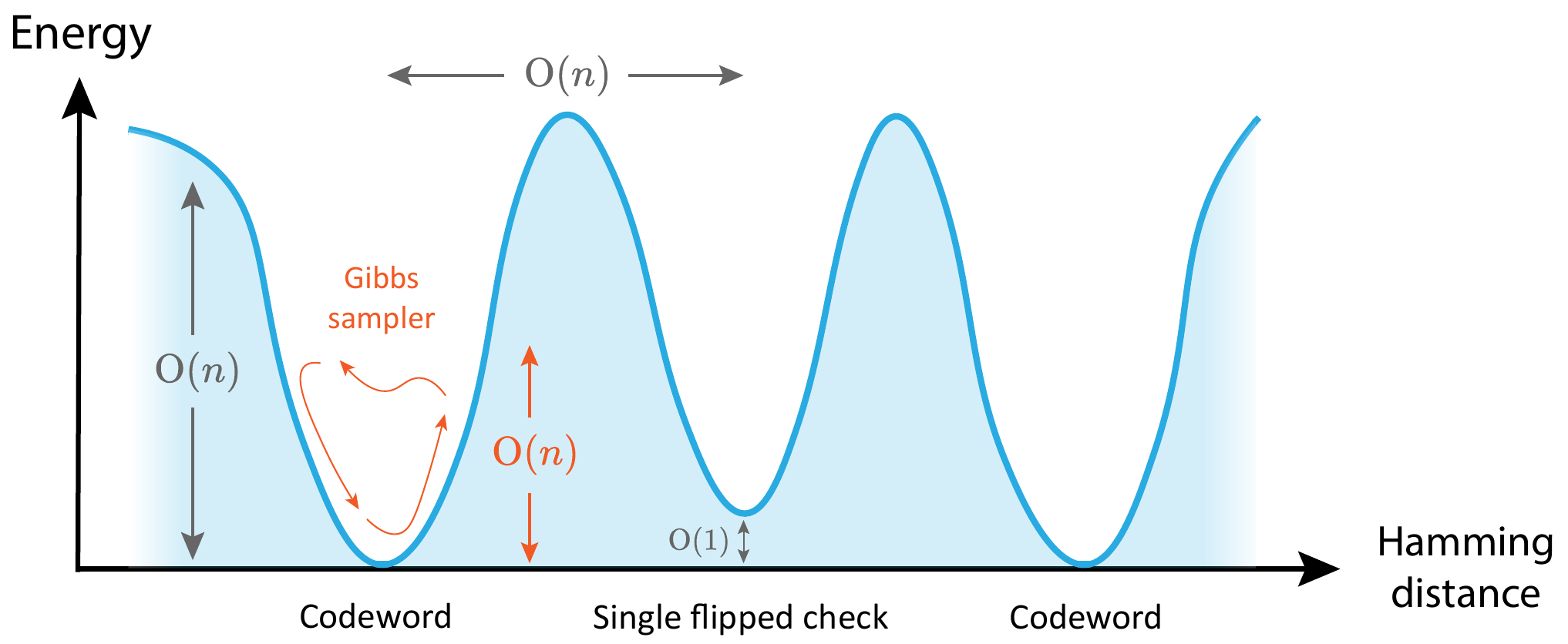}
\caption{A snapshot of the energy landscape of a typical good classical LDPC code is depicted. Extensively deep minima separate codewords and low-energy configurations consisting of single flipped parity checks. Below a critical temperature, the Gibbs sampler becomes trapped inside its initial minimum and is unable to explore the full state space, preventing the system from reaching thermal equilibrium.}
\label{fig:energy landscape}
\end{figure*}

How is it possible to have thermodynamic triviality, while the Gibbs sampler simultaneously protects a codeword?  The answer lies in the peculiar energy landscape of these LDPC codes, illustrated in Figure \ref{fig:energy landscape}.  Locally near each codeword, the linear confinement property guarantees a very deep minimum in the energy landscape. If we could restrict $Z(\beta)$  to only count states $\mathbf{x}$ with $|\mathbf{x}| \le \gamma n$, the Helmholtz free energy would exhibit a phase transition: above a critical temperature, the typical state would have $|\mathbf{x}|/n \rightarrow \gamma$, while below it the typical state is much closer to the codeword $\mathbf{0}$, as guaranteed by (\ref{eq:piE/pi0}).   In reality however, $Z(\beta)$ counts the \emph{whole} state space, and when we look beyond the local landscape near a given codeword, we see a ``swamp" of many very deep minima, some of which contain other codewords, but far more of which contain ``sick" configurations where, e.g., a single parity check is flipped \cite{Montanari_book, Rakovszky:2023fng, Rakovszky:2024iks}.  Indeed, we have seen that there must exist a state with a single flipped parity-check since $H$ has no left null vectors; at the same time, such a state will be far from all codewords, and it too will have an extremely deep energy well where the Gibbs sampler would be stuck.  The key point is the conceptual distinction between \emph{static} equilibrium and \emph{dynamical} equilibration -- the thermodynamic free energy detects \emph{all} of the exponentially many fake codewords, while the Gibbs sampler is stuck near whichever minimum it starts in.

To have a genuine thermodynamic phase transition, redundancy must be added to the parity checks, such that the sick configurations of Figure \ref{fig:energy landscape} gain finite energy density. The energetic suppression of these sick configurations is encapsulated by a closely related property to confinement known as soundness, which stipulates that small energy violations can always be produced by small errors. Note that the LDPC codes we have discussed are confined but not sound due to the low-energy sick configurations. The 2D Ising model, on the other hand, is both (sublinearly) confined and sound. The intimate relationship between redundancy and soundness has been studied in the context of locally testable codes (LTCs) \cite{Sasson_2009}, in which one can reliably determine whether a state is a codeword by only querying a small number of physical bits. ``Good" locally testable codes with constant rate, constant relative distance and constant query complexity ($c^3$-LTCs) have recently been found \cite{Panteleev_2022, lin2022, dinur2022LTC}. While redundancy and soundness is a feature of these codes, we have seen that they are unnecessary to have ergodicity-breaking at low temperatures.


\section{Quantum codes}

We now turn to a summary of analogous results for quantum LDPC codes \cite{Breuckmann_2021} of the Calderbank-Shor-Steane (CSS) type \cite{Calderbank_1996, Steane_1996}. We begin with two classical LDPC codes with $M\times N$ parity-check matrices $H_X$ and $H_Z$ respectively, subject to the orthogonality constraint $H_X H_Z^{\mathrm{T}} = 0$ (here the matrix multiplication uses $\mathbb{F}_2$ algebra).  For each row of $H_X$, we define an $X$-stabilizer which is the product of a Pauli $X$ acting on each qubit in the parity check; for $H_Z$, the stabilizers consist of Pauli $Z$s.  The orthogonality constraint above ensures that the $X$-type and $Z$-type stabilizers commute, and hence can be simultaneously diagonalized.  The $+1$ eigenspace of all stabilizers forms the Hilbert space of logical qubits, or codewords. Logical $X$/$Z$ operators correspond to products of Pauli $X$/$Z$ that commute with all $Z$/$X$-type stabilizers; if we have $K$ of each, then the code stores $K$ logical qubits.  The code distance $D$ is the length of the smallest such logical operator.  We define a quantum Hamiltonian to be the sum over these stabilizers: \begin{equation}
    \mathcal{H} = - \sum_{\text{$X$-stabilizer } S} \,\prod_{i\in S} X_i \;- \sum_{\text{$Z$-stabilizer } S} \,\prod_{i\in S}Z_i \, .
\end{equation}

It is quite challenging to find a quantum LDPC code of $N$ qubits in which the code distance $D=\mathrm{O}(N)$, although recently some have been found \cite{Panteleev_2022, qTanner_codes, dinur2022good}.  To keep the discussion as transparent as possible, we will instead discuss the simplest family of quantum LDPC codes, called hypergraph product (HGP) codes \cite{HGP}.  HGP codes are formed by a graph product of two classical codes, each of which we take to be a good random LDPC code of the kind previously mentioned with parameters $[n,k=\mathrm{O}(n),d=\mathrm{O}(n)]$.  The HGP construction is illustrated in Figure \ref{fig:HGP construction}, while explicit formulas are relegated to Appendix \ref{app:HGP codes}.  One can prove the following important facts about such random HGP codes, each of which holds almost surely in the thermodynamic limit: (\emph{1}) they are constant-rate codes with $N=n^2+(n-k)^2$, $D=d$, and $K/N = \mathfrak{r}^2$, with $\mathfrak{r}$ inherited from the classical LDPC code and given in (\ref{eq:main_rate}).  There are $M=n(n-k)$ checks of both $X$ and $Z$ type, and there are no redundant checks. (\emph{2}) Like classical LDPC codes, they can exhibit linear confinement \cite{Leverrier_2015}: given any binary string $\mathbf{x}\in\mathbb{F}^N_2$ corresponding to an $X$-type or $Z$-type error with weight $|\mathbf{x}|\le \gamma n$, we have $|H_X\mathbf{x}|,|H_Z\mathbf{x}|\ge \eta |\mathbf{x}|$ for some $\eta=\mathrm{O}(1)$. (\emph{3}) Every stabilizer involves exactly $\Delta_B+\Delta_C$ qubits, and each qubit is involved in at most $\Delta_C$ $X$-stabilizers and $\Delta_C$ $Z$-stabilizers.

\begin{figure*}[t]
\centering
\includegraphics[width=0.85\textwidth]{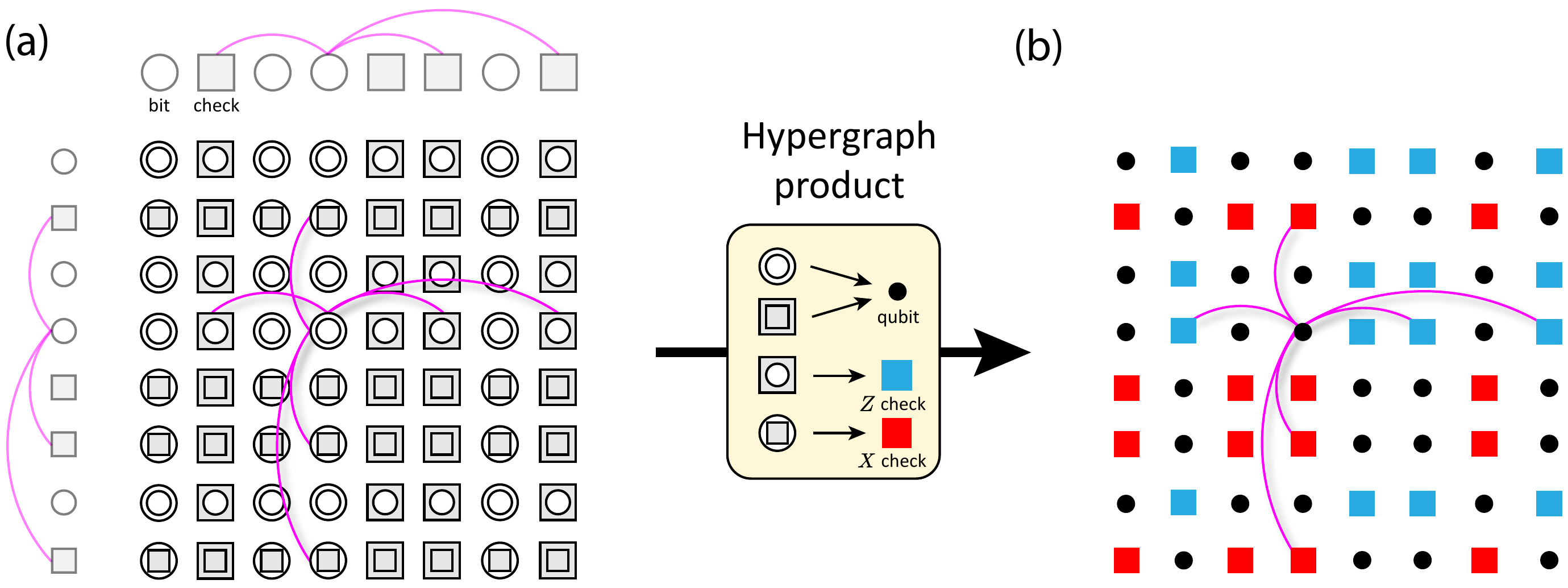}
\caption{The hypergraph product construction is illustrated. (a) The Euclidean graph product is taken between the Tanner graph of a classical parity-check code and itself, producing four types of vertices. (b) The hypergraph product maps the four types of vertices in the product graph to qubits, $X$-checks and $Z$-checks, thus producing a quantum CSS code. The interactions in the hypergraph product code resemble those of its classical input code (in magenta).}
\label{fig:HGP construction}
\end{figure*}

With these facts in hand, it is straightforward to show that there is no finite-temperature thermodynamic phase transition in such models. The proof is exactly analogous to that of the classical models: as the stabilizers mutually commute, we can evaluate the partition function in subspaces of fixed stabilizer eigenvalues, each of which has Hilbert space dimension $2^K$. Therefore, since the stabilizers are linearly independent,
\begin{equation}
  Z(\beta)=  \mathrm{tr}\left(\mathrm{e}^{-\beta \mathcal{H}}\right) = 2^{K} \left(2\cosh \beta\right)^{N-K}.
\end{equation}
The free energy density $F(\beta) = -N^{-1}\beta^{-1}\log Z$ is again an analytic function for $0\le \beta < \infty$, so there is no thermodynamic phase transition.

However, once again, Gibbs sampling is exponentially slow below a critical temperature.  Since the Hamiltonian $\mathcal{H}$ is a sum of independent $X$-checks and $Z$-checks, our quantum Gibbs sampler will simply measure the set of e.g. $X$-stabilizers containing a given qubit $i$, and with suitable probability analogous to (\ref{eq:db}), apply Pauli $Z_i$ to the state afterwards. We prove in Appendix \ref{app:quantum thermal decoding} that the time it takes to reach a state for which known decoders might fail to decode is $t \gtrsim \mathrm{e}^{\mathrm{O}(\sqrt{N})}$ for $\beta$ above some critical temperature.  The proof strategy is similar to the one for classical LDPC codes, but is slightly more involved as for HGP codes we do not have $D\sim N$.  We use the fact that errors typically form disconnected clusters, and these clusters rarely conspire to fool an eventual decoder.  States with such adversarial clusters are, at low temperature, found in the Gibbs ensemble with probability $\mathrm{e}^{-\mathrm{O}(\sqrt{N})}$, and so it will take us a long time to reach one. This is analogous to Peierls's argument that the two-dimensional Ising model is ordered at low temperatures because we are unlikely to find long domain walls \cite{Peierls_1936}.  Numerical simulations confirming the slow dynamics are presented in Appendix \ref{app:numerics}.

We understand the existence of this ergodicity-breaking in Gibbs sampling in exactly the same way as illustrated for classical LDPC codes.  The quantum LDPC codes have a very complex energy landscape, such that the Gibbs sampler will get stuck in a very deep energy minimum for exponentially long times; in contrast, the \emph{thermodynamics} is sensitive to the ``swamp" of exponentially many deep local minima corresponding to very large errors that flip few parity checks.   A thermal phase transition requires redundancy which our HGP codes do not have.  Numerical simulations had previously suggested that passive decoding might be possible even above a nonzero critical temperature for topological order \cite{melko}; we stress that in our construction, the model is completely thermodynamically trivial. Nonvanishing energy barriers have been the primary features of Haah's cubic code \cite{haah1} and welded codes \cite{Michnicki_2014}. Nonetheless, the lifetime of logical information in these codes is finite at any positive temperature due to entropic effects \cite{haah2, Siva_2017}.

Despite the qualitative similarity between the coexistence of thermodynamic triviality, and low-temperature self-correction, classical self-correction likely does \emph{not} directly imply self-correction in a quantum code formed from the hypergraph product of the classical codes: we provide both analytical and numerical arguments in Appendix \ref{app:tree codes} that there exist classical self-correcting codes whose hypergraph product is not a quantum self-correcting memory. This is because, as we explain, confinement in classical codes does not directly correspond to confinement in the hypergraph product quantum code.


\section{Measurement-free quantum error correction}

Due to the spatial nonlocality of the $\mathsf{k}$-local LDPC codes, it may be challenging to simply ``let them thermalize" with a cold bath: physical qubits must be arranged in three-dimensional space, and the most efficient way to implement the LDPC interactions may require dynamical motion of the physical qubits. Nevertheless, the existence of passive decoders enable new experimental paradigms for error correction when compared to spatially local (e.g. surface) codes.  One strategy uses ``measurement-free quantum error correction" (MFQEC) \cite{Ahn_2002, Sarovar_2005}, in which the user can only apply unitary gates and a specific ``reset" dissipative quantum channel to the physical qubits. In the many-body setting, MFQEC typically proceeds by engineering local dissipation \cite{Pastawski_2011}. More formally, given density matrix $\rho$ for the physical qubits, we can only apply $\rho \rightarrow U\rho\, U^\dagger$ for unitary $U$, or apply $\rho \rightarrow R_i[\rho]$ where channel $R_i$ acts on single qubit $i$ as $R_i[\rho_i] = |0\rangle\langle 0|$; here we defined $Z|0\rangle = |0\rangle$.  In addition to these processes, we also have errors occurring in the system, which we take to be independent and few-qubit, albeit occurring at a constant rate per qubit. The strategy behind MFQEC is to introduce ancilla qubits in addition to the $N$ physical qubits of the code.  At the beginning of each error correction round, we apply the reset channel such that the state of the system is $|\psi_{\mathrm{total}}\rangle = |\psi_{\mathrm{phys}}\rangle \otimes |\mathbf{0}\rangle_{\mathrm{anc}}$.
Next, we encode the state of stabilizers into the $Z$-eigenstate ancilla qubits. In conventional quantum error correction, we would then measure these ancilla qubits in the $Z$-basis, and apply feedback based on the outcomes. In MFQEC,  we instead just directly apply feedback to the physical qubits based on multi-qubit gates which couple them again to the ancillas. After this feedback is applied, we reset the ancillas.  Using Stinespring's formulation of measurement \cite{Stinespring,Friedman:2022vqb} wherein measurement outcomes are stored in auxiliary qubits in Hilbert space such that the measurement process becomes unitary, a round of MFQEC is mathematically equivalent to a round of conventional QEC.

If passive decoders rely only on local feedback, then the decoding circuit for MFQEC has \emph{finite depth} and requires only finite ancillas per physical qubit to implement. In Appendix \ref{app:implementation}, we discuss details of how to Gibbs sample with MFQEC and show that such a Gibbs sampler implies the existence of a fault-tolerant decoder. The crucial observation is that a Gibbs sampler is local \emph{and} always introduces errors into the system with non-zero probability. If the net error rate (including errors in syndrome measurement and decoding) $p_{\mathrm{error}} \le \mathrm{e}^{-2\beta\Delta_C}$, where $\beta>\beta_{\mathrm{c}}$ implies that we are below the critical temperature where Gibbs sampling is slow, and $\Delta_C$ is the maximal number of stabilizers that a single qubit error can flip, then we can always \emph{add} additional errors into the system (for qubits where some stabilizers were already violated) such that our error correction scheme realizes a perfect Gibbs sampler at low temperature, which we have proven is a passive decoder.  

In practice, we expect the Gibbs sampler can be replaced by a greedy decoder, such as bit-flip \cite{Sipser_1996}, which always tries to lower the energy of the system. Such a local decoder will \emph{not} in general drive the system \emph{exactly} to its logical codeword.  When information is to be retrieved, it can be done using conventional decoders. Gibbs sampling thus serves as a particular kind of ``single-shot decoder", where a single round of noisy syndrome measurement suffices to suppress errors to a correctable level \cite{Bombin_2015, Campbell_2019}.  Single-shot decoders exist for a variety of quantum LDPC code families \cite{Quintavalle_2021, Higgott_2023}, with some decoders also requiring only local syndrome information \cite{Fawzi_2018, Gu:2023pzw}. The single-shot decoders in \cite{Quintavalle_2021, Higgott_2023} still require global syndrome information, and so direct implementation with MFQEC is not practical. $\mathsf{k}$-locality in passive decoding is crucial to make MFQEC tractable; slow, computationally expensive decoding can be done once information is retrieved out of memory. Gibbs samplers lead to MFQEC decoders with significantly lower circuit depth than the single-shot decoders in \cite{Fawzi_2018, Gu:2023pzw}.

Perhaps the platform where MFQEC-based passive decoding is most desirable is neutral-atom quantum computing \cite{Saffman2010, Kaufman_2021, Cong_2022, Bluvstein_2022, Jenkins_2022, Bluvstein_2023}.  (\emph{1}) In this platform, the reset operation can be about 100 times faster than mid-circuit measurement and feedback \cite{Lis_2023, Norcia_2023, Huie_2023}.  The circuit depth needed for MFQEC is a constant-factor $\lesssim 3$ larger than that needed for QEC based on mid-circuit measurement and feedback (at the cost of adding $\mathrm{O}(N)$ extra ancillas), so we expect that a passive decoder would operate faster.  The primary disadvantage of the passive decoder -- a possibly slow implementation of nonlocal atomic motion -- would be just as problematic for the syndrome measurements themselves.  (\emph{2}) HGP codes are well-suited to the hardware capabilities of atom arrays \cite{xu2023,Hong:2023trf}, in which acousto-optical deflectors can perform row and column permutations nonlocally in space relatively quickly (at least for near-term system sizes).    Hence, spatial nonlocality of the LDPC code is not necessarily prohibitive of its implementation in the absence of a large quantum-repeater network \cite{Fowler_2010, Azuma_2023}. (\emph{3}) Neutral atom platforms have increased their number of physical qubits \cite{Bluvstein_2023} more rapidly than other platforms, which suggests that the additional ancilla qubit overhead for the fastest passive decoding strategy is less burdensome.  Lastly, our numerical simulations suggest that the rigorous bounds on threshold error rates are far more stringent than needed in practice.  A detailed cost-benefit analysis of measurement-free vs. syndrome-based error correction in this platform remains to be performed.

\section{Outlook}

We have proved that classical and quantum LDPC codes with linear confinement are self-correcting memories, in which $N$ physical (qu)bits which undergo local, memoryless, time-reversal-symmetric Gibbs-sampling dynamics at low but nonzero temperature,  protecting $\mathrm{O}(N)$ logical (qu)bits for infinite time in the thermodynamic ($N\rightarrow\infty$) limit.  This property holds even as the codes are thermodynamically trivial -- they have no thermodynamic phase transitions at nonzero temperature.   Our result provides a deeper understanding of the surprising connections between statistical physics and error correction: contrary to prior intuition, thermal phase transitions are unnecessary for passive error correction, which can be performed by sampling the thermal Gibbs state, in either classical or quantum codes. Although our analysis is of random codes, see \cite{Capalbo_2002, Golowich_2023} for explicit constructions with linear confinement. We do not know whether these explicit codes also share the same thermodynamic triviality as their random counterparts.

An immediate application of our results is a local Markovian (memoryless) decoder based on Gibbs sampling that can passively protect quantum information in (expander) hypergraph product codes, which presents an alternative ``decoding'' strategy for error correction compared to existing ones \cite{Leverrier_2015, Panteleev_2021, Quintavalle_2022, Delfosse_2022}; note that while our Gibbs sampler is not guaranteed to return the noisy state exactly back to the codespace, it is guaranteed to keep the noisy state sufficiently close to its initial codeword so that eventually one can successfully decode the final state using one of the existing (full) decoders. Especially for MFQEC, the complexity of the thermal decoder is substantially reduced relative to prior examples.

The self-correcting capability of LDPC codes has a number of intriguing implications and connections with other research thrusts in physics. The LDPC codes we described in this paper are arguably the simplest (albeit high-dimensional) examples of models exhibiting dynamical state/Hilbert space shattering in the thermodynamic limit: the system is dynamically trapped in an exponentially small region of the state space, even in the absence of any protecting symmetry or even simply frustration or Hilbert space constraints.  Our result thus makes a sharp connection between quantum error-correcting codes and recent models of ergodicity breaking in (quantum) statistical physics \cite{sala, Khemani_2020, Hart_2022, Stephen_2024, stahl2023, Han:2024yfm}. In the above models, the systems are closed and evolve according to their own Hamiltonians. The dynamics we study, on the other hand, are local spin flips resembling open system dynamics. For closed-system dynamics of LDPC codes, we refer the reader to a related model exhibiting explicit eigenstate localization \cite{LDPC_MBL}, which is an even more drastic breakdown of ergodicity (to infinite time). The classical LDPC code Hamiltonians we have described  \cite{Montanari_2006,Montanari_2006_2} are also qualitatively different from many common models of spin glasses \cite{Edwards_1975, Sherrington_1975,parisi} such as the Sherrington-Kirkpatrick model, and random instances of many NP-hard optimization problems \cite{krzakala}, which have analogous ergodicity-broken dynamical regimes at low temperature, yet also have frustration and thermodynamically-detectable phase transitions.

The existence of locally testable, good quantum LDPC codes ($c^3$-qLTCs) is still an open question. These codes are nonlocal generalizations of the 4D toric code, similar to $c^3$-LTCs and the 2D Ising model. Importantly, they will contain linear dependencies amongst their parity checks in a way such that \emph{all} low-energy states will necessarily be close to codewords, with ``close'' being quantified by the codes' soundness functions \cite{Aharonov_2015}. For example, in the 4D toric code, small error clusters form closed loops of violated checks whose area scales quadratically with its perimeter, which defines the soundness function. Thus, just like the 4D toric code, we expect $c^3$-qLTCs to possess genuine thermodynamic phase transitions. Historically, since all the prototypical examples of passive quantum memories contained thermodynamic phase transitions, the natural object to study for a constant-rate, passive quantum memory would be a qLTC. In contrast to this intuition, we have shown that existing (non-LTC) constructions of constant-rate quantum codes, namely the hypergraph product codes, possess the sufficient ingredients for passive error correction. In addition, since the existing good (constant-rate $K\sim N$ and linear-distance $D\sim N$) quantum LDPC codes \cite{Panteleev_2022, qTanner_codes, dinur2022good} are likely not locally testable, we conjecture that the non-extensive redundancy of random hypergraph product codes is also shared by these code families. If true, our proof immediately generalizes to such models and implies that the circuit complexity of preparing states is \emph{not} a universal characteristic of a thermodynamic phase -- namely, that the circuit depth needed to prepare a state at various energy densities will transition from finite to divergent without crossing any thermal phase transition.  After all, such codes exhibit the no-low-energy-trivial-states (NLTS) property \cite{freedman2013, Anshu_2023} -- arbitrary low-energy states of a Hamiltonian, up to a finite energy density, are topologically ordered: they cannot be constructed from a product state via a finite-depth circuit. Based on a similar intuition as passive memories mentioned above, it used to be expected that NLTS was closely related to local testability of qLTCs \cite{Eldar_2017, Panteleev_2022}; the surprising result that local testability was not needed  \cite{Anshu_2023} may be deeply related to our own observation that thermodynamics and self-correction are not necessarily related.  We currently define two pure states to be in the same phase only when they are related by a finite-depth unitary circuit \cite{Chen_2010}, with the analogy for mixed states less clear  \cite{rakovszky_stable, sang2023, Chen:2023auj}.  We conclude that NLTS in non-redundant codes would thus reveal two states in the same trivial thermodynamic phase, yet which cannot be efficiently prepared from each other. A better understanding of such a discrepancy is an interesting problem for the future, and may sharpen our understanding of how to define phases of matter in quantum many-body systems.

\emph{Note added.---} In forthcoming work, other authors have independently investigated thermodynamics vs. dynamics in LDPC codes \cite{Breuckmann_2024}.

\section*{Acknowledgements}

We thank Oliver Hart, Mike Hermele, Vedika Khemani, Anthony Leverrier, Rahul Nandkishore, Zohar Nussinov, Marvin Qi and Charles Stahl for helpful discussions, Xun Gao for a careful reading of the manuscript, and especially Adam Kaufman for many insights on neutral atoms.  This work was supported by the Alfred P. Sloan Foundation through Grant FG-2020-13795 (AL), the Air Force Office of Scientific Research via Grants FA9550-21-1-0195 and FA9550-24-1-0120 (YH, JG, AL), and the Office of Naval Research via Grant N00014-23-1-2533 (YH, AL). 


 


\appendix
\renewcommand{\thesubsection}{\thesection.\arabic{subsection}}

\section{Classical LDPC codes}

\subsection{Linear parity-check codes}\label{app:classical LDPC}
 A classical linear code $\mathcal{C}$ encodes $k$ logical bits inside $n$ physical bits using a $m \times n$ parity-check matrix with binary $\mathbb{F}_2 \simeq \mathbb{Z}_2$ coefficients (written as $H\in \mathbb{F}_2^{m\times n}$). From $H$ we can construct a $k \times n$ generator matrix $G$ which spans the \textbf{codespace} $\ker(H)$, the set of right binary null vectors, or \textbf{codewords}, of $H$. $\ker(H)$ is a binary vector space of dimension $k\ge 0$ -- we say that $k$ is the number of \textbf{logical bits} stored in the code. Defining the \textbf{Hamming weight} of a binary vector $\mathbf{x} \in \mathbb{F}^n_2$ as the number of nonzero entries, and denoting this number with $\abs{\mathbf{x}}$, we say that the \textbf{code distance} $d$ is defined as
 \begin{equation}
    d = \min_{\mathbf{x}\in \mathcal{C}\backslash \mathbf{0}} \abs{\mathbf{x}} \, .
\end{equation}
The code parameters are often packaged using the bracketed notation $[n,k,d]$. By the rank-nullity theorem, the parity-check matrix $H$ will have $r$ linearly independent nontrivial left null vectors, where \begin{equation}
    r = m-n+k \ge 0.
\end{equation}
If $r=0$, corresponding to a full-rank parity-check matrix $H$, we say that the code is \textbf{non-redundant}, and the equation
\begin{equation}
    \mathbf{s} = H \mathbf{e}
\end{equation}
has $2^k$ solutions $\mathbf{e}$, which differ by logical codewords (alternatively, the solution is unique in $\mathbb{Z}_2^n/\ker(H)$).

Let $B$ be the vector space of all physical bits, with $|B|=n$, and let $S$ be the vector space of all parity checks with $|S|=m$. One can form a 3-term chain complex
\begin{equation} \label{eq:3-term chain}
    B \xlongrightarrow{H} S \xlongrightarrow{M} R \, ,
\end{equation}
where $R$ denotes the spaces of redundancies, and the boundary maps $H$ and $M$ are the parity-check matrix and the matrix which collects its left null vectors respectively. These left null vectors can be interpreted as parity constraints amongst the checks themselves, commonly referred to as metachecks. $R$ is defined as the vector space of metachecks. 

We say that $H$ is $(\Delta_B, \Delta_C)$ \textbf{LDPC (low-density parity-check)} if the maximum Hamming weights of its rows and columns are $\Delta_C, \Delta_B = O(1)$ respectively. In other words, every check acts on a \emph{constant} number of physical bits, and every bit only participates in a \emph{constant} number of checks. If both $k = \Theta(n)$ (constant rate) and $d = \Theta(n)$ (linear distance), then we say the code has asymptotically ``good" parameters. For most of this paper, our interest is on LDPC codes, although some of the results we derive below are not specific to this case. 

Good LDPC codes are often constructed from regular, bipartite expander graphs; such codes are commonly referred to as \emph{expander codes} \cite{Gallager_1962, Sipser_1996}. We now review some relevant results about expander codes. 

\begin{defn}[Bipartite expansion \cite{Sipser_1996}]\label{defn:bipartite expansion}
    Given a regular bipartite graph $\mathcal{G} = (B \cup C,E)$ with uniform left-right degrees $\Delta_B, \Delta_C = O(1)$, we say that $\mathcal{G}$ is $(\gamma,\delta)$ left-expanding if for any subset $B'\subset B$ with volume $\abs{B'} \leq \gamma n$, the size of its boundary $\partial B' \subset C$ obeys
    \begin{align}\label{eq:left expansion}
        \abs{\partial B'} \geq (1-\delta) \Delta_B \abs{B'} \, .
    \end{align}
    The definition of right-expansion follows analogously by swapping the roles of $B$ and $C$. If $\mathcal{G}$ is both left and right-expanding, then we simply say it is a $(\gamma,\delta)$ expander.
\end{defn}

Def. \ref{defn:bipartite expansion} is also referred to as (one-sided or two-sided) lossless expansion in the literature, with $\delta$ quantifying the amount of ``loss'' from maximum expansion. We can define a $(\Delta_B,\Delta_C)$-LDPC expander code $\mathcal{C}_\mathcal{G}$ by identifying the left-nodes $B$ with physical bits and the right-nodes $C$ with parity checks. The parity-check matrix $H$ is given by the bipartite adjacency matrix of $\mathcal{G}$. Using the above definition of a bipartite expander, one can immediately derive lower bounds on the code parameters of $\mathcal{C}$.

\begin{thm}[Constant rate \cite{Gallager_1962}]\label{thm:constant rate}
    If $\Delta_B < \Delta_C$, then the rate of $\mathcal{C}_\mathcal{G}$ is at least
    \begin{align}
        k/n \geq 1 - \Delta_B/\Delta_C \equiv \mathfrak{r}_{\rm des} \, ,
    \end{align}
    where $\mathfrak{r}_{\rm des} \equiv 1-\Delta_B/\Delta_C$ is called the \emph{design rate}.
\end{thm}
\begin{proof}
    Let $n = \abs{B}$ and $m = \abs{C}$ be the number of bits and checks respectively. The number of edges emanating out of all bits and checks must be equal for a bipartite graph, so we have that $n\Delta_B = m\Delta_C$. Since the rank of $H$ is upper-bounded by the number of rows $m$, by the rank-nullity theorem, we conclude that
    \begin{align}
        \frac{k}{n} \geq \frac{n-m}{n} = 1 - \frac{m}{n} = 1 - \frac{\Delta_B}{\Delta_C} = \mathfrak{r}_{\rm des} \, .
    \end{align}
\end{proof}

\begin{lem}[Unique neighbor expansion \cite{Sipser_1996}]\label{lem:unique neighbors}
    Suppose $G$ is left-expanding with $\delta < 1/2$. Then for any subset of bits $B' \subset B$ with size $\abs{B'} \leq \gamma n$, the number of unique neighbors $\partial_u B' \subset \partial B' \subset C$, checks which have only one edge connecting to $B'$, is at least
    \begin{align}
        \abs{\partial_u B'} \geq \Delta_B (1-2\delta) \abs{B'} \, .
    \end{align}
\end{lem}
\begin{proof}
    Since $\mathcal{G}$ is regular, the number of edges emanating out of $B'$ is $\Delta_B B'$. Since $\mathcal{G}$ is also left-expanding with $\delta<1/2$, from \eqref{eq:left expansion} we have $\abs{\partial B'} \geq (1-\delta) \Delta_B \abs{B'}$ for $\abs{B'} \leq \gamma n$. Each of the neighboring checks in $\partial B'$ must be connected by at least one edge from the definition of being a neighbor, so the number of remaining edges is $\Delta_B B' - (1-\delta) \Delta_B \abs{B'} = \delta \Delta_B \abs{B'}$. By the pigeon-hole principle, the number of neighboring checks connected by only one edge is at least $\abs{\partial_u B'} \geq (1-\delta) \Delta_B \abs{B'} - \delta \Delta_B \abs{B'} = \Delta_B (1-2\delta) \abs{B'} > 0$.
\end{proof}

We see that if $\mathcal{G}$ exhibits sufficient expansion $(\delta<1/2)$, then large subsets of bits will be connected to large amounts of unique checks. The above notion of unique neighbor expansion immediately gives us a confinement property \cite{Quintavalle_2021}, which loosely speaking, says that small errors produce small syndromes. We can quantify the confinement of expander codes as follows.

\begin{cor}[Linear confinement]\label{cor:confinement}
    Suppose $\mathcal{G}$ is left-expanding with $\delta < 1/2$. Then for any error $\mathbf{e}$ with weight $\abs{\mathbf{e}} \leq \gamma n$, the weight of its syndrome $\mathbf{s}(\mathbf{e}) = H \mathbf{e}$ is at least
    \begin{align}
        \abs{\mathbf{s}(\mathbf{e})} \geq \Delta_B (1-2\delta) \abs{\mathbf{e}} \, .
    \end{align}
\end{cor}

We note that the linear confinement property of expander codes has a physical interpretation in terms of macroscopic energy barriers between codeword states of the corresponding Hamiltonian. We can also use Lemma \ref{lem:unique neighbors} to derive a lower bound on the code distance.

\begin{thm}[Linear distance \cite{Sipser_1996}]\label{thm:linear distance}
    If $G$ is left-expanding with $\delta < 1/2$, then the code distance of $\mathcal{C}_\mathcal{G}$ is at least
    \begin{align}\label{eq:linear distance}
        d \geq 2\gamma (1-\delta) n \, .
    \end{align}
\end{thm}
\begin{proof}
    We will prove by contradiction that the weight of any nonzero codeword must be at least $2\gamma (1-\delta) n$. Suppose we have a codeword $\mathbf{x} \neq \mathbf{0}$ whose nonzero support is the set $\mathrm{supp}(\mathbf{x}) = B' \subset B$ with $\abs{B'} < 2\gamma(1-\delta)n$. If $\abs{B'} \leq \gamma n$ then Lemma $\ref{lem:unique neighbors}$ tells us that we have at least $\abs{\partial_u B'} \geq \Delta_B (1-2\delta) \abs{B'}$ unique neighbors, and hence unsatisfied checks, and so it is impossible for $\mathbf{x}$ to be a codeword. Now assume $\abs{B'} > \gamma n$. Partition $B'$ into a subset $S\subset B'$ of size $\abs{S} = \gamma n$ and its complement $\bar{S} \equiv B' \,\backslash\, S \subset B'$ of size $\abs{\bar{S}} < \gamma(1-2\delta) n$. Lemma \ref{lem:unique neighbors} tells us that $S$ has at least $\abs{\partial_u S} \geq \Delta_B (1-2\delta) \abs{S} = \gamma \Delta_B (1-2\delta) n$ unsatisfied checks. On the other hand, there are at most $\Delta_B \abs{\bar{S}} < \gamma\Delta_B(1-2\delta) n \leq \abs{\partial_u S}$ edges emanating out of the complement $\bar{S}$. So there are not enough edges in $\bar{S}$ to pair up all of the unsatisfied checks of $S$. As a consequence, $\mathbf{x}$ will violate at least one check, which contradicts the assumption that $\mathbf{x}$ is a codeword. Thus we require $\abs{\mathbf{x}} \geq 2\gamma(1-\delta) n$ for all nonzero codewords $\mathbf{x} \in \ker(H)$, and so by definition, the code distance satisfies \eqref{eq:linear distance}.
\end{proof}

Now that we have reviewed how bipartite expander graphs can lead to good LDPC codes, the problem boils down to constructing bipartite graphs with the required expansion properties. Explicit, algebraic constructions of lossless expanders are known to exist \cite{Capalbo_2002, Golowich_2023} with the required properties. On the other hand, it also turns out that random bipartite graphs are lossless expanders with high probability \cite{ModernCodingTheory}. For simplicity of analysis, we will focus on regular ensembles, where the target degree distributions are singular around $\Delta_B,\Delta_C$; we call such an ensemble the $(\Delta_B,\Delta_C)$-LDPC ensemble. The success of the random construction can be quantified by the following result.

\begin{thm}[Random expansion; Theorem 8.7 of \cite{ModernCodingTheory}]
\label{thm:random expansion}
    Let $\mathcal{G}$ be chosen uniformly at random from the $(\Delta_B, \Delta_C)$-LDPC ensemble with fixed $n$ and $\Delta_B > 1/\delta$. Let $\gamma_{\rm max}>0$ be the solution of the equation
    \begin{align}
        \frac{\Delta_B-1}{\Delta_B}h_2(\gamma_{\rm max}) - \frac{1}{\Delta_C}h_2\big(\Delta_C\gamma_{\rm max}(1-\delta)\big) - \Delta_C\gamma_{\rm max}(1-\delta)\, h_2\left( \frac{1}{\Delta_C(1-\delta)} \right) = 0 \, ,
    \end{align}
    where $h_2(x) \equiv -x\log x - (1-x)\log(1-x)$ is the binary entropy function. Then for $0<\gamma<\gamma_{\rm max}$ and $\alpha \equiv \Delta_B \delta - 1$, 
    \begin{align} \label{eq:randomexpansion}
        \mathbf{P} \left\{ \text{$\mathcal{G}$ \emph{is a} $(\gamma, \delta)$ \emph{left-expander}} \right\} = 1 - O(n^{-\alpha}) \, .
    \end{align}
    The analogy for right-expansion follows from switching the roles of $B \leftrightarrow C$.
\end{thm}

As a consequence, simply constructing a random parity-check matrix $H$ with target column and row weights $3 \leq \Delta_B < \Delta_C$ will result in a good LDPC code with probability 1 in the thermodynamic limit $n\rightarrow\infty$. We note that in practice, the code distances associated with these codes are often much higher than the theoretical guarantees. Theorem \ref{thm:constant rate} provides a lower bound on the rate of the code. The next result shows that the random construction produces codes whose rate is within $O(1/n)$ of the lower bound with high probability:

\begin{lem}[Rate approaches design rate; Lemma 3.27 of \cite{ModernCodingTheory}]\label{lem:rate vs design rate}
    Let $\mathcal{C}$ be chosen uniformly at random from the $(\Delta_B, \Delta_C)$-LDPC ensemble with fixed $n$ and $2 \leq \Delta_B < \Delta_C$. Let $\mathfrak{r}_{\rm des} = 1 - \Delta_B/\Delta_C$ denote the design rate. Then the actual rate $\mathfrak{r} \equiv k/n$ satisfies
    \begin{align}\label{eq:random nonredundancy}
        \mathbf{P} \left\{ n\mathfrak{r} = n\mathfrak{r}_{\rm des} + \nu \right\} = 1 - o_n(1) \, ,
    \end{align}
    where $\nu = 1$ if $\Delta_B$ is even and $\nu = 0$ otherwise.
\end{lem}
The extra constant $\nu=0,1$ comes from the fact that when $\Delta_B$ is even, there is an automatic redundancy coming from the constraint that the sum of all checks (mod 2) is zero. Lemma \ref{lem:rate vs design rate} asserts that all other checks are linearly independent with high probability.

An interesting property of non-redundant codes is that all $2^m = 2^{n-k}$ syndromes have a corresponding physical error up to the action of a codeword. To see this fact, first note that $H$ has full rank because all its rows are linearly independent. Because row and column ranks are equal, we can thus select $m=n-k$ linearly independent columns of $H$ to form a new, invertible $m \times m$ matrix $\tilde{H}$. For any syndrome $\mathbf{s}$, a corresponding error $\mathbf{e}(\mathbf{s})$ can be computed using $\mathbf{e}(\mathbf{s}) = \tilde{H}^{-1} \mathbf{s}$. As a consequence, there exists error vectors which correspond to single-flipped parity checks (excitations). Nonetheless, from the expansion properties of the underlying Tanner graph, we can show that such low-energy excitations necessarily have extended support. Intuitively, one can think of a low-energy excitation as a codeword of a modified code consisting of the original code but with the flipped check removed. Since removing a single check is unlikely to change the global expansion properties of the Tanner graph, the modified code will also have macroscopic distance, which implies that our low-energy excitation is far from all original codewords. We formalize this argument in the following Corollary.

\begin{cor}[Extended low-energy excitations]\label{cor:extended excitation}
    Suppose we have a $(\Delta_B,\Delta_C)$-LDPC code $\mathcal{C}$ whose Tanner graph $\mathcal{G}$ is left-expanding with $\delta<1/2$. Then for any nonzero syndrome $\mathbf{s} \neq \mathbf{0}$ with weight $\abs{\mathbf{s}} < \Delta_B (1-2\delta)$, its corresponding errors $\mathbf{e}(\mathbf{s})$ have weight at least
    \begin{align}\label{eq:extended excitation}
        \abs{\mathbf{e}(\mathbf{s})} \geq 2\gamma(1-\delta) n \, .
    \end{align}
    Moreover, if we have $\delta \leq 1/4$, then the above cutoff on the syndrome weight can be increased to $\abs{\mathbf{s}} < \Delta_B$.
\end{cor}
\begin{proof}
    We will prove \eqref{eq:extended excitation} by contradiction. Suppose we have a syndrome $\mathbf{s}$ with weight $\abs{\mathbf{s}} < \Delta_B (1-2\delta)$. First we will assume that the weights of its corresponding errors satisfy $\abs{\mathbf{e}} < \gamma n$. Since $\delta<1/2$, Corollary \ref{cor:confinement} tells us that the syndrome weight is lower-bounded by $\abs{\mathbf{s}} \geq \Delta_B(1-2\delta)\abs{\mathbf{e}}$. For a nonzero syndrome, we must have $\abs{\mathbf{e}} \geq 1$, and so we arrive at $\abs{\mathbf{s}} \geq \Delta_B(1-2\delta)$, a contradiction. Now assume $\gamma n \leq \abs{\mathbf{e}} < 2\gamma(1-\delta)n$. Following the linear distance proof of Theorem \ref{thm:linear distance}, we partition $\mathrm{supp}(\mathbf{e}) = S \sqcup \bar{S}$ with $\abs{S} = \gamma n$ and $\abs{\bar{S}} < \gamma(1-2\delta)n$. Lemma \ref{lem:unique neighbors} tells us that $S$ has at least $\abs{\partial_u S} \geq \Delta_B(1-2\delta)\abs{S} = \gamma\Delta_B(1-2\delta)n$ unique neighbors. On the other hand, we also know that $\bar{S}$ has at most $\Delta_B\abs{\bar{S}} \leq \gamma\Delta_B(1-2\delta)n - \Delta_B$ emanating edges. Consequently, the syndrome weight is at least $\abs{\mathbf{s}} \geq \Delta_B \geq \Delta_B(1-2\delta)$, a contradiction. Thus, for syndromes with weight $\abs{\mathbf{s}} < \Delta_B (1-2\delta)$, we must require \eqref{eq:extended excitation}. Finally, if $\delta < 1/4$, then we have $\abs{\mathbf{s}} \geq 2\Delta_B (1-2\delta) \geq \Delta_B$.
\end{proof}

As a consequence of Theorem \ref{thm:random expansion}, Lemma \ref{lem:rate vs design rate} and Cor. \ref{cor:extended excitation}, a random LDPC code with $\Delta_C > \Delta_B \geq 5$ will, with high probability, contain extended low-energy excitations associated with low-weight syndromes $\abs{\mathbf{s}} < \Delta_B$. Since Cor. \ref{cor:extended excitation} holds for \emph{all} valid error configurations $\mathbf{e}$ satisfying $\mathbf{s} = H\mathbf{e}$, we can also conclude that these low-energy excitations are necessarily far (in Hamming distance) from \emph{all} codewords. We hence refer to such low-energy states as \emph{non-decodable} states. By linearity, we can employ the confinement property (Corollary \ref{cor:confinement}) to say that the energy barriers for each of these minima are also macroscopic. As we will see in the next section, the lack of redundant checks in random LDPC codes also has crucial implications on the thermodynamics of their corresponding Hamiltonians.

\subsection{Thermodynamics and Kramers-Wannier duality}
\label{app:KW duality}

As was recently emphasized in \cite{Rakovszky:2023fng}, and as is well-known in the classical coding literature itself, there exists a \textbf{Kramers-Wannier-dual code} $\mathcal{C}_{\rm KW}$ for any classical linear code $\mathcal{C}$. The KW-dual code is found by looking at the co-chain complex associated with (\ref{eq:3-term chain}); its parity-check matrix is given by $M^\transpose$. 

Using the parity-check matrix $H$, we can define a Hamiltonian $\mathcal{H}$ acting on a system of $n$ spins $z_i = \pm 1$ by associating each row of $H$ as a multi-spin interaction:
\begin{align}\label{eq:Hamiltonian}
    \mathcal{H}(z) = - \sum_{c\in S} \prod_{i:H_{ci}=1} z_i \, .
\end{align}
The thermodynamic partition function of a code is then defined as 
\begin{equation}\label{eq:Z(beta)}
    Z_H(\beta) = \sum_{z \in \lbrace \pm 1 \rbrace^n} \mathrm{e}^{-\beta \mathcal{H}(z)}
\end{equation}
The subscript $H$ reminds us that the classical state space of the partition function is the output space of $H$: $\mathbb{F}_2^n$, the set of all physical bit strings -- this notation will shortly justify itself.

We now show that the thermodynamic partition function of the KW-dual code is identical (in terms of analytic behavior) to that of the original code, at a modified temperature.  The following theorem generalizes the classic \emph{self-duality} of the two-dimensional Ising model on a square lattice with periodic boundary conditions \cite{KW_duality}.  While we think it quite elegant in its own right, it also has crucial implications for the thermodynamics of LDPC codes.
\begin{thm}[Generalized Kramers-Wannier duality]\label{thm:Generalized KW}
    The thermodynamic partition function $Z_H(\beta)$ for a code and that of its KW-dual $Z_M(\beta)$ are related by
    \begin{align} \label{eq:KW}
        Z_H(\beta) &= 2^{n} \left(\cosh\beta\right)^m \mathrm{e}^{-\beta' r} \sum_{q\in\mathbb{Z}_2^R} \exp\Bigg[\beta^\prime \sum_{c\in S} \prod_{i:M^\transpose_{ci}=1} q_i \Bigg] = 2^{n} \left(\cosh\beta\right)^m \mathrm{e}^{-\beta' r} \times Z_M(\beta^\prime) \, ,
    \end{align}
    where \begin{equation}\label{eq:modified temperature}
        \mathrm{e}^{-2\beta^\prime} = \tanh \beta \, .
    \end{equation}
\end{thm}
\begin{proof}
Borrowing notation from \eqref{eq:3-term chain}, we write the corresponding partition function as
\begin{align}
    Z_H(\beta) = \sum_{z\in\mathbb{Z}_2^B} \exp\Bigg[\beta \sum_{c\in S} \prod_{i:H_{ci}=1} z_i\Bigg] = \sum_{z\in\mathbb{Z}_2^B} \sum_{s\in\mathbb{Z}_2^S} \prod_{c\in S} \mathrm{e}^{\beta s_c} \, \frac{1}{2} \Bigg( 1+s_c\prod_{i:H_{ci}=1} z_i \Bigg) \, ,
\end{align}
where $h_c = \mathbb{F}^n_2$ is the $c$th row of $H$, and $\beta$ is the usual inverse temperature. We now perform the sum over the physical spins $z_i = \pm 1$, noticing that terms only vanish if there is an even number of $z_i$ terms for all $i$.   The combinations of syndromes $s_c$ that survive are either (\emph{1}) $s_c=1$ or (\emph{2}) those which cannot be expressed in the form $\mathbf{s} = H\mathbf{e}$; namely, the set of all $s_c$ we wish to consider are precisely the  left null vectors of $H$.  Hence we may write
\begin{align}
    Z_H(\beta) &= 2^{n-m} \sum_{s\in\mathbb{Z}^S_2} \prod_{c\in S} e^{\beta s_c}  \times \prod_{\bm{\ell}\in M} \Bigg( 1 + \prod_{c:\bm{\ell}_c=1} s_c \Bigg) = 2^{n-m} \sum_{\mathbf{q}\in\mathbb{F}^R_2} \sum_{s\in\mathbb{Z}^S_2} \prod_{c\in S} e^{\beta s_c} \times \prod_{c:(M^\transpose \mathbf{q})_c=1} s_c  \notag \\
    &=  2^{n-m} \sum_{\mathbf{q}\in\mathbb{F}^R_2} \Bigg( \prod_{c:(M^\transpose \mathbf{q})_c=0} \sum_{s_c=\pm 1} e^{\beta s_c} \Bigg) \Bigg( \prod_{c:(M^\transpose \mathbf{q})_c=1} \sum_{s_c=\pm 1} s_c e^{\beta s_c} \Bigg) = 2^{n-m} \sum_{\mathbf{q}\in\mathbb{F}^R_2} \left( 2\cosh{\beta} \right)^{m-\abs{M^\transpose \mathbf{q}}} \left( 2\sinh{\beta} \right)^{\abs{M^\transpose \mathbf{q}}}  \notag \\
    &= 2^{n} \left( \cosh{\beta} \right)^m \sum_{\mathbf{q}\in\mathbb{F}^R_2} \left( \tanh{\beta} \right)^\abs{M^\transpose \mathbf{q}}  = 2^n \left(\cosh{\beta}\right)^m \sum_{q\in\mathbb{Z}^R_2} \exp\left[ -\frac{1}{2} \log\left( \frac{1}{\tanh{\beta}}\right) \sum_{c\in S} \Bigg( 1 - \prod_{i:M^\transpose_{ci}=1} q_i \Bigg) \right]
\end{align}
where in the second equality in the last line we switched from $\lbrace 0,1\rbrace$ (``$\mathbb{F}_2$-valued") to $\lbrace 1, -1\rbrace$ (``$\mathbb{Z}_2$-valued") coefficients in the sum over redundancies. Now defining $\beta^\prime$ using (\ref{eq:modified temperature}), we arrive at (\ref{eq:KW}). 
\end{proof}

We see that our original partition function is related to one where the ``spins" are the redundancies $q$, the ``Hamiltonian" is given by the (transposed) metacheck matrix $M^\transpose$, and the inverse temperature is given by \eqref{eq:modified temperature}. By far, the most important implications of this theorem for the present paper are the following two corollaries.
\begin{cor}\label{cor:no phase}
    Let $\mathcal{C}$ be a linear code chosen uniformly at random from the $(\Delta_B,\Delta_C)$-LDPC ensemble with $2 \leq \Delta_B < \Delta_C$, and let its corresponding Hamiltonian be given by \eqref{eq:Hamiltonian}.  As $n\rightarrow \infty$, there is almost surely no thermodynamic phase transition; namely, the free energy density
    \begin{equation}\label{eq:free energy density}
        f_H(\beta) = \lim_{n\rightarrow \infty}\left[-\frac{1}{n\beta}\log Z_H(\beta)\right]
    \end{equation}
    exists and is smooth for $0<\beta<\infty$. 
\end{cor}
\begin{proof}
    This result follows immediately from Lemma \ref{lem:rate vs design rate}, which asserts that as $n\rightarrow \infty$, the probability that a random LDPC code has either $r=0$ or $r=1$ redundant checks approaches 1.  As a consequence, for a random LDPC code, we use (\ref{eq:KW}) to find  that almost surely \begin{equation}
        f_H(\beta) = -\frac{1}{\beta} \left[\log 2 + \frac{\Delta_B}{\Delta_C} \log\left( \cosh \beta\right) \right] \equiv f_{r=0}(\beta) \, .
    \end{equation}
    Clearly this function is smooth in the desired domain.
\end{proof}

\begin{cor}[Extensive redundancy needed for phase transition \cite{Yoshida_2011, freedman2013, Weinstein_2019}]\label{cor:need O(n) redundancy}
    For Hamiltonians \eqref{eq:Hamiltonian} of any linear code, if there is a phase transition in $Z_H(\beta)$ as $n\rightarrow \infty$, then $r=\Omega(n)$.
\end{cor}
\begin{proof}
    We evaluate the free energy density \eqref{eq:free energy density}, using \eqref{eq:KW}, to obtain
    \begin{align}
        f_H(\beta) = f_{r=0}(\beta) -\frac{1}{\beta} \lim_{n\rightarrow\infty} \left[ \frac{r}{n} + \frac{1}{n} \log Z_M(\beta'(\beta)) \right] \, .
    \end{align}
    Any non-analytic contribution to $f_H(\beta)$ must come from $Z_M(\beta')$. Since the sum in $Z_M(\beta')$ \eqref{eq:KW} involves $2^r$ terms of bounded magnitude, and $\beta'$ is a smooth function of $\beta$ \eqref{eq:modified temperature}, the limit vanishes unless $r = \Omega(n)$. We conclude that $r=\Omega(n)$ is necessary for any phase transition.
\end{proof}

There is a physically intuitive picture for why a low or non-redundant code does not exhibit a phase transition. Consider the low-energy landscape of such a code. This landscape is dominated by basins surrounding not only codewords but also non-decodable low-energy states corresponding to the extended excitations from Cor. \ref{cor:extended excitation}. These non-decodable basins ``dilute" the Gibbs probability measure and prevent its condensation around codewords.  The abrupt condensation of the typical states in the ensemble near codewords is necessary to have a phase transition; indeed this is precisely the ferromagnetic transition of the 2D Ising model.  Without redundancy, the low-energy non-decodable (sick) states provably prevent such condensation -- there is a crossover upon lowering the temperature in which the Gibbs measure favors being increasingly close to codewords \emph{or} the sick states, until at $\beta=\infty$ we find only the true codewords.

\subsection{Lower bound on the classical mixing time}

We now prove that LDPC codes with sufficient underlying expansion, despite having no redundant checks and consequently no thermodynamic phase transition, exhibit exponentially slow dynamics under a local Gibbs sampler (e.g. Metropolis). The idea will be to use the confinement property of expander codes to show the existence of a macroscopic energy barrier between codewords. One can then show that the average time for a local Gibbs sampler to reach states on this energy barrier grows exponentially with the system size.

\begin{thm}[Slow classical Gibbs sampling]
\label{thm:exp mixing time classical}
    Let $\mathcal{C}$ be a $(\Delta_B,\Delta_C)$-LDPC code of length $n$ whose Tanner graph is a $(\gamma=\Omega(1),\delta<1/2)$ left-expander, and let $Z(\beta)$ be the corresponding thermal partition function according to \eqref{eq:Z(beta)}. Then there exists a critical inverse temperature
    \begin{align}
        \beta_c \leq \frac{1+\log\left(2/\gamma\right)}{2\Delta_B(1-2\delta)} \, ,
    \end{align}
    above which the average hitting time of a local Gibbs sampler to a state with $d/2$ errors obeys
    \begin{align}\label{eq:t_mix classical}
        t_{\rm hit} \geq \mathrm{e}^{\alpha n}
    \end{align}
    for some constant $\alpha(\beta)>0$ independent of $n$.
\end{thm}
\begin{proof}
    Let $H$ denote the parity-check matrix and $\mathcal{H}$ its corresponding Hamiltonian according to \eqref{eq:Hamiltonian}. Since $\delta<1/2$, Corollary \ref{cor:confinement} tells us that we have linear confinement of errors: for any error $\mathbf{e}\in\mathbb{F}^n_2$ such that $\abs{\mathbf{e}}\le \gamma n$, the weight of its syndrome $\abs{\mathbf{s}(\mathbf{e})} = |H\mathbf{e}|\ge \eta \abs{\mathbf{e}}$ where $\eta = \Delta_B(1-2\delta)$. Theorem \ref{thm:linear distance} also tells us that the Hamming distance of all nonzero codewords is at least $d \geq 2\gamma(1-\delta)n > \gamma n$. Choose the bottleneck set $E$ to be the cut set of states with $\abs{\mathbf{x}} = \gamma n/2$. Note that a local Gibbs sampler must traverse $E$ in order to reach a state with at least $d/2$ errors. Since $\mathcal{C}$ is a linear code, without loss of generality, we can start in the ground state corresponding to the $\mathbf{0}$ codeword.

    We begin by rewriting the code's Hamiltonian \eqref{eq:Hamiltonian} in the form
    \begin{align}\label{eq:F2 Hamiltonian}
        \mathcal{H}(\mathbf{x}) = -m + 2\abs{H\mathbf{x}}
    \end{align}
    for states $\mathbf{x} \in \mathbb{F}^n_2$. From the confinement property (Corollary \ref{cor:confinement}), we know that the energy landscape of \eqref{eq:F2 Hamiltonian} around $\mathbf{x}=\mathbf{0}$ is comprised of a deep potential well with an energy barrier of height at least $\mathcal{H}(\mathbf{x}\in E) \geq \gamma\eta n-m$. Now recall that any local Gibbs sampler must satisfy the detailed balance condition
    \begin{align}\label{eq:detailed balance}
        \frac{\mathbf{P}_t[\mathbf{x}_1 \rightarrow \mathbf{x}_2]}{\mathbf{P}_t[\mathbf{x}_2 \rightarrow \mathbf{x}_1]} = \frac{\pi[\mathbf{x}_2]}{\pi[\mathbf{x}_1]} = \exp \big[\beta\left( \mathcal{H}(\mathbf{x}_1) - \mathcal{H}(\mathbf{x}_2) \right)\big] \, ,
    \end{align}
    where $\pi[\mathbf{x}]$ is the equilibrium (Gibbs) probability of $\mathbf{x}$. Importantly, because the code distance obeys $d > \gamma n$ (Theorem \ref{thm:linear distance}), we know that the local sampler \emph{must} traverse the bottleneck set $E$ in order to reach the other ground states. We then bound the transition probability from our codeword state $\mathbf{0}$ to any state on $E$: after any time $t$ of Gibbs sampler evolution:
    \begin{align}
        \mathbf{P}_t[ \mathbf{0}\rightarrow E ] &= \sum_{\mathbf{x}\in E} \mathbf{P}_t[\mathbf{0}\rightarrow\mathbf{x}] \leq \sum_{\mathbf{x}\in E} \mathrm{e}^{-\beta [\mathcal{H}(\mathbf{x}) - \mathcal{H}(\mathbf{0})]}   \notag \\
        &\leq {n \choose \gamma n/2}\, \mathrm{e}^{-\beta\eta\gamma n} \leq \left( \frac{2\mathrm{e}n}{\gamma n} \right)^{\gamma n/2} \mathrm{e}^{-\beta\eta\gamma n}  \leq  \exp\left[ \frac{1}{2}\gamma n \left( 1+\log(2/\gamma) - 2\beta\eta \right) \right] \, ,
    \end{align}
    where we used the detailed balance condition \eqref{eq:detailed balance} in the second inequality and the fact that $\mathbf{P}_t(\mathbf{x}\rightarrow \mathbf{0}) \le 1$. Inverting the above expression for the transition probability gives us a lower bound on the average hitting time of the bottleneck set:
    \begin{align}\label{eq:t_hit bound long}
        t_{\rm hit}[\mathbf{0}\rightarrow E ] \geq \exp \left[ \frac{1}{2}\gamma n \left( 2\beta\eta - 1 - \log(2/\gamma) \right) \right] \, .
    \end{align}
    Thus, the average hitting time will be exponentially long in the system size as long as
    \begin{align}
        \beta > \frac{1+\log(2/\gamma)}{2\eta} \, ,
    \end{align}
    which completes the proof.
\end{proof}

By the Markov chain bottleneck theorem, our hitting time bound is also a lower bound for the mixing time. Combining Theorem \ref{thm:random expansion}, Corollary \ref{cor:need O(n) redundancy} and Theorem \ref{thm:exp mixing time classical}, and using the simple fact that if $\beta=0$ the Gibbs sampler is a random walk on the hypercube $\mathbb{F}_2^n$ (which does not mix slowly \cite{levin_markovchains}), we obtain the following Corollary.

\begin{cor}[Typical classical LDPC memory]
\label{cor:classicalphases}
If $\Delta_C>\Delta_B\geq 3$, then a code chosen uniformly at random from the $(\Delta_B,\Delta_C)$-LDPC ensemble at fixed $n$ has, with probability one (as $n\rightarrow \infty$): (\emph{1}) a \emph{dynamical phase transition} corresponding to a non-analyticity in $\log t_{\mathrm{mix}}(\beta)$, the mixing time of the Metropolis Gibbs sampler,  as a function of $\beta$; (\emph{2}) no thermodynamic phase transition (non-analyticity in $\log Z(\beta)$). 
\end{cor}

The physical interpretation of Corollary \ref{cor:classicalphases} is rather interesting: a typical classical LDPC code gives rise to a problem which is thermodynamically trivial (no nonzero temperature phase transitions) and yet exhibits dynamical \emph{ergodicity breaking} -- the Gibbs sampler gets stuck near codewords (or near a non-decodable low-energy state) for exponentially long times in system size.

We also emphasize that Corollary \ref{cor:classicalphases} implies that the thermal Gibbs sampler can be justifiably called a \emph{thermal/passive memory}.  While it is true that a \emph{typical} system will not be found in a codeword state after finite time due to a nonzero density of small errors, it is the case that one can use standard decoders such as bit flip \cite{Sipser_1996} or message passing \cite{BP_decoding} to decode the LDPC code at any time before the hitting time of Theorem \ref{thm:exp mixing time classical} and obtain the exact ground state (codeword). From the error correction perspective therefore, these thermal systems serve as excellent memories of the stored information, which can be efficiently decoded.


\section{Quantum LDPC codes}
\label{app:quantum LDPC}

We now turn our attention to the study of quantum LDPC codes.  To the extent possible, our presentation and results will closely mirror the classical story above.  A nice physics perspective on such codes can be found in \cite{Rakovszky:2023fng,Rakovszky:2024iks}.

\subsection{CSS codes and the hypergraph product}
\label{app:HGP codes}

A quantum stabilizer code is a $2^K$-dimensional subspace of $N$ qubits which is the simultaneous $+1$ eigenspace of a commuting set of Pauli operators, a generalization of parity checks to qubits \cite{gottesman1997}. A Calderbank-Shor-Steane (CSS) code is a particular instance of a stabilizer code in which the checks are strictly $X$-type or $Z$-type Paulis \cite{Calderbank_1996, Steane_1996}. CSS codes can be defined by two classical binary linear codes $\mathcal{C}_X$ and $\mathcal{C}_Z$ whose parity-check matrices satisfy the orthogonality condition $H^{}_X H^\transpose_Z = 0$, corresponding to the commutativity of the checks. This orthogonality condition induces a 5-term chain complex
\begin{align}\label{eq:5-term chain}
    R_X \xlongrightarrow{M^\transpose_X} S_X \xlongrightarrow{H^\transpose_X} Q \xlongrightarrow{H_Z} S_Z \xlongrightarrow{M_Z} R_Z \, ,
\end{align}
where $Q$ is the space of physical qubits, and $S$ and $R$ are the spaces of syndromes and redundancies respectively for the subscript Pauli type. This 5-term chain complex was previously used in the context of single-shot error correction \cite{Campbell_2019}. For CSS codes, the generalized Kramers-Wannier duality from Sec. \ref{app:KW duality} becomes trivial from a thermodynamics standpoint because reversing the direction of the arrows in \eqref{eq:5-term chain} simply switches the roles of $X$ and $Z$. Due to error digitization in stabilizer codes, decoding for CSS codes can be done independently on $\mathcal{C}_X$ (for $Z$ errors) and $\mathcal{C}_Z$ (for $X$ errors). We say a CSS code is LDPC if and only if both $\mathcal{C}_X$ and $\mathcal{C}_Z$ are LDPC codes.

The easiness of obtaining good classical LDPC codes ironically causes a considerable challenge for their quantum counterparts. Because a random classical LDPC code has a large distance with high probability (recall Section \ref{app:classical LDPC}), choosing a random LDPC code for $\mathcal{C}_X$ will typically result in any orthogonal $\mathcal{C}_Z$ to have large weight, and thus be non-LDPC.

The hypergraph product (HGP) was the first quantum LDPC construction to achieve codes with constant rate $K = \Theta(N)$ while maintaining a polynomial scaling of the code distance $D = \Theta(\sqrt{N})$ \cite{HGP}. The construction takes in two arbitrary classical input codes, described by the 3-term chain \eqref{eq:3-term chain}, and performs a homological tensor product to form a 5-term chain, from which the corresponding CSS code can be identified using \eqref{eq:5-term chain}. Specifically, the $X$ and $Z$-type parity-check matrices are given by
\begin{subequations}
\begin{align}
    H_X &= \left( \begin{array}{c|c} H_1 \otimes \ident_{n_2}\; & \;\ident_{m_1} \otimes H^\transpose_2 \end{array} \right) \\
    H_Z &= \left( \begin{array}{c|c} \ident_{n_1} \otimes H_2\; & \; H^\transpose_1 \otimes \ident_{m_2} \end{array} \right) \, ,
\end{align}
\end{subequations}
from which one can straightforwardly verify $H^{}_X H^\transpose_Z = 2H^{}_1 \otimes H^\transpose_2 = 0$ over $\mathbb{F}_2$. Geometrically, the Tanner graph of the HGP code resembles a Euclidean graph product of those of its input classical codes (see Fig. 2 of the main text). The quantum code parameters are given by
\begin{subequations}
\begin{align}
    N &= n_1n_2 + m_1m_2 \\
    K &= k^{}_1 k^{}_2 + k^\transpose_1 k^\transpose_2 \\
    D &= \min\left( d^{}_1, d^{}_2, d^\transpose_1, d^\transpose_2 \right)  \label{eq:HGP distance} \, ,
\end{align}
\end{subequations}
where the transpose superscript denotes the parameters of the respective transpose codes. Logical $X$ ($Z$) operators resemble codewords of the input classical codes traversing in the horizontal (vertical) direction.

The relative simplicity of the HGP construction leads to the quantum code inheriting many properties of its classical input codes. For ease of notation, we use the same input code (twice) to form the HGP code. If the classical input code is a $(\Delta_B,\Delta_C)$ LDPC code, then the associated HGP code is $(\max(\Delta_B,\Delta_C), \Delta_B+\Delta_C)$ LDPC: each physical qubit participates in at most $\max(\Delta_B,\Delta_C)$ $X$-checks and likewise $Z$-checks, and each check acts on at most $\Delta_B+\Delta_C$ physical qubits. If the input classical code is an expander code, then the resulting HGP code is called a \emph{quantum expander code} \cite{Leverrier_2015}. We define the $(\Delta_B,\Delta_C)$-HGP ensemble as the collection of CSS codes produced from the hypergraph product of the $(\Delta_B,\Delta_C)$-LDPC ensemble.

\begin{cor}[Quantum code distance]\label{cor:HGP distance}
    Suppose $\mathcal{C}$ is a $(\Delta_B,\Delta_C)$-LDPC code whose underlying Tanner graph $\mathcal{G}$ is a $(\gamma,\delta)$ expander with $\gamma=\Omega(1)$ and $\delta<1/2$. Then the quantum code $\mathcal{Q}_\mathcal{C}$, defined as the hypergraph product of $\mathcal{C}$ with itself, has a minimum distance
    \begin{align}\label{eq:HGP expander distance}
        D \geq 2\gamma (1-\delta) \times \min(n,m) = \Theta\left(\sqrt{N}\right) \, .
    \end{align}
\end{cor}
\begin{proof}
    Since $\mathcal{G}$ is left-expanding with $\delta<1/2$, Theorem \ref{thm:linear distance} tells us that the minimum distance of $\mathcal{C}$ is $d \geq 2\gamma(1-\delta)n$. At the same time, $\mathcal{G}$ is also right-expanding with $\delta<1/2$, so we have $d^\transpose \geq 2\gamma(1-\delta)m$ by swapping the roles of bits and checks. \eqref{eq:HGP expander distance} then follows from \eqref{eq:HGP distance}.
\end{proof}

When analyzing operator weights (for e.g. confinement) in CSS codes, we need to be a bit careful because of quantum degeneracy: operators differing by stabilizer elements act equivalently on the codespace. In other words, for any operator $O$, we can construct its corresponding orbit $O + \mathcal{C}^{\perp}_X + \mathcal{C}^{\perp}_Z$ under the stabilizer group. As a consequence, we can arbitrarily grow the weight of an error by attaching stabilizers without affecting its syndrome. In order to discount such possibilities in our counting, we define equivalence classes of errors and focus on minimal-weight elements for each class:

\begin{defn}[Reduced error]
    For a given $X$ error $\mathbf{e}_X \in \mathbb{F}^{N}_2$, define its \emph{reduced weight} as
    \begin{align}
        \abs{\mathbf{e}_X}_{\rm red} = \min_{\mathbf{e}'\in (\mathbf{e}_X+\mathcal{C}^\perp_X)} \abs{\mathbf{e}'} \, ,
    \end{align}
    where $\mathbf{e}_X+\mathcal{C}^\perp_X$ is an equivalence class of errors for $\mathbf{e}_X$. The analogy for $Z$ errors follows by switching $X\leftrightarrow Z$. An equivalent error that achieves the reduced weight is called a \emph{reduced error}.
\end{defn}

Importantly, in order for an error to effect an undesired logical operation, its reduced weight must grow to the code distance. For ease of notation, we henceforth drop the Pauli subscript on errors $\mathbf{e}$ and implicitly assume $X$- or $Z$-type. The notion of confinement has also been extended to HGP codes, which we recite below.

\begin{cor}[HGP confinement; Corollary 9 of \cite{Leverrier_2015}]
\label{cor:HGP confinement}
    Suppose $\mathcal{C}$ is a $(\Delta_B,\Delta_C)$-LDPC code whose underlying Tanner graph $\mathcal{G}$ is a $(\gamma,\delta)$ expander with $\gamma=\Omega(1)$ and $\delta<1/6$. Let $\mathcal{Q}_\mathcal{C}$ be the quantum code defined as the hypergraph product of $\mathcal{C}$ with itself. Then any error $\mathbf{e}$ with reduced weight $\abs{\mathbf{e}}_{\rm red} < \min(\gamma n,\gamma m)$ has a corresponding syndrome $\mathbf{s}(\mathbf{e})$ with weight at least
    \begin{align}
        \abs{\mathbf{s}(\mathbf{e})} \geq \frac{1}{3} \abs{\mathbf{e}}_{\rm red} \, .
    \end{align}
\end{cor}

We emphasize that the condition on the expansion ($\delta<1/6$) of the parent classical code, where HGP confinement has been proven, is stricter than that for classical LDPC codes ($\delta<1/2$; Corollary \ref{cor:confinement}).

\subsection{Thermodynamics} 
We first obtain the following quantum generalization of Corollary \ref{cor:no phase}: like in the classical setting, the hypergraph product of a randomly chosen classical LDPC code with itself gives a quantum code with trivial thermodynamics:

\begin{thm} \label{thm:no thermal phase quantum}
    Let $\mathcal{G}$ be chosen uniformly at random from the ($\Delta_B,\Delta_C)$-LDPC ensemble of expander graphs, with $\Delta_C > \Delta_B \geq 2$ and $\Delta_B$ odd, and let $\mathcal{C} \equiv \mathcal{C}_\mathcal{G}$ be the associated classical LDPC code.  Let $\mathcal{Q}_\mathcal{C}$ be the CSS code formed out of the hypergraph product of $\mathcal{C}$ with itself.  Then with probability (\ref{eq:random nonredundancy}) approaching 1 as $n\rightarrow \infty$, the thermodynamic partition function \begin{equation}
        Z(\beta) = \mathrm{tr}\left(\mathrm{e}^{-\beta H}\right) = 2^{k^2} (2\cosh \beta)^{2n(n-k)}. \label{eq:quantumZ}
    \end{equation}
\end{thm}
\begin{proof}
    Lemma \ref{lem:rate vs design rate} tells us that with high probability (\ref{eq:random nonredundancy}), the classical code will be non-redundant.  Our goal is to show that this non-redundancy neatly carries to the quantum CSS code, leading to trivial thermodynamics.

    To prove this fact, it is sufficient to show that up to the action of each logical operator, we may uniquely specify a state in Hilbert space by the simultaneous eigenvalue of every stabilizer check (chosen independently). This fact in turn follows the non-redundancy of the code.  After all, the only linearly independent (in the $\mathbb{F}_2$-sense) Pauli strings are those corresponding to logical operations, which necessarily commute with $\mathcal{H}$ and imply degeneracy of the spectrum.  Thus, we can easily perform the trace in $Z(\beta)$ in this stabilizer eigenbasis.  Since the eigenvalues of all check operators may be chosen independently due to the non-redundancy of the code, (\ref{eq:quantumZ}) follows from the fact that a non-redundant HGP has $k^2$ logical bits, and that there are $n(n-k)$ $X$-type checks and $Z$-type checks.
\end{proof}

\subsection{Lower bound on the quantum mixing time}
\label{app:quantum thermal decoding}

We now prove the analogue of Theorem \ref{thm:exp mixing time classical} for HGP product codes whose parent LDPC codes exhibit sufficient underlying expansion. Unlike the scenario for classical LDPC codes, we need to be more careful in our analysis of HGP codes since the energy barrier is only $\mathrm{O}(\sqrt{N})$. Random states drawn from the Gibbs ensemble at any nonzero temperature will have $\mathrm{O}(N)$ errors on average, but the errors will typically form disconnected clusters and rarely conspire in an adversarial manner. Our strategy will be to employ a Peierls-like argument to show that adversarial clusters of errors, below a cutoff given by the expansion, are thermodynamically unfavorable at sufficiently low temperatures. The probability of encountering one of these adversarial clusters then provides a lower bound on the mixing time.

We first need to precisely define what we mean by an error cluster. We say that two qubits are connected if and only if they share at least one check. This notion of connectivity can be captured by the adjacency graph of a code.

\begin{defn}[Adjacency graph] \label{def:adjacency graph}
    The adjacency graph $\mathcal{G}_A$ of a length-$N$ code $\mathcal{C}$ is defined as a simple graph of $N$ vertices where an edge connects two vertices if and only if the corresponding bits in $\mathcal{C}$ share a parity check.
\end{defn}

The adjacency graph can be obtained by treating the parity-check matrix as a (hyper)edge-vertex incidence matrix. This graph has been used in the past to prove locality bounds in both classical and quantum codes \cite{Bravyi_2009, Bravyi_2010, Baspin_2022}. Importantly, irreducible logical operators, in the sense that they cannot be decomposed into products of smaller ones, must have support on connected subgraphs of $\mathcal{G}_A$. We now define the notion of a $(s,\ell)$-cluster, introduced in \cite{Kovalev_2013}, which will characterize the adversarial error clusters that could fool an eventual maximum-likelihood decoder.

\begin{defn}[$(s,\ell)$-cluster] \label{def:(s,l) cluster}
    A $(s,\ell)$-cluster is a connected subgraph of $\mathcal{G}_A$ containing $s$ vertices, on which there are $\ell$ errors.
\end{defn}

We now introduce a definition of \textbf{non-decodable state}, which is a state that \emph{could possibly} trick a maximum-likelihood decoder in the worst case.   In particular, the following definition will ensure that if we start with a codeword, any logical error necessarily will have passed through a non-decodable state.

\begin{defn}[Non-decodable state] \label{def:non-decodable state}
    A non-decodable state is a state with at least one $(s,\ell)$-cluster with size $s\geq D$ and $\ell \geq s/2$ errors. 
\end{defn}

We are now ready to bound the mixing time of the Gibbs sampler. The idea will be to use the linear confinement property (Corollary \ref{cor:HGP confinement}) to show that the rate of producing large non-decodable clusters is exponentially suppressed in the cluster size below a critical temperature.

\begin{thm}[Slow quantum Gibbs sampling]
\label{thm:exp mixing time quantum}
    Let $\mathcal{C}$ be a $(\Delta_B,\Delta_C)$-LDPC code of length $n$ whose Tanner graph is a $(\gamma=\Omega(1),\delta<1/6)$ expander, and for simplicity assume $\Delta_B<\Delta_C$ and $n>m = \Theta(n)$. Let $\mathcal{Q}_\mathcal{C}$ be a quantum CSS code of length $N=\Theta(n^2)$ defined as the hypergraph product of $\mathcal{C}$ with itself with associated parity-check matrices $H_X$ and $H_Z$.  Let 
    \begin{equation}
        z \equiv \mathrm{e} \Delta_C \left(\Delta_B+\Delta_C-1 \right) \, ,
    \end{equation}
    where $\mathrm{e}$ is the usual base of the natural logarithm. Then as $N\rightarrow \infty$, there exists a critical inverse temperature 
    \begin{align}
        \beta_{\rm c} \leq 3\log(2z) 
    \end{align}
    such that for all $\beta > \beta_{\mathrm{c}}$, the mixing time of Gibbs sampling is bounded by
    \begin{align}\label{eq:t_mix quantum}
        t_{\rm mix} \geq \mathrm{e}^{\alpha \sqrt{N}}
    \end{align}
    for some constant $\alpha(\beta)>0$ independent of $N$.
\end{thm}
\begin{proof}
    Since $\mathcal{Q}_C$ is a CSS code, we focus on $X$ errors since the analysis for $Z$ errors follows analogously.  Without loss of generality, we bound the time it takes to degrade information stored near codeword $\mathbf{0}$.
    
    Let $\mathcal{G}_A$ be the adjacency graph associated with $H_Z$ according to Def. \ref{def:adjacency graph}. Recall that $H_Z$ is $(\Delta_C,\Delta_B+\Delta_C)$-LDPC, and so $\mathcal{G}_A$ has maximum degree $\Delta_C(\Delta_B+\Delta_C-1)$. Since $\delta<1/2$, Corollary \ref{cor:HGP distance} tells us that the quantum code distance obeys $D\geq 2\gamma(1-\delta)m$. Furthermore, since $\delta<1/6$, Corollary \ref{cor:HGP confinement} tells us that all errors with reduced weight $\abs{\mathbf{e}}_{\rm red} \leq \gamma m$ have syndrome weight $\abs{\mathbf{s}(\mathbf{e})} \geq \frac{1}{3}\abs{\mathbf{e}}_{\rm red}$. We emphasize that disconnected error clusters on $\mathcal{G}_A$ have independent confinement.
    
    Let's determine what we should choose for our bottleneck set $E$. From Def. \ref{def:non-decodable state}, we know that a non-decodable state must contain at least one $(s\geq D,\ell\geq s/2)$-cluster. We accordingly choose our bottleneck set $E$ be the cut set of states that contain exactly one $(s=\gamma m, \ell\geq fs)$-cluster and no larger, with $f \le \frac{1}{2}(1-2/s)$. Note that any local trajectory through state space to a non-decodable state must pass through some $\mathbf{x}_t \in E$ at some time $t$. The ``no larger" condition also automatically implies that our cluster has an error-free perimeter, since a $(s,\ell)$-cluster with an error on its perimeter can always be enlarged to a $(s+1,\ell+1)$-cluster.

    \begin{figure*}[t]
    \centering
    \includegraphics[width=0.85\textwidth]{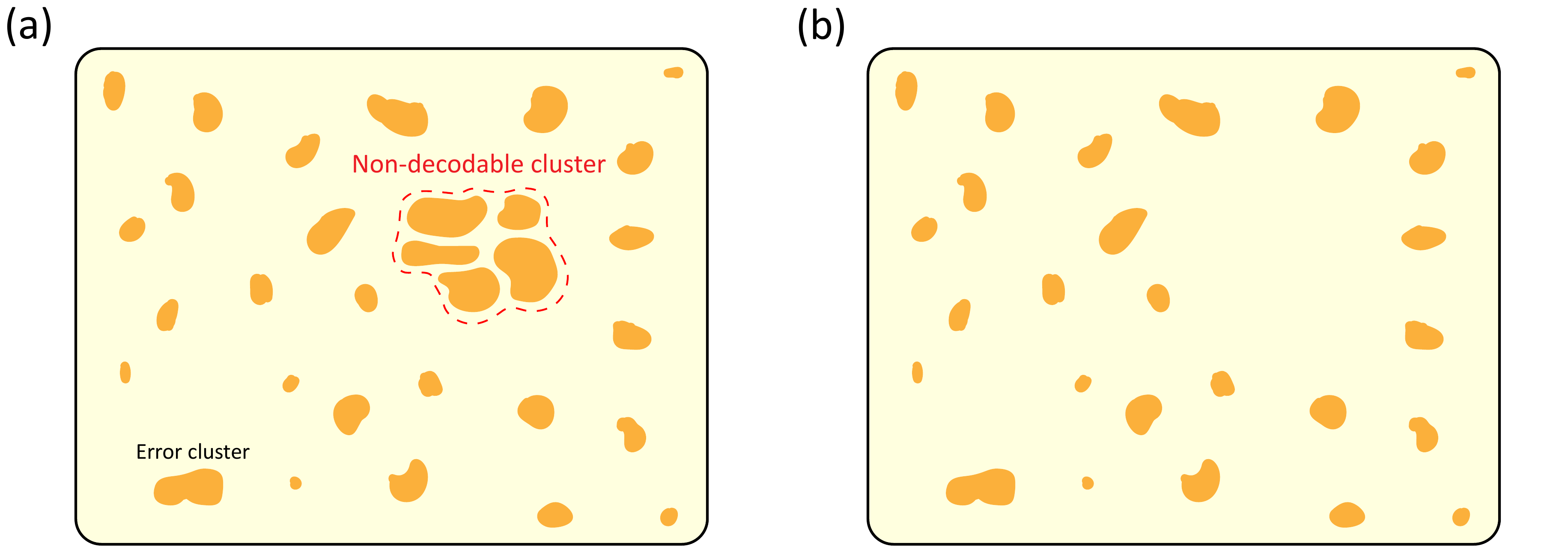}
    \caption{(a) An example of a bottleneck state $\mathbf{x} \in E$ is depicted with the non-decodable cluster highlighted. (b) The partner state $\mathbf{y}(\mathbf{x}) \in A$ with this cluster removed. The Gibbs sampler will take a very long time to reach states containing non-decodable clusters at low temperature.}
    \label{fig:bad cluster}
    \end{figure*}

    Now that we know that the Gibbs sampler must traverse $E$, we can proceed to bound the hitting time to $E$. Starting in a suitably chosen ensemble of initial decodable states close to codeword $\mathbf{0}$, which are drawn from a set (introduced below) $A$. The idea will be to show that the transition probability from $A$ to $E$ is exponentially small, which ensures that our Gibbs sampler is unlikely to stray far (in terms of decodability) from the initial codeword state.
    
    We now define our decodable set $A$.  Suppose we have a state $\mathbf{x} \in E$. There will, by construction, exist a decodable partner state $\mathbf{y}(\mathbf{x})$ consisting of $\mathbf{x}$ but with the $(s,\ell)$-cluster removed, see Fig. \ref{fig:bad cluster}.  As such, we deduce that $\mathbf{y}(\mathbf{x})$ is a decodable state close to codeword $\mathbf{0}$.  We define the set \begin{equation}
        A \equiv \lbrace \mathbf{y}^\prime \in \mathbb{F}_2^N \; : \; \exists \; \mathbf{x}\in E\text{ such that }\mathbf{y}^\prime = \mathbf{y}(\mathbf{x})\rbrace.
    \end{equation}
    Notice that all states in $A$ are, by definition, decodable since $\mathbf{x}\in E$ (by definition of $E$) had \emph{exactly one} dangerous cluster.   Let us now upper bound the relative degeneracy $g(E,A) \equiv |E|/|A|$ by estimating the number of possible clusters. Let $z = \mathrm{e}\Delta_C(\Delta_B+\Delta_C-1)$. For any of the $N$ sites, it can participate in no more than $z^{\gamma m}$ connected subgraphs of size $\gamma m$ \cite{Aliferis_2008}. There can also be no more than $2^{\gamma m}$ configurations of errors in these subgraphs: each site could have an error or not. As a result, we have $g(E,A) \leq N\times (2z)^{\gamma m}$. 

    By our definition of $\mathcal{G}_A$, any term in $\mathcal{H}(\mathbf{x})$ only couples $x_i$ and $x_j$ if $i\sim j$ are connected by an edge in $\mathcal{G}_A$.  Therefore, by construction, we deduce that
    \begin{equation}
    \mathcal{H}(\mathbf{x}) = \mathcal{H}(\mathbf{y}(\mathbf{x})) + \mathcal{H}_{\text{cluster}} \, ,    \label{eq:clusterdiff}
    \end{equation}
    where $\mathcal{H}_{\mathrm{cluster}}$ can be bounded using linear confinement:
    \begin{align}\label{eq:clusterbound}
        \mathcal{H}_{\text{cluster}} \geq \frac{2}{3}f\gamma m \, .
    \end{align}   
    We can accordingly bound the equilibrium probability ratio of $E$ and $A$ as
    \begin{align} \label{eq:P(E)/P(A)}
    \frac{\pi[E]}{\pi[A]} &\leq \frac{1}{\pi[A]}\sum_{\mathbf{x}\in E} \pi[\mathbf{x}] \le \frac{1}{\pi[A]}\sum_{\mathbf{x}\in E} \pi[\mathbf{y}(\mathbf{x})]\, \mathrm{e}^{-\frac{2}{3}\beta f\gamma m} \notag \\
    &\leq g(E,A) \times \mathrm{e}^{-\frac{2}{3}\beta f\gamma m}  \notag \\
        &\leq N\times \exp\left[ -\gamma m\left( \frac{2}{3}\beta f - \log(2z) \right) \right] \, ,
    \end{align}
    where in the first line we used (\ref{eq:clusterdiff}) and (\ref{eq:clusterbound}); in the second line we noted that the sum over all $\mathbf{y}(\mathbf{x})$ can count each element of $A$ at most $g(E,A)$ times, and in the third line we used our bound on $g(E,A)$. 
    We can now bound the transition probability from $A$ to $E$ at any time $t$ in terms of a ratio of their equilibrium probabilities:
    \begin{align}
        \mathbf{P}_t[A\rightarrow E] &\equiv \sum_{\mathbf{y}\in A} \mathbf{P}_t[\mathbf{y}\rightarrow E] \frac{\pi[\mathbf{y}]}{\pi[A]} 
        = \sum_{\mathbf{x}\in E} \sum_{\mathbf{y}\in A} \mathbf{P}_t[\mathbf{y}\rightarrow\mathbf{x}]\, \frac{\pi[\mathbf{y}]}{\pi[A]}  \notag \\
        &= \frac{1}{\pi[A]} \sum_{\mathbf{x}\in E} \sum_{\mathbf{y}\in A} \mathbf{P}_t[\mathbf{x}\rightarrow\mathbf{y}]\, \pi[\mathbf{x}]  = \frac{1}{\pi[A]} \sum_{\mathbf{x}\in E} \mathbf{P}_t[\mathbf{x}\rightarrow A]\, \pi[\mathbf{x}]  \notag \\
        &\leq \frac{\pi[E]}{\pi[A]} \leq N\times \exp\left[ -\gamma m\left( \frac{2}{3}\beta f - \log(2z) \right) \right] \, ,
    \end{align}
    where in the second line we used detailed balance \eqref{eq:detailed balance} and summed over states in $A$; in the third line we bounded $\mathbf{P}_t[\mathbf{x}\rightarrow A]\le 1$ and used (\ref{eq:P(E)/P(A)}).
    Inverting the above transition probability gives us a lower bound on the expected hitting time:
    \begin{align}\label{eq:t_hit quantum}
        t_{\rm hit}[A\rightarrow E] \geq N^{-1} \times \exp\left[ \gamma m\left( \frac{2}{3}\beta f - \log(2z) \right) \right] \, .
    \end{align}
    Thus, the hitting time to a non-decodable state will be exponentially long in the polynomial system size $m=\Theta\big(\sqrt{N}\big)$ as long as
    \begin{align}
        \beta > \frac{3}{2f}\log(2z) > 3\left( 1+\frac{2}{\gamma m-2} \right) \log(2z) \, .
    \end{align}
    Crucially, none of the above steps required counting any (reduced) error clusters of size greater than $\gamma m$, and so we never left the linear confinement regime. By the Markov chain bottleneck theorem, we conclude that \eqref{eq:t_hit quantum} is a lower bound on the mixing time if $\beta>\beta_{\mathrm{c}}$ in the thermodynamic limit, wherein $m,n\rightarrow\infty$.
\end{proof}

The Peierls argument for the 2D Ising model demonstrates that, below a critical temperature, domain walls are locally unfavorable and hence are globally suppressed. In the same spirit, we have shown that local thermal fluctuations below a certain temperature under the HGP Hamiltonian are unlikely to produce adversarial error clusters which would cause eventual maximum-likelihood decoding to fail. However, unlike the full Peierls argument, which would examine adversarial clusters of all sizes, we only focus on clusters of sizes near the code distance. For local Gibbs sampling, this restriction gives rise to a dynamical bottleneck which prevents ergodicity. Since the dynamics of $X$ and $Z$ errors are equivalent, the thermal Gibbs sampler achieves a \emph{thermal/passive quantum memory} -- encoded quantum information survives for a time which is exponentially long in the (polynomial) system size.

Combining Theorem \ref{thm:random expansion}, Theorem \ref{thm:no thermal phase quantum} and Theorem \ref{thm:exp mixing time quantum}, we obtain the following Corollary.

\begin{cor}[Typical quantum LDPC memory] \label{cor:typical quantum memory}
    If $\Delta_C>\Delta_B\geq 7$, then the quantum CSS code produced from the hypergraph product of a classical code chosen uniformly from the $(\Delta_B,\Delta_C)$-LDPC ensemble at fixed $n$ has, with probability one (as $n\rightarrow \infty$): (\emph{1}) a \emph{dynamical phase transition} corresponding to a non-analyticity in $\log t_{\mathrm{mix}}(\beta)$, the mixing time of the Metropolis Gibbs sampler,  as a function of $\beta$; (\emph{2}) no thermodynamic phase transition (non-analyticity in $\log Z(\beta)$). 
\end{cor}

Lastly, we note that Theorem \ref{thm:exp mixing time quantum} can be straightforwardly adapted to the classical LDPC setting to obtain an alternative bound on the classical critical temperature, as formalized in the following Corollary.

\begin{cor}[Alternative classical mixing time bound to Theorem \ref{thm:exp mixing time classical}] \label{cor:classical mixing bound v2}
    Let $\mathcal{C}$ be a $(\Delta_B,\Delta_C)$-LDPC code of length $n$ whose Tanner graph is a $(\gamma=\Omega(1),\delta<1/2)$ left-expander. Let
    \begin{align}
        z \equiv \mathrm{e} \Delta_B \left(\Delta_C-1\right) \, ,
    \end{align}
    where $\mathrm{e}$ is the usual base of the natural logarithm. Then as $n\rightarrow\infty$, there exists a critical inverse temperature
    \begin{align}
        \beta_{\rm c} \leq \frac{\log(2z)}{\Delta_B(1-2\delta)}
    \end{align}
    such that for all $\beta>\beta_{\rm c}$, the mixing time of Gibbs sampling is bounded by
    \begin{align}
        t_{\rm mix} \geq \mathrm{e}^{\alpha n}
    \end{align}
    for some constant $\alpha(\beta)>0$ independent of $n$.
\end{cor}

\section{Numerical simulations}\label{app:numerics}

\begin{figure*}[t]
\centering
\includegraphics[width=0.325\textwidth]{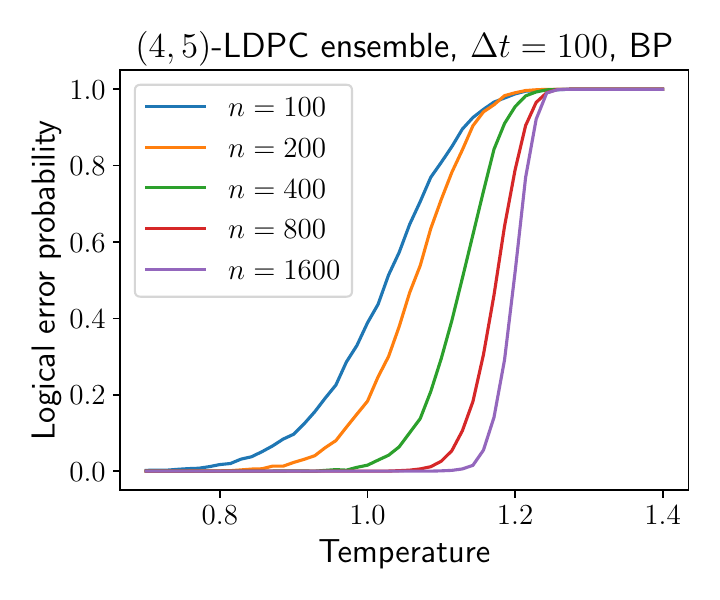}
\includegraphics[width=0.325\textwidth]{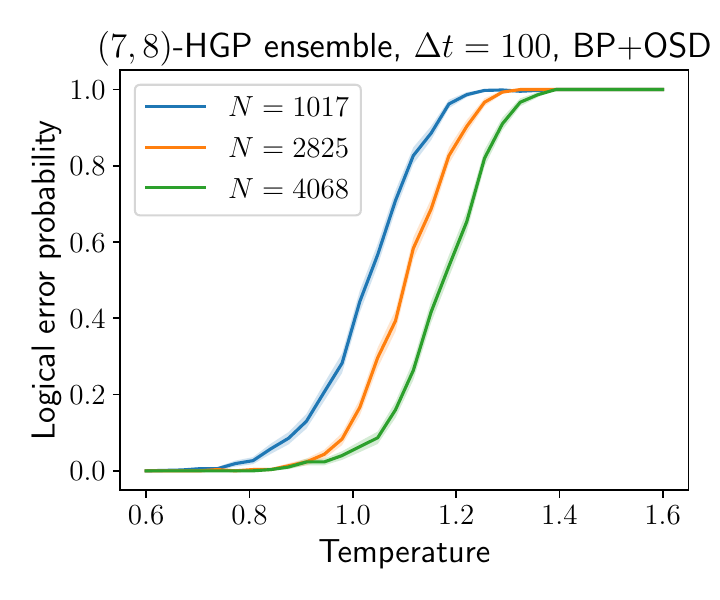}
\includegraphics[width=0.325\textwidth]{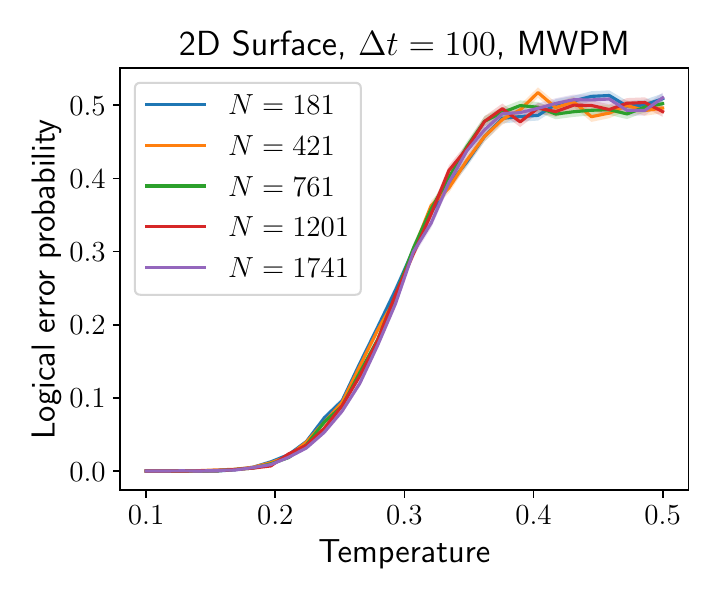}
\\
\includegraphics[width=0.325\textwidth]{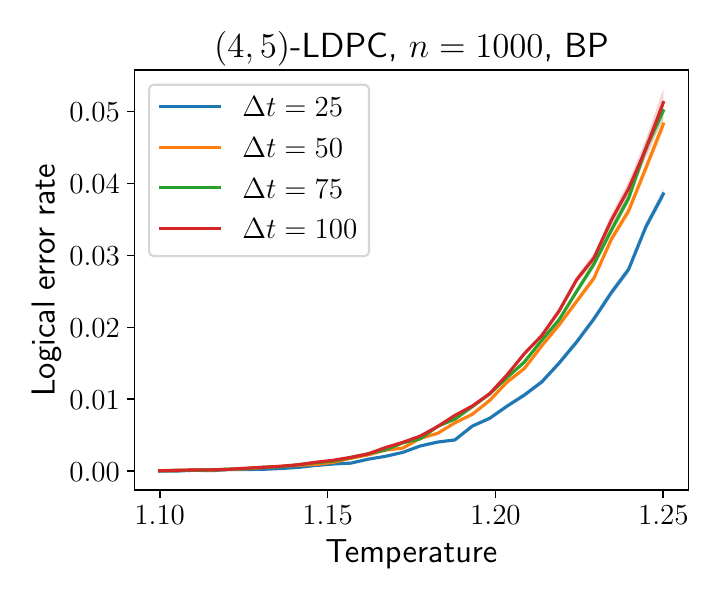}
\includegraphics[width=0.325\textwidth]{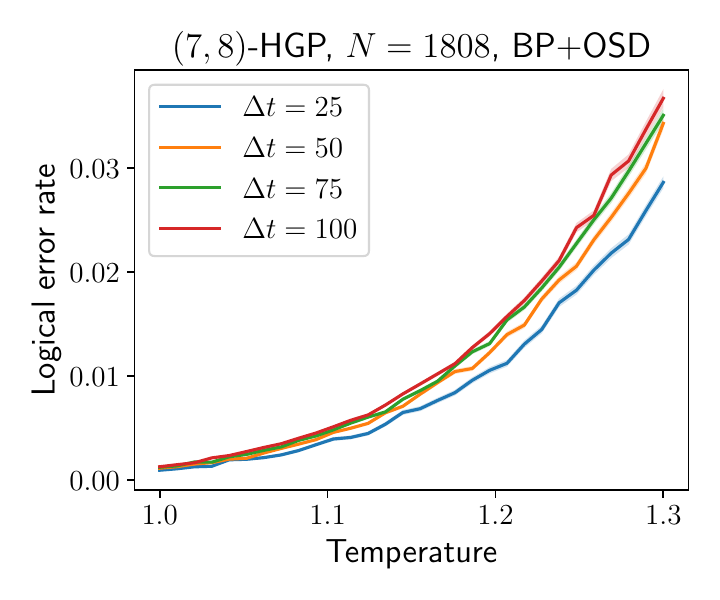}
\includegraphics[width=0.325\textwidth]{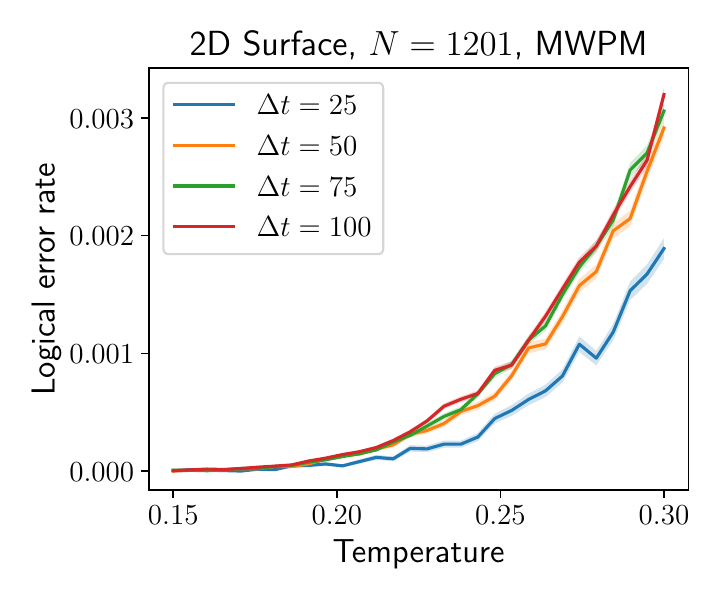}
\caption{\textbf{Top:} The logical block error probability as a function of temperature is plotted. We set $\Delta t = 100$ and sample several codes from the $(4,5)$-LDPC ensemble, $(7,8)$-HGP ensemble and 2D surface code family. \textbf{Bottom:} The logical block error rate (normalized by $\Delta t$) as a function of temperature is plotted near the observed transition point. We examine one particular code from each of the aforementioned ensembles and vary the equilibration time. Error bands (shaded) denote the standard error.}
\label{fig:numerics}
\end{figure*}

In this section, we conduct numerical simulations of passive error correction on both the classical and quantum LDPC codes previously discussed. For the classical codes, we randomly sample from the $(4,5)$-LDPC ensemble using the configuration model: we initialize $n$ bits and $\frac{5}{4}n$ checks and pair up their emanating edges according to a random permutation. For the quantum codes, we randomly sample from the $(7,8)$-HGP ensemble by constructing $(7,8)$-LDPC codes using the configuration model and taking the hypergraph product to produce $(8,15)$-LDPC quantum CSS codes. Since the $X$ and $Z$ sectors of the HGP are identical, we focus only on decoding of $X$ errors. All of our sampled codes have no check redundancies. In both classical and quantum simulations, we begin in the $\mathbf{0}$ codeword. We then simulate thermal equilibration at a chosen temperature by performing $\Delta t=100$ Monte Carlo sweeps of the Metropolis algorithm, where each sweep updates all sites exactly once. After the Monte Carlo sweeps, we feed the error syndrome into a belief-propagation decoder \cite{BP_decoding} and verify if the output codeword is $\mathbf{0}$. For HGP codes, the presence of small loops in their Tanner graphs prevents the convergence of belief propagation (BP), and so we also perform ordered-statistics post-processing \cite{Panteleev_2021}. We use open source software to implement belief propagation \cite{Roffe_LDPC_Python_tools_2022} and ordered-statistics decoding (BP+OSD) \cite{roffe_decoding_2020}. In particular, we perform a maximum of 10 iterations of minimum-sum BP with the additional use of OSD-0 for quantum decoding. We sweep from low to high temperatures and terminate the simulation if we fail to decode all Monte Carlo samples at a given temperature. For comparison, we also simulate thermal equilibration for the 2D surface code with minimum-weight perfect-matching (MWPM) decoding, using the open-source software PyMatching \cite{pymatching}.

The simulation results are showcased in Fig. \ref{fig:numerics}. We observe evidence of dynamical transitions for both classical and quantum LDPC codes from (i) the suppression of logical errors with increasing system size and (ii) the asymptotic behavior of the logical error rate with increasing equilibration time. For the classical $(4,5)$-LDPC ensemble, we observe this transition near $T_{\mathrm{c}} = 1/\beta_{\mathrm{c}} \approx 1.2$, consistent with our lower bounds of $T_{\rm c} \ge 0.17$ from Theorem \ref{thm:exp mixing time classical} and $T_{\rm c} \ge 0.29$ from Corollary \ref{cor:classical mixing bound v2}. For the quantum $(7,8)$-HGP ensemble, we observe a transition around $T \approx 1.1$, consistent with our lower bound of $T_{\rm c} \ge 0.052$ from Theorem \ref{thm:exp mixing time quantum}. In contrast, for the 2D surface code, we observe that the logical error probability is independent of system size. Indeed we should expect such behavior since at any nonzero temperature there is a finite density of anyons, which can diffuse freely to form an extensive error string in constant time \cite{Brown_2016}.


\section{Codes built from tree graphs}\label{app:tree codes}

\subsection{Repetition code on a classical tree}

In this subsection, we review why the mixing times of classical Ising models on trees are exponentially long in the levels of the trees.  On a general graph $G=(V,E)$, the Ising model Hamiltonian is \begin{equation}
    \mathcal{H} = -\sum_{\lbrace i,j\rbrace \in E} Z_i Z_j.
\end{equation}
Using (\ref{eq:parity check}) in the main text, one can convert this to a parity-check matrix $H$ with row weight $\Delta_C=2$.  It is straightforward to see that on a connected graph this is simply a repetition code, with all bits in the same state.  If the graph has no loops, then there are no redundant parity checks, and Theorem \ref{thm:no thermal phase quantum} implies that the system is thermodynamically trivial.   

Despite this, we first show that there is nontrivial dynamics, at least on a certain tree: \begin{thm}\label{thm:classical tree memory}
Consider a $r$-level tree with $n=1+3+\cdots + 3^r = \frac{1}{2}(3^{r+1}-1)$ vertices, where each interior vertex has exactly 3 children.  Then for $\beta$ larger than O(1) constant, the mixing time $t_{\mathrm{mix}} \ge \exp[\mathrm{O}(r)] = n^{\mathrm{\Omega}(1)}$ under Glauber dynamics.   Hence this system has thermal self-correction. 
\end{thm}
This result was proved rigorously in \cite{berger2005glauber}, but we will find it helpful for our later discussion to present a slightly modified proof which follows our proof of Theorem \ref{thm:exp mixing time quantum} more closely. 
\begin{proof}
Given a state of the spins $\mathbf{z}=\pm 1$ on the tree, define the recursive majority function $g: \mathbb{Z}^n_2 \rightarrow \mathbb{Z}^n_2$ as follows. For a given state, we label all the boundary vertices by the value of their spins. Next, we label their parent vertices $w$ with whatever state the majority of the children of $w$ are in. We inductively label all the interior vertices from the labels of their children and the value of the recursive majority function $g$ is the label of the root. 

Let $S_+$ and $S_-$ be the set of states with $g = 1$ and $-1$. We choose the bottleneck set $\partial S_-\subset S_+$ as states that are one flip away from $S_-$. Since the function $g$ only depends on boundary vertices, for each state in $\partial S_-$, we can move the state to $S_-$ by flipping one vertex $v$ on the boundary (note that $v$ may not be unique). If we connect any such $v$ and the root with a path, for each interior vertex on the path, their label is +1 and there is one and only one child with label -1. We call such a path the \emph{fault path} and denote it by $\varphi$. 

For a given configuration $\mathbf{z}$, we denote the label of any vertex $v$ by $m_v = g_v(\mathbf{z})$. 
%
Let $F^-_\varphi$ be the set of states in $\partial S_-$ with a specific fault path $\varphi$ and $W^-_{\varphi}$ be the set of states in $F^-_\varphi$ where all the vertices on the fault path have the same label as their spin, +1. We first study only the configurations in $W^-_{\varphi}$.  Letting $\pi \propto \exp[-\beta \mathcal{H}]$ denote the Gibbs probability distribution, we have:

\begin{figure}[t]
    \centering
    \includegraphics[width=0.7\textwidth]{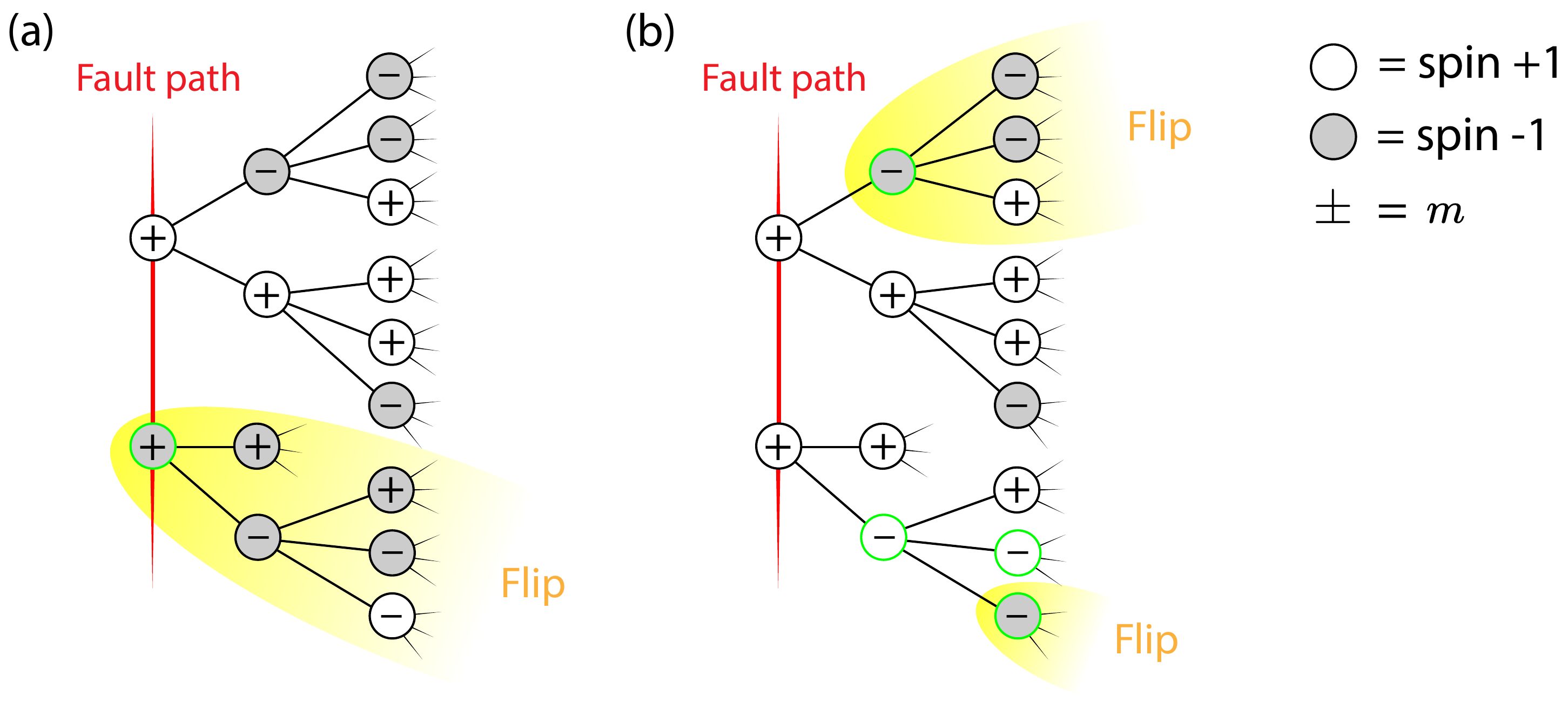}
    \caption{The two steps of the mapping between $F^-_{\varphi} \in S_-$ and its complement $\bar{F}^-_{\varphi} \in S_+$ are depicted. (a) The first step involves flipping the branch of any -1 spin (shaded gray) on the fault path, thereby mapping $F^-_{\varphi}\rightarrow W^-_{\varphi}$. (b) The second step involves the inductive strategy (green circles) described in Lemma \ref{lem:F_phi_0}.}
    \label{fig:classical tree decoding}
\end{figure}

\begin{lem}\label{lem:F_phi_0}
   If $\beta>\log 2$, the equilibrium probability of states in $W^-_{\varphi}$ is exponentially small in $r$:
    \begin{equation}
        \pi[W^-_{\varphi}] < \mathrm{e}^{-2(r-1)(\beta-\log 2)}. \label{eq:lemS42}
    \end{equation}
\end{lem}

\begin{proof}
 The strategy will follow closely the proof of Theorem \ref{thm:exp mixing time quantum}: we aim to define a ``partner" set $\bar{W}^-_{\varphi} \subset S_+$, chosen so that we can compare the equilibrium probability of a state $\mathbf{z} \in W^-_{\varphi}$ as
     \begin{align}\label{eq:z in F- bound}
         \frac{\pi[\mathbf{z}\in W^-_{\varphi}]}{\pi[\bar{\mathbf{z}} \in \bar{W}^-_{\varphi}]}  = \mathrm{e}^{-\beta[E(\mathbf{z}) - E(\bar{\mathbf{z}})]} \, , 
     \end{align}
     where $\bar{\mathbf{z}} \in \bar{W}^-_{\varphi}$ is the ``partner" of $\mathbf{z}$.  Intuitively, the partner should have significantly fewer domain walls, so that the ratio in \eqref{eq:z in F- bound} is as small as possible; also, not too many $\mathbf{z}$ should share a partner $\bar{\mathbf{z}}$, so that we can eventually bound $\pi(W^-_{\varphi}) / \pi(\bar{W}^-_{\varphi})$. Define a mapping $W^-_{\varphi} \rightarrow \bar{W}^-_{\varphi}$ as follows: for each vertex $v$ on the fault path, we inductively apply the following strategy to $v$, its children, grandchildren, etc., until we reach the boundary.   For a given vertex, we check whether its children whose labels are $m=-1$ also have spin $Z=-1$. For every such child with both label and spin being -1, we flip the whole branch starting from the child. For every child with label -1 but spin 1, we apply the above strategy to its child: see Fig.~\ref{fig:classical tree decoding}(b). After the mapping, the energy will always be lower. Note that for states $\mathbf{z} \in  W^-_{\varphi}$, partnering is not a one-to-one mapping: multiple $\mathbf{z}$ can share the same $\bar{\mathbf{z}}$.
     
     As a simple example, consider states where all the children of vertices on the fault path with label -1 also have spin -1: label this special subset $Y^-_{\varphi_0}\subset W^-_{\varphi}$.  To construct the partner subset $\bar Y^-_{\varphi_0}\subset \bar W^-_{\varphi}$, we will flip all the branches starting from these children, and the energy is lowered by $2(r-1)$. By construction, for any vertex $v$ on the fault path, 2 of the children of $v$ have $m=+1$ and 1 has $m=-1$.   Given a known fault path $\varphi_0$, there are then at most $2^{r-1}$ possible $\mathbf{z}$ that can map to a given $\bar{\mathbf{z}}$, where $2^{r-1}$ simply counts the possible configurations of the -1 children in the tree: $v$ has 3 children, one of which lies along the known fault path and therefore necessarily has $m=1$).   Using \eqref{eq:z in F- bound} we deduce \begin{equation}
         \pi[Y^-_{\varphi_0}] \le 2^{r-1}\pi[\bar Y^-_{\varphi_0}] \mathrm{e}^{-2(r-1)\beta} \le \mathrm{e}^{-(2\beta - \log 2) (r-1)},
     \end{equation}
     showing that once $\beta > \frac{1}{2}\log 2$, $\pi[Y^-_{\varphi_0}]$ is very small for thermodynamically large trees.

    We now must consider all the possible ways to build $\bar{\mathbf{z}}$ for more general fault paths, weighted by the Boltzmann penalty factor as in (\ref{eq:z in F- bound}). For a specific vertex $v$ on the fault path, if its child with label -1 has spin -1, the energy will be lowered by 2 by flipping that child and all of its grandchildren.  Following the same steps as above, we deduce a relative Boltzmann weight of \begin{equation}\label{eq:relativeweightonfaultpath}
       \sum_{\mathbf{z} \text{ sent to }\bar{\mathbf{z}}} \frac{\pi[\mathbf{z}\in W^-_{\varphi}]}{\pi[\bar{\mathbf{z}}\in \bar F_{\varphi_0}^-]} \le 2c_1(\beta) \times [\text{relative weight from the remaining vertices on fault path}]
    \end{equation} 
    Next, we consider the relative weight coming from vertices on the fault path where the child with label -1 has spin 1 but all the grandchildren with label -1 have spin -1.  Again, there are 2 choices of children that could have label $-1$, but now we know that either 2 or 3 of the relevant grandchildren have spin -1 and label -1.  We will flip all of these children.  The relative weight coming from these grandchildren flips then \begin{equation}
        2 c_2(\beta) - 2c_1(\beta) := 2 \left(3\mathrm{e}^{-4\beta} + \mathrm{e}^{-6\beta}\right),
    \end{equation}
    with the first term coming from 2 grandchildren with spin -1, and the second term coming from 3.  We've defined $c_2(\beta)$ such that we can now adjust (\ref{eq:relativeweightonfaultpath}) by replacing $c_1$ with $c_2$.   We can continue to iterate this argument for a given vertex: at the next level, we find \begin{equation}
        2c_{n+1}(\beta) - 2c_1(\beta) = 2\left(3 c_n(\beta)^2 + c_n(\beta)^3\right).
    \end{equation}
     If $\beta \geq 2\log 2$, it can be checked that $c_2(\beta) < 2\mathrm{e}^{-2\beta} <1$ and $c_3(\beta) < \mathrm{e}^{-2\beta} + 4c_2(\beta)^2 \leq 2\mathrm{e}^{-2\beta} $; repeating this argument we deduce that for any $n$, \begin{equation}
         c_n(\beta) < 2\mathrm{e}^{-2\beta}.
     \end{equation}Now we see that we can just replace $c_1$ with $2\mathrm{e}^{-2\beta}$ in (\ref{eq:relativeweightonfaultpath}), which leads to \eqref{eq:lemS42}, after summing over $\bar{\mathbf{z}}$ and using $\pi[\bar W^-_{\varphi}] < 1$.
\end{proof}

Now we return to the more general case:

\begin{lem}\label{lem:F^-_phi}
   If $\beta \ge \frac{3}{2}\log 2$, the equilibrium probability of states in $F^-_{\varphi}$ is exponentially small in $r$:
    \begin{equation}
        \pi[F^-_{\varphi}] < \mathrm{e}^{-(2\beta-3\log 2)(r-1)}. \label{eq:lemS43}
    \end{equation}
\end{lem}

\begin{proof}
     For states in $F^-_\varphi \backslash W^-_{\varphi}$, there are some interior vertices $v$ on the fault path with spin -1. For each such $v$, consider the operation $\tilde g$ wherein we flip $v$ and all the branches starting from the children of $v$ that are not on the fault path, see Fig.~\ref{fig:classical tree decoding}(a). In this way, only parity checks that are on the fault path can be changed. After the flipping, there are no violated parity checks on the fault path, while the parity of all the checks off the fault path is unchanged by construction.  Hence, $\tilde g$ always lowers the energy of the state. 
     
     Since each state on the fault path has either $z=1$ or $z=-1$, we conclude that at most $2^{r-1}-1$ states in $F^-_\varphi \backslash W^-_{\varphi}$ that can be mapped to the same state in $W^-_{\varphi}$ by $\tilde g$.  Notice that $\tilde g$ does not change the label of any vertex on the fault path: because each flip induces by $\tilde g$ flips 2 whole branches of children, one of which came with label +1 and the other with label -1, the flip simply exchanges the $\pm 1$ label children, while leaving the label of the fault path vertex unchanged. Therefore, $\pi[F^-_\varphi] < 2^{r-1} \pi[W^-_{\varphi}]$, and employing \eqref{eq:lemS42} we obtain \eqref{eq:lemS43}.
\end{proof}   
We now complete the proof of the theorem.  For $\partial S_-$, there are $3^{r-1}$ different fault paths. Therefore, 
\begin{equation} \label{eq: bottle neck C}
    \pi[\partial S_-] < 3^{r-1} \pi[F^-_{\varphi}] < \mathrm{e}^{-(2\beta-\log 24)(r-1)}
\end{equation}
and the mixing time 
\begin{equation}
    t_{\text{mix}} \geq \frac{\pi[S_-]}{4\pi[\partial S_-]} > \frac{1}{8} \mathrm{e}^{-(2\beta-\log 24)(r-1)}.
\end{equation}
We remind the reader that a proof for more general trees is in \cite{berger2005glauber}.
\end{proof}

\subsection{Hypergraph product codes that are not good quantum memories}

For the classical trees, the path $\varphi$ provides a lower bound of the energy barrier that is linear in $r \sim \log n$.  We saw that this was, happily, not an impediment to a classical self-correcting memory.  However, crucially, the proof of Theorem \ref{thm:classical tree memory} was more complicated than Theorem \ref{thm:exp mixing time classical} or Theorem \ref{thm:exp mixing time quantum} because we had to be very careful about the possibility of large entropic contributions to the Boltzmann weight of finding clusters of any given size.  Put another way, we had to be very explicit to rule out the possibility of an entropy-driven breakdown of thermal memory, even in the presence of the $\mathrm{\Theta}(\log n)$ energy barrier: we did so by showing that a logical error must pass through rare configurations where, at a minimum, there is one flipped parity check at each layer of the tree.

\begin{figure}[t]
    \centering
    \includegraphics[width=0.65\textwidth]{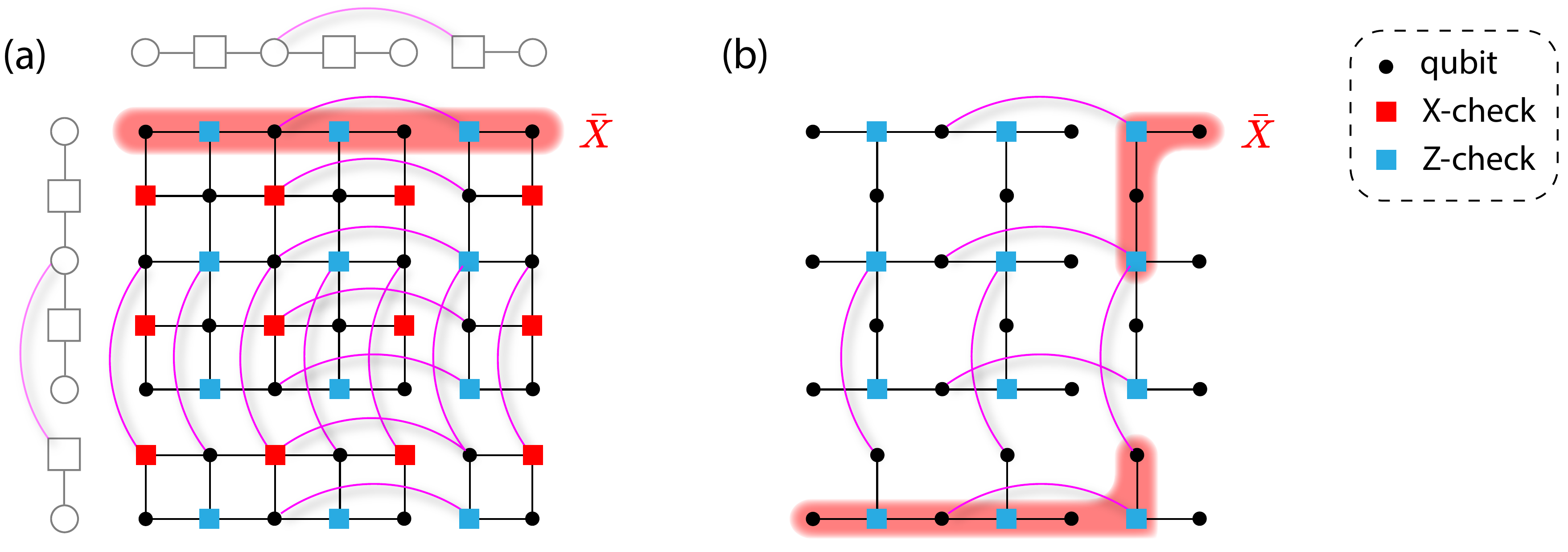}
    \caption{(a) The hypergraph product code of a small degree-3 tree is illustrated. Black dots, blue squares, and red squares represent qubits, $Z$ checks, and $X$ checks respectively. The support of a logical $\bar{X}$ operator is denoted by the qubits inside the shaded red rectangle. (b) A logical $\bar{X}$ operator equivalent to the previous one is drawn.}
    \label{fig:tree HGP}
\end{figure}

We now consider the quantum LDPC code obtained from the hypergraph product of the $r$-level classical code with itself.  It will still have the same energy barrier, but the probability of the bottleneck set may no longer be assured to be exponentially small in $r$. As shown in Fig. \ref{fig:tree HGP}, for any given logical operator, by appending appropriate stabilizers that are adjacent to it, we can move the corresponding branch to another row without increasing the energy barrier of the logical operator. In terms of the bottleneck set, we still have $\pi[F^-_\varphi] \leq \mathrm{e}^{-\mathrm{O}(r)}$, but now the entropic factor is increased by $\lesssim n^{\mathrm{O}(r)}$, where $n$ is the system size of the classical code.   Indeed, naively in Fig. \ref{fig:tree HGP}, we can freely ``slide" the support of the logical $\bar{X}$ operator (which runs from left to right) along vertical columns, because the ``check-check"-type qubits of the hypergraph product only connect to two $Z$ stabilizers. Suppose we have one of these adversarial error strings that traverse vertically across check-check qubits. If one of its flipped checks is too close to the leaf of the tree, then our large error may be stabilizer-equivalent to smaller errors supported on different rows/trees. However, if the violated parity check is $\ell$ generations removed from the leaves of the tree, it appears to us that there is no such stabilizer reduction to remove a very dangerous entropic factor of at least $3^\ell$; if we slide beyond $3^\ell$ rows then we again can reduce the weight of the error using stabilizers. We were therefore unable to prove the existence of a bottleneck set whose probability is exponentially small, and of course our above argument suggests in fact that it will not exist -- namely, that entropic effects will always dominate. Hence we conjecture that this tree-hypergraph product code is not a passive, self-correcting memory in contrast to the corresponding classical code.

\begin{figure*}[t]
\centering
\includegraphics[width=0.325\textwidth]{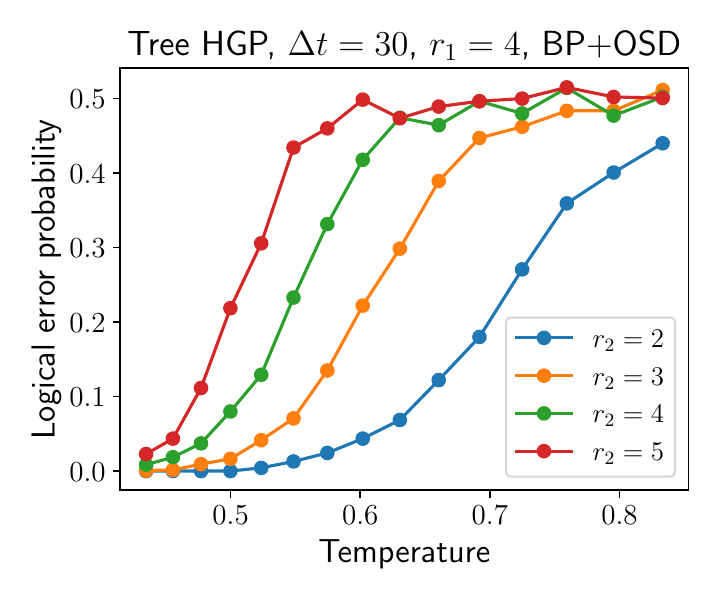}
\includegraphics[width=0.325\textwidth]{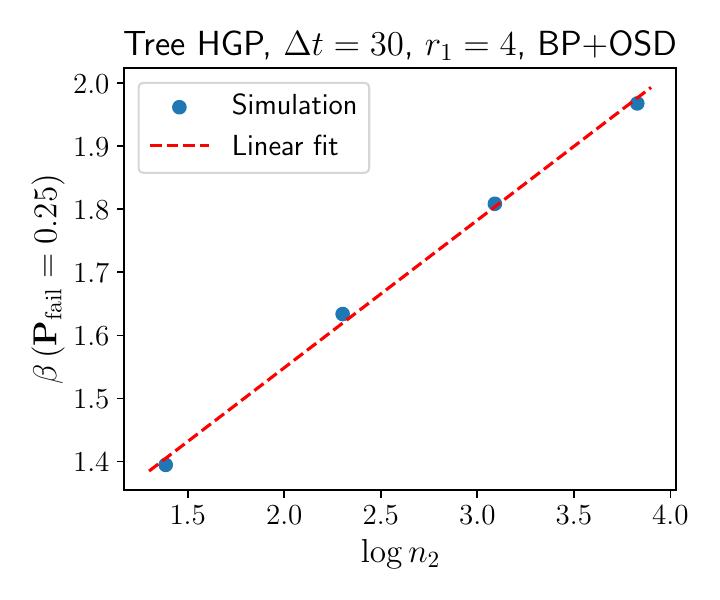}
\includegraphics[width=0.325\textwidth]{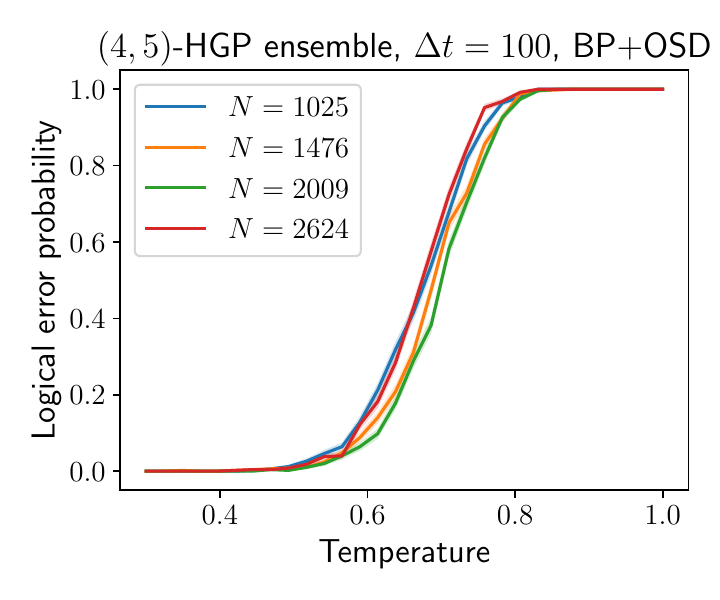}
\caption{\textbf{Left:} The logical error probability as a function of temperature is plotted for a hypergraph product code consisting of two classical Ising models on 3-regular trees with levels $r_1$ (fixed) and $r_2$ (variable). \textbf{Middle:} Fixing the logical error probability at 0.25, we interpolate the data on the left plot to estimate the critical inverse temperature as a function of the size of the second classical tree. A linear fit is also overlayed. \textbf{Right:} The logical error probability as a function of temperature is plotted for several codes from the $(4,5)$-HGP ensemble.}
\label{fig:tree numerics}
\end{figure*}

We perform numerical simulations of passive error correction for hypergraph product codes of 3-regular trees described above in order to establish confidence for the conjecture; see Fig. \ref{fig:tree numerics}. To obtain numerical evidence for the entropic breakdown of memory that we argued for above, we fix the size of one tree in the hypergraph product code while varying the number of levels in the other tree, thus making the hypergraph product asymmetric. Explicitly, the $\llbracket N,K,D \rrbracket$ parameters of the tree-hypergraph product code are given by
\begin{subequations}
\begin{align}
    N &= n_1 n_2 + (n_1-1)(n_2-1) = O\big( \mathrm{e}^{r_1+r_2} \big)  \\
    K &= 1  \\
    D &= \min\left(n_1,n_2\right) \, .
\end{align}
\end{subequations}
The scaling of the critical temperature with the size of the second tree displays a logarithmic dependence which is consistent with free energy arguments: if $c_1 \sim n\mathrm{e}^{-2\beta}$, for example, we would find that \begin{equation}
    \beta_{\mathrm{c}}(n) = \beta_{\mathrm{c},0} + c \log n.
\end{equation}
To fix the value of $c$, we note that adding long error strings along check-check qubits in the HGP code can be stabilizer-equivalent to a disconnected error cluster, as discussed above. Proposing that we get an extra entropic factor of order $3^\ell$ at distance $\ell$ from the leaves, a crude argument suggests that we will gain an added entropic factor of at least $3 \times 3^2 \times 3^3 \times \cdots \times 3^r \sim 3^{r^2/2} \sim n^{r/2}$ in \eqref{eq: bottle neck C}.  This suggests that $c\gtrsim 0.25$, which is actually quite close to what we observe in Fig. \ref{fig:tree numerics}.

We also numerically simulate several codes from the $(4,5)$-HGP ensemble and do not observe evidence of any decoding threshold (see the rightmost plot in Fig. \ref{fig:tree numerics}), despite the fact that the classical $(4,5)$-LDPC ensemble is capable of passive error correction (also recall Fig. \ref{fig:numerics}). Our numerics suggest that certain classical passive memories may not have straightforward quantum analogues through the hypergraph product due to the additional degrees of freedom introduced by the graph product. Since it was recently shown that the energy barriers of hypergraph product codes are inherited from their classical input codes \cite{zhao2024}, an interesting direction for future work would be to carefully study the additional entropic effects needed to obtain quantum memory in these codes.

\section{Open quantum system dynamics} \label{app:implementation}

This section contains details of our discussion of MFQEC and passive error correction.

\subsection{Lindblad dynamics}
Our goal is to prepare a stationary Gibbs state \begin{equation}
    \sigma \propto \mathrm{e}^{-\beta \mathcal{H}},
\end{equation}
where \begin{equation}
    \mathcal{H} = -\sum_{\text{stabilizer }a}S_a
\end{equation}and as in the main text, the stabilizers $\{ S_a \}$ are a set of mutually commuting Pauli strings, i.e., $S_a^2 = \ident$ and $[S_a, S_b] = 0$ for all $a$, $b$.  In this section, we assume that each $S_a$ is independent (in the sense that one cannot be built out of a product of any others), as was the case for the non-redundant codes described above. The Gibbs state is diagonal in a basis that consists of eigenvectors of all Pauli strings $\{ S_a \}$, which we can denote as $\sigma = \sum_s \sigma_s |s\rangle\langle s|$.  Notice that with this choice, there remain $2^K$ possible $s$ for each set of stabilizer eigenvalues, and we will be a little lazy and omit that explicit notation.

Let's start with ideal circuits without any errors. We choose a Lindbladian of Davies form \cite{Davies_1979}:
\begin{equation}
    \partial_t \rho = \mathcal{L}(\rho) = \sum_{s s'}  L_{ss'} \left( |s\rangle\langle s'| \rho |s'\rangle\langle s|  -   \frac{1}{2}\{ |s'\rangle\langle s'|, \rho \} \right) 
    \, ,
\end{equation}
where the coefficients $L_{ss'}$ are positive and satisfy $L_{ss'} = \mathrm{e}^{-\beta \left(\mathcal{H}_s - \mathcal{H}_{s'}\right)}L_{s's}$. We also want the dynamics to be local, which means all the jump operators only act on finite numbers of qubits. Consider an operator $P$, which can be a Pauli operator on a single qubit or a Pauli string on multiple qubits. If there exists some stabilizers that anticommute with $P$: $\mathcal{A}_P \equiv \{ a \, | \, S_a P + P S_a = 0 \}$ and $\mathcal{A}_P \neq \emptyset$, we can choose jump operators that consist of $P$ and projections:
\begin{equation}
    P(\mathbf{n}) = P \Pi_P(\mathbf{n}) \equiv P \left[ \prod_{a \in \mathcal{A}_P} \frac12 (I + n_a S_a) \right]   
    \, ,
    \label{eqn:stabilizer-eigenvector}
\end{equation}
where $\mathbf{n} \equiv \{n_a\}_{a \in \mathcal{A}_P}$ and $n_a \in \{-1, +1\}$, which defines the projector $\Pi_P(\mathbf{n})$. Note that $P = \sum_{\mathbf{n}} P(\mathbf{n})$. Since we have an LDPC CSS code, any local operator $P$ is only involved in a finite number of stabilizers, and each stabilizer has support on a finite number of Pauli operators, so the jump operator $P(\mathbf{n})$ is also $\mathsf{k}$-local. A $\mathsf{k}$-local Lindbladian with stationary state $\sigma$ can then be found explicitly \cite{aqm}:
\begin{equation} \label{eq: local Lindbladian}
    \mathcal{L}(\rho) = \sum_{P, \mathbf{n}} \gamma_{P(\mathbf{n})} \mathcal{L}_{P(\mathbf{n})}(\rho) = \sum_{P, \mathbf{n}}  \gamma_{P(\mathbf{n})} \left( P(\mathbf{n}) \rho P(\mathbf{n})^\dagger  -   \frac{1}{2}\{ \Pi_P(\mathbf{n}), \rho \} \right) 
    \, ,
\end{equation}
with \begin{equation}
    \gamma_{P(\mathbf{n})} = \mathrm{e}^{-2\beta \sum_{a \in \mathcal{A}_P} n_a}\gamma_{P(\mathbf{-n})}
\end{equation}
corresponding to ``Gibbs sampling" dynamics that protects the Gibbs mixed state.

We now show that the dynamics above can be effectively implemented using ancilla qubits and quantum gates. If we want to implement the dynamics $\mathcal{L}_{P(\mathbf{n})}$, we need an ancilla qubit that is initialied as $\ket{0}$ for each $a \in \mathcal{A}_P$. The ancilla qubits can store information about the physical qubits, which can then be used to apply feedback. For each stabilizer $S_a = \prod_i P_i$ that anticommutes with $P$, where $P_i$ is a Pauli operator on qubit $i$, we apply gates $\prod_i \mathrm{e}^{\ii \frac{\pi}{4} (I_i -P_i ) (I_a - X_a)}$ to extract syndromes, where $I_a$ and $X_a$ are operators on the ancilla qubits for $S_a$. After applying gates for each stabilizers, we apply gate \begin{equation}
    U_{\mathrm{feedback}} = \exp\left[\ii\frac{\pi}{2}(I-P)\prod_a \frac{I_a+n_aZ_a}{2}\right]
\end{equation}
as the feedback. The last step is to reset all ancilla qubits to $\ket{0}$.

Neglecting the time it takes to apply the gates above, the dynamics we implemented in the previous paragraph on physical and ancilla qubits can be converted into an effective Lindbladian acting only on the original physical qubits. After tracing out ancilla qubits, we find that
\begin{equation} \label{eq: no feedback terms}
    \mathcal{L}_{P(\mathbf{n})}'(\rho) = P(\mathbf{n}) \rho P(\mathbf{n})^\dagger  -   \frac{1}{2}\{ \Pi_P(\mathbf{n}), \rho \} + \sum_{\mathbf{m}\neq \mathbf{n}} \Pi_P (\mathbf{m}) \rho \Pi_P (\mathbf{m}) - \frac{1}{2}\{ \Pi_P(\mathbf{m}), \rho \}
    \, .
\end{equation}
Since the projector $\Pi_P(\mathbf{m})$ commutes with $\sigma$, the additional dynamics we implemented other than $\mathcal{L}_{P(\mathbf{n})}$ will not affect the stationarity of $\sigma$, so we can effectively implement $\mathcal{L}_{P(\mathbf{n})}$ with the above process. 

Ultimately, using this Lindbladian picture, we can predict the rates at which we should correct for errors in the MFQEC circuit; the only difference in this mathematical framework is that we apply this feedback effectively instantaneously and at a continuous rate for each qubit. 

For the whole dynamics $\mathcal{L}$ in (\ref{eq: local Lindbladian}), instead of implementing only $\mathcal{L}_{P(\mathbf{n})}$ at a time, we can implement $\sum_\mathbf{n} \gamma_{P(\mathbf{n})} \mathcal{L}_{P(\mathbf{n})}$ together: for each operator $P$, in time interval $\text{d}t$, we apply gates that extract syndromes with probability \begin{equation}
    p_{\text{syndrome}} \propto \text{d}t\, \max_{\{\mathbf{m}\}}\left(\gamma_{P(\mathbf{m})}\right),
\end{equation} then apply $U_{\mathrm{feedback}}$ for $P(\mathbf{n})$ with probability \begin{equation}
    p_{P(\mathbf{n})} = \frac{\gamma_{P(\mathbf{n})}}{\max_{\{\mathbf{m}\}}\left(\gamma_{P(\mathbf{m})}\right)}.
\end{equation}
The Lindbladian thus becomes
\begin{align} \label{eq: local Lindbladian in practice}
    \mathcal{L}'(\rho) &= \sum_P \max_{\{\mathbf{m}\}}\left(\gamma_{P(\mathbf{m})}\right) \sum_{\mathbf{n}}\mathcal{L}^{\mathrm{ideal}}_{P(\mathbf{n})}(\rho) \notag \\
    &= \sum_P \max_{\{\mathbf{m}\}}\left(\gamma_{P(\mathbf{m})}\right) \sum_{\mathbf{n}} \left( p_{P(\mathbf{n})} P(\mathbf{n}) \rho P(\mathbf{n})^\dagger +\left(1- p_{P(\mathbf{n})}\right)\Pi_P (\mathbf{n}) \rho \Pi_P (\mathbf{n}) -   \frac{1}{2}\{ \Pi_P(\mathbf{n}), \rho \} \right) 
    \, ,
\end{align}
which is similar to $\mathcal{L}$ in (\ref{eq: local Lindbladian}), differing only in terms that don't affect the stationarity of $\sigma$.

Fig.~\ref{fig:circuit} shows two examples of such dynamics. There are three stabilizers: $S_1 = Z_1Z_2Z_3Z_4$, $S_2 = Z_3 Z_4$ and $S_3 = Z_4 Z_5$. In the first example, consider Pauli operator $X_3$ that act on $q_3$, there are two stabilizers that anticommute with $X_3$, so jump operators related to $X_3$ would be $X_3(n_1, n_2) = \frac{1}{4}X_3(I + n_1 S_1)(I + n_2 S_2) $. The first dashed box in Fig.~\ref{fig:circuit} together with gates applied on $a_1$ and $a_2$ show how to implement $\mathcal{L}_{X_3(-1,-1)}$. Similarly, the second box and all other gates outside the boxes show how to implement $\mathcal{L}_{X_4(1,-1,-1)}$. The channel $R$ resets the ancilla qubits to $\ket{0}$, and is implemented as follows: \begin{equation}
    R_i[\rho] = \mathrm{tr}_i \rho \otimes |0\rangle\langle 0|_i.
\end{equation}

\begin{figure}[t]
    \centering
    \begin{quantikz} [row sep={0.7cm,between origins}]
    \lstick{$\ket{q_1}$} & \ctrl{5} & & \qw & \qw & \qw & \qw & \qw & \qw & \qw &  \\
    \lstick{$\ket{q_2}$} & \qw & \ctrl{4}   & \qw & \qw & \qw & \qw & \qw & \qw & \qw &  \\
    \lstick{$\ket{q_3}$} & \qw & \qw  & \ctrl{4}   & \qw & \qw & \targ{}\gategroup[6,steps=1,style={dashed,rounded corners, inner sep=2pt},background,label style={label position=below,anchor=north,yshift=-0.2cm}]{{\sc $X_3$}} & \qw & \qw & \qw &  \\
    \lstick{$\ket{q_4}$} & \qw & \qw  & \qw  & \ctrl{4} & \qw & \qw & \qw\gategroup[5,steps=2,style={dashed,rounded corners, inner sep=2pt},background,label style={label position=below,anchor=north,yshift=-0.2cm}]{{\sc $X_4$}} & \targ{} & \qw  &\\
    \lstick{$\ket{q_5}$} & \qw & \qw  & \qw  & \qw & \ctrl{3} & \qw  & \qw & \qw & \qw &\\
    \lstick{$\ket{a_1} = |{0}\rangle$}  & \targ{} & \targ{} & \targ{} & \targ{} & \qw & \ctrl{-3} & \gate{X} & \ctrl{-2} & \gate{R} & \qw \\
    \lstick{$\ket{a_2} = |{0}\rangle$}  & \qw & \qw & \targ{} & \targ{}  & \qw & \ctrl{-1} & \qw & \ctrl{-1} & \gate{R} &\qw \\
    \lstick{$\ket{a_3} = |{0}\rangle$}  & \qw  & \qw & \qw & \targ{} & \targ{} & \qw & \qw & \ctrl{-1} & \gate{R} &\qw
    \end{quantikz}%
    \caption{Two examples of how to implement dynamics in (\ref{eq: local Lindbladian}). There are three stabilizers: $S_1 = Z_1Z_2Z_3Z_4$, $S_2 = Z_3 Z_4$ and $S_3 = Z_4 Z_5$ so we need three ancilla qubits to extract syndromes by the CNOT gates outside the dashed boxes. The gates in the boxes apply $X_3$ and $X_4$ according to different syndromes. The reset channel $R$ forces ancilla qubits to their initial state $\ket{0}$.}
    \label{fig:circuit}
\end{figure}
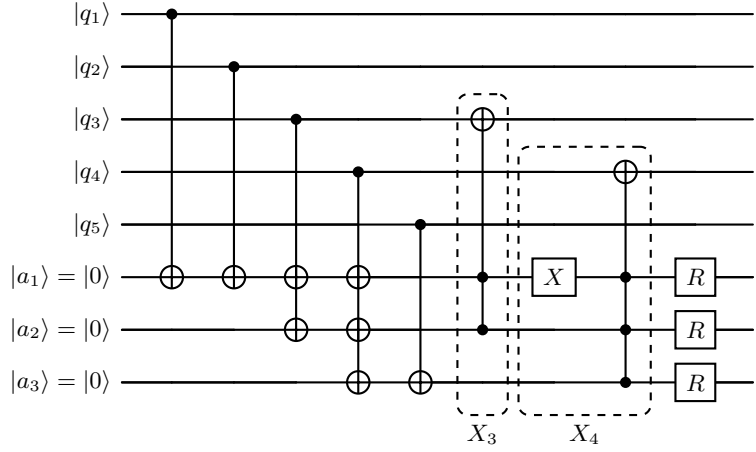

\subsection{Implementing a perfect Gibbs sampler}
Let us illustrate how to perform MFQEC for our Gibbs sampler, using $\le 3$-qubit gates, and reset.  For simplicity, we focus only on the step that ``corrects" $X$ errors on qubit $i$, as the story is analogous for $Z$ errors and other qubits. For now we also assume that the circuit has no unknown errors; we will relax that assumption shortly.   Let ancilla qubits $a_1,\ldots,a_{\Delta_C}$ store the ``measurement outcomes" for the $\Delta_C$ $Z$-checks $C_1,\dots,C_q$ containing physical qubit $i$.  The first step of MFQEC is to apply the unitary 
\begin{equation}
    U_1 = \prod_{q=1}^{\Delta_C} \prod_{j\in\mathrm{supp}(C_q)} \mathsf{CNOT}_{j \rightarrow a_q},
\end{equation}
where $\mathsf{CNOT}_{i\rightarrow j} = \frac{1+Z_i}{2}+\frac{1-Z_i}{2}X_j$.  In the second step, we apply the $\mathsf{SORT}_{\Delta_C}$ channel, which shuffles the ancilla qubits such that the zeros are to the left and the ones are to the right. We later provide an implementation for $\mathsf{SORT}$. In order to realize a perfect Gibbs sampler, we apply a coherent gate that implements the following controlled-unitary using one additional ancilla qubit $c$, which is again initialized in state $|0\rangle$:
\begin{equation}
    U_2 = \left[\prod_{j=1}^{\Delta_C} \mathrm{e}^{\mathrm{i}\theta_j \big(1-Z_{a_j}\big) X_i Z_c/2 }\right]X_c \, ,
\end{equation}
where we choose each angle $\theta_j$ such that
\begin{equation}
    \sin^2 \left(\sum_{k=j}^{\Delta_C}\theta_k\right) = \mathrm{e}^{2\beta \min(0,\Delta_C-2j)} \, . \label{eq:sin2theta}
\end{equation}
With these rotation angles, the cumulative probability of applying $X_i$ matches those precisely set by the Metropolis algorithm. In particular, we \emph{always} apply feedback if doing so reduces the energy, i.e. the number of violated parity checks; however, we also introduce our own errors with a small probability. 
Once this is done, we apply the reset channel to all ancilla qubits. Tracing out $c$ gives us a reduced density matrix that contains superpositions of modes rotated by $\pm \theta$, which is the same as an incoherent Pauli $X$ error with probability $\sin^2\theta$.  

Given the Gibbs sampler above, we can readily convert it into a genuine passive decoder in the presence of unwanted errors.  If, for example, there are single-qubit incoherent $X$ errors that occur at rate $\gamma$, then we modify (\ref{eq:sin2theta}) to \begin{equation}
    \sin^2 \left(\sum_{k=j}^{\Delta_C}\theta_k\right) = \left(\mathrm{e}^{2\beta \min(2j,\Delta_C)}-1\right)\gamma.
\end{equation}
Combining our modified sampler with the intrinsic errors reproduces a perfect Gibbs sampler, which we have seen is a passive quantum memory.

\subsection{Fault tolerance}

Now let us relate the Gibbs sampling dynamics above to quantum error correction. For simplicity, we focus on correcting single-qubit Pauli errors, with generic correction strategies described in \cite{aqm}.    If $Q$ is a Pauli string acting on O(1) sites, then adding such an error modifies our Lindbladian to:
\begin{equation}
    \mathcal{L}(\rho) \to \mathcal{L}(\rho) + g\left( Q\,\rho\, Q - \rho \right)
    \, ,
    \label{eqn:nonunitary error}
\end{equation}
where $g>0$ is the rate of the Pauli error, and $Q$ is a single-qubit Pauli that doesn't commute with some of the stabilizers. Since $Q = \sum_{\mathbf{n}} Q(\mathbf{n})$, we have
\begin{equation}
    g\left( Q\,\rho\, Q - \rho \right) = 
    g\sum_{\mathbf{m},\mathbf{n}} \left( Q(\mathbf{m})\, \rho\, Q(\mathbf{n})^\dagger - \frac12 \{ Q(\mathbf{n})^\dagger Q(\mathbf{m}), \rho \} \right)
    \,.
\end{equation}
Because \begin{equation} \label{eq: offdiagonal terms}
    Q(\mathbf{m})\, \sigma\, Q(\mathbf{n})^\dagger - \frac12 \{ Q(\mathbf{n})^\dagger Q(\mathbf{m}), \sigma \} = \delta_{\mathbf{m}\mathbf{n}} \left( Q \Pi_Q(\mathbf{n})\, \sigma \, \Pi_Q(\mathbf{n}) Q- \frac12 \{ \Pi_Q(\mathbf{n}), \sigma \}\right),
\end{equation} when $\mathbf{m} \neq \mathbf{n}$, this part of error doesn't affect stationarity and doesn't need to be corrected. For the remaining terms that do have $\mathbf{m} = \mathbf{n}$, this part is in the same form as (\ref{eq: local Lindbladian}) with different coefficients. We can just increase these coefficients until the ratios between the coefficients are also the same as (\ref{eq: local Lindbladian}) so the stationarity of $\sigma$ still holds and the error is effectively corrected. As an example, one way to correct the error in (\ref{eqn:nonunitary error}) is to increase the coefficient $\gamma_{Q(\mathbf{n})}$ by 
\begin{equation} \label{eq: Pauli error protocol}
\delta \gamma_{Q(\mathbf{n})} = \text{sgn}\left(-\sum n\right) \left(\mathrm{e}^{-2\beta \sum n}-1 \right)g
\, .
\end{equation}

So far, we assumed all the gates we apply are perfect.   We will now assume that the system is subject to incoherent Pauli errors for each possible qubit (physical and ancilla). For all the \textsf{CNOT}-like gates, the corresponding errors are equivalent to applying random Pauli strings on both sides of the gates. For the reset operations, the errors would lead to the wrong initial states of the ancilla qubits, which is also equivalent to applying Pauli gates to the correct initial states. We can move all the Pauli errors to the right-hand side of these gates and then implement a modified dynamics.  For example, suppose we wish to implement the process $\mathcal{L}_{P(\mathbf{n})}$ described above, which in general contained unitary operations corresponding to gates of the form $\prod_i \mathrm{e}^{\ii \frac{\pi}{4} (I_i -P_i ) (I_a - X_a)}$.  These gates become \begin{equation}
    \prod_i \mathrm{e}^{\ii \frac{\pi}{4} (I_i -P_i ) (I_a - X_a)} \rightarrow \prod_i Z_{\mathrm{error}} \mathrm{e}^{\ii \frac{\pi}{4} (I_i -P_i ) (I_a - X_a)} Z_{\mathrm{error}} = \prod_i \mathrm{e}^{\ii \frac{\pi}{4} (I_i \pm P_i ) (I_a \pm X_a)}, 
\end{equation}where the $\pm$ signs at the end correspond to whether or not the error commuted or anticommuted with the $P_i$ or $X_a$.   We can handle other errors analogously: $\mathrm{e}^{\ii\frac{\pi}{2} (I - P)\prod_a \frac{1}{2} (I_a +n_a Z_a)} \rightarrow \mathrm{e}^{\ii\frac{\pi}{2} (I \pm P)\prod_a \frac{1}{2} (I_a \pm n_a Z_a)}$ after moving all the Pauli errors to the right, which is equivalent to implementing $\mathcal{L}_{P(\mathbf{n}')}$ with $\mathbf{n}' \neq \mathbf{n}$ and then applying the Pauli errors. Recall that $\mathcal{L}_{P(\mathbf{n})}$ was defined in (\ref{eq: local Lindbladian}). Putting all of this together, we find that our Lindbladian $\mathcal{L}_{P(\mathbf{n})}$ is modified to
\begin{align} \label{eq: gate errors}
   \mathcal{L}_{P(\mathbf{n})}(\rho) &\rightarrow \sum _{Q_e,\mathbf{n}'} p(Q_e,\mathbf{n}'| P(\mathbf{n}))\left( Q_e P(\mathbf{n}')\rho P(\mathbf{n}')^\dagger Q_e - \frac{1}{2}\{ \Pi_P(\mathbf{n}'), \rho \} \right) \notag \\
    &\rightarrow \mathcal{L}_{P(\mathbf{n})}^{\mathrm{ideal}}(\rho) + \sum_{Q_e,\mathbf{m}} \delta p_{Q_e P(\mathbf{m}),P(\mathbf{n})} \mathcal{L}_{Q_e P(\mathbf{m})}(\rho) 
    ,
\end{align}
where in the first line, $Q_e$ are the Pauli operations (including $I$, or having no error) that have been moved to the right and $p(Q_e,\mathbf{n}'| P(\mathbf{n}))$ is the probability of having Pauli ``error" $Q_e$ and applying $P(\mathbf{n}')$ when the ``apparent" syndrome measurements returned $\mathbf{n}$. In the second line, we rewrite the expression to emphasize that when $p(Q_e,\mathbf{n}^\prime|P(\mathbf{n}))$ is small (i.e. the gates are high fidelity), the Lindbladian is quite close to the ``ideal case".
As the ancilla qubit feedback only corrects for single qubit errors, we emphasize that $Q_e$ will contain only single qubit terms once we trace out the state of the ancilla qubits (as is appropriate, once they are reset during MFQEC).

In the presence of errors throughout our circuit, it is useful to expand the sum over $\mathbf{n}$ above to incorporate all of the $X$ or $Z$ type stabilizers that act on a given site (not just the ones that anticommute with Pauli $P$).  After all, the error $Q_e$ may be a $Z$ type error if we intended to correct an $X$ type error.  So the most generic type of Lindbladian modeling our MFQEC protocol takes the form of \begin{equation}
    \mathcal{L} = \sum_{P,\mathbf{n}} \gamma_{P(\mathbf{n})}\mathcal{L}^{\mathrm{ideal+error}}_{P(\mathbf{n})}  + \sum_{P,\mathbf{n}}\delta \gamma_{Q_e P(\mathbf{n})} \mathcal{L}_{Q_e P(\mathbf{n})}(\rho) 
\end{equation}where $\mathcal{L}^{\mathrm{ideal+error}}_{P(\mathbf{n})}$ includes the physical qubit errors discussed in (\ref{eqn:nonunitary error}), and \begin{equation}
    \delta \gamma_{Q_eP(\mathbf{n})} = \sum_{Q^\prime P(\mathbf{n}^\prime)} \delta p_{Q_e P(\mathbf{m}),Q^\prime(\mathbf{n^\prime})}\cdot \gamma_{Q^\prime(\mathbf{n}^\prime)}.
\end{equation}  We will consider an ideal protocol where we apply $Y$ errors deliberately at some rate, to ensure that all $\gamma_{P(\mathbf{n})}>0$, which will shortly prove convenient.  Our goal is to show that $\mathcal{L}$ is an \emph{exact} Gibbs sampler, which will be satisfied so long as the detailed balance condition \cite{aqm} \begin{equation} \label{eq:ugly}
(\gamma+\delta \gamma)_{P(\mathbf{n}_{\mathrm{acomm}},\mathbf{n}_{\mathrm{comm}})} = (\gamma+\delta \gamma)_{P(-\mathbf{n}_{\mathrm{acomm}},\mathbf{n}_{\mathrm{comm}})}\mathrm{e}^{-2\beta \sum_{\text{acomm }a} n_a} 
\end{equation}
holds. Here $\mathbf{n}_{\mathrm{acomm/comm}}$ denote the stabilizer values of the anticommuting/commuting stabilizers that act on the same site as single-qubit Pauli $P$.  (\ref{eq:ugly}) is clearly solved if 
\begin{equation}
    (\gamma+\delta \gamma)_{P(\mathbf{n}_{\mathrm{acomm}},\mathbf{n}_{\mathrm{comm}})}  = \gamma_0 \times \max\left(1, \mathrm{e}^{-2\beta \sum_{\text{acomm }a} n_a}\right) = \gamma^{\mathrm{Gibbs}}_{P(\mathbf{n}_{\mathrm{acomm}},\mathbf{n}_{\mathrm{comm}})} \, , \label{eq:solvegibbs}
\end{equation}where $\gamma_0>0$ is some fixed positive constant.  Since $\gamma^{\mathrm{Gibbs}}$ is positive, we deduce that there exists a \emph{finite, $N$-independent} value of $\gamma_0$ such that for sufficiently small error rate $\delta p$, (\ref{eq:solvegibbs}) admits a solution for which all of the rates $\gamma_{P(\mathbf{n})}>0$, because for sufficiently small $\delta p$, the matrix $(I+\delta p)^{-1}$ is invertible and has bounded matrix elements.

\subsection{Sort channel implementation}

The Gibbs sampler that we described in the main text requires a non-unitary $\mathsf{SORT}$ channel to work.  Here we provide an explicit implementation for this channel.  It can be efficiently realized with sorting networks \cite{Art_of_CS, bundala2013, Codish_2014, harder2022}, which are circuits constructed using a 2-input binary comparator as a primitive: see Fig. \ref{fig:comparator} for a circuit depiction and Table \ref{tab:sorting networks} for an optimal gate-cost analysis. The use of native Toffoli (\textsf{CCNOT}) gates in Rydberg systems \cite{Shi_2018, Yin_2020} would allow one to bypass the two-qubit gate decomposition and significantly reduce the overhead.

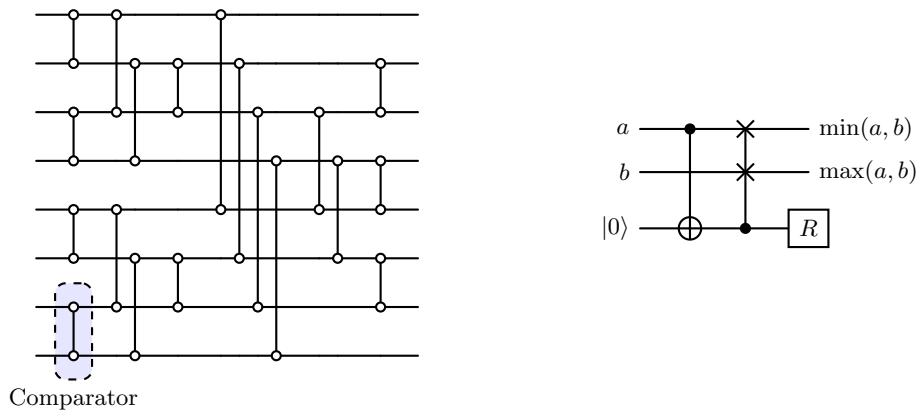
\begin{figure}[t]
    \centering
    \begin{quantikz}[column sep=0.3em]
        & [1em]\ctrl[open]{1} & [1em]\ctrl[open]{2} &&[1em]& [1em]\ctrl[open]{4} &&&&[1em]&&[1em]&[1em]  \\
        & \ctrl[open]{0} && \ctrl[open]{2} & \ctrl[open]{1} && \ctrl[open]{4} &&&&& \ctrl[open]{1} &  \\
        & \ctrl[open]{1} & \ctrl[open]{0} && \ctrl[open]{0} &&& \ctrl[open]{4} && \ctrl[open]{2} && \ctrl[open]{0} &  \\
        & \ctrl[open]{0} && \ctrl[open]{0} &&&&& \ctrl[open]{4} && \ctrl[open]{2} & \ctrl[open]{1} &  \\
        & \ctrl[open]{1} & \ctrl[open]{2} &&& \ctrl[open]{0} &&&& \ctrl[open]{0} && \ctrl[open]{0} &  \\
        & \ctrl[open]{0} && \ctrl[open]{2} & \ctrl[open]{1} && \ctrl[open]{0} &&&& \ctrl[open]{0} & \ctrl[open]{1} &  \\
        & \ctrl[open]{1}\gategroup[2,steps=1,style={dashed,rounded corners,fill=blue!10, inner xsep=2pt},background,label style={label position=below,anchor=north,yshift=-0.2cm}]{Comparator} & \ctrl[open]{0} && \ctrl[open]{0} &&& \ctrl[open]{0} &&&& \ctrl[open]{0} &  \\
        & \ctrl[open]{0} && \ctrl[open]{0} &&&&& \ctrl[open]{0} &&&& 
    \end{quantikz}
    \hspace{2cm}
    \begin{quantikz}
        \lstick{$a$} & \ctrl{2} & \swap{1} & \rstick{$\min(a,b)$} \\
        \lstick{$b$} & & \swap{1} & \rstick{$\max(a,b)$} \\
        \lstick{$\ket{0}$} & \targ{} & \ctrl{0} & \gate{R}
    \end{quantikz}
    \caption{\textbf{Left:} An optimal sorting network for 8 inputs is depicted using 19 comparators with a circuit depth of 6 \cite{Batcher_1968}. \textbf{Right:} The 2-input comparator gadget (boxed in the left circuit) is shown, which sorts the binary inputs $a,b$ using an ancilla, a \textsf{CNOT} gate and a Fredkin (\textsf{CSWAP}) gate. The Fredkin gate can be decomposed into 7 \textsf{CNOT} gates and single-qubit rotations \cite{Cruz_2024}.}
    \label{fig:comparator}
\end{figure}

\begin{table}[t]
    \centering
    \def\arraystretch{1.5}
    \setlength\tabcolsep{0.35em}
    \begin{tabular}{|c|c|c|c|c|c|c|c|c|c|c|}
    \hline
    $n^{}_{\rm in}$ & 3  & 4  & 5  & 6  & 7   & 8   & 9   & 10  & 11  & 12  \\ \hline
    Depth            & 3  & 3  & 5  & 5  & 6   & 6   & 7   & 7   & 8   & 8   \\ \hline
    Size             & 3  & 5  & 9  & 12 & 16  & 19  & 25  & 29  & 35  & 39  \\ \hline
    2-qubit gates       & 24 & 40 & 72 & 96 & 128 & 152 & 200 & 232 & 280 & 312 \\ \hline
    \end{tabular}
    \caption{The depth and size (number of comparators) of optimal sorting networks for a range of inputs is tabulated. The number of 2-qubit gates ($8\times$ size) upon gate decomposition is also shown.}
    \label{tab:sorting networks}
\end{table}

As a concrete example, let us now compare  passive vs. active quantum error correction for the smallest (guaranteed) self-correcting HGP code, again focusing on the case for $X$ errors as the $Z$ error discussion is analogous.  To keep the discussion a little bit simpler, we focus on just a simple-majority ``greedy decoder", which we expect will perform at least as well as the Gibbs sampler (which sometimes introduces errors). We draw a random classical code from the regular $(7,8)$-LDPC ensemble which exhibits $\delta<1/6$ expansion with high probability (Theorem \ref{thm:random expansion}). We then take the hypergraph product of this classical code with itself to obtain a $(8,15)$-LDPC quantum CSS code. Since the check weight is 15, the syndrome extraction circuit to couple measure/ancilla qubits with physical qubits via two-qubit gates has depth at least 15. At this point, the procedures for active and passive error correction differ. For active error correction, we would now simply measure all ancilla qubits in the computational basis and feed the results into a classical decoding algorithm. The output of this decoding algorithm then allows us to undo the error or alternatively track the Pauli frame of the code in software. For passive error correction, we essentially need to run the decoding algorithm as a quantum circuit itself. Since each physical qubit participates in no more than 8 checks, according to Table \ref{tab:sorting networks}, we need to run a depth-6 sorting network on 8 ancilla qubits per physical qubit. Each comparator contains a Fredkin gate (Fig. \ref{fig:comparator}), which can be further decomposed into 2 \textsf{CNOT}s and a Toffoli. With the ability to perform native Toffoli gates, our ancilla sorting circuit will thus have depth 24; we note that further parallelization may reduce this depth. After sorting, we then apply a \textsf{CNOT} gate between the 4th ancilla qubit from the right (control) and the physical qubit (target), which implements the majority vote. The total circuit depth for passive decoding after syndrome extraction is hence 25, assuming native Toffoli gates.  Using a direct majority-vote circuit, rather than sorting, may further reduce the gate cost \cite{Kheirandish_2020}; e.g. a 5-input majority voter would use 50 2-qubit gates instead of 72 for the sorting network.

For disconnected physical qubits that do not share any checks, we run simultaneous decoding in parallel. To run this entire decoding step in one round, we can introduce 14 additional ancilla qubits per syndrome qubit to which we copy the syndrome values via \textsf{CNOT} gates, thus enabling simultaneous decoding on all physical qubits.  The copying of each measurement outcome to many different ancillas will add at most circuit depth 14.

We note that additional overhead will be required if we wanted to implement the full Gibbs sampler.

\bibliography{thebib}
\end{document}

%% file: preamble.tex
\usepackage{graphicx}
\usepackage[dvipsnames]{xcolor}
\usepackage{rotating}
\usepackage{amsmath,amssymb,graphics,amsthm,isomath}
\usepackage{amsfonts,dsfont,mathtools}
\usepackage{bbm}
\usepackage{bm}
\usepackage{array}
\usepackage{appendix}
\newcolumntype{P}[1]{>{\centering\arraybackslash}p{#1}}
\usepackage{multirow}
\usepackage{physics}
\usepackage{tikz-cd}
\usepackage{adjustbox}
\usepackage{stmaryrd}

\usepackage[colorlinks=true, urlcolor=violet, linkcolor=blue, citecolor=red, hyperindex=true, linktocpage=true]{hyperref}
\usepackage[capitalise,compress]{cleveref}

\allowdisplaybreaks

\newtheorem{thm}{Theorem}
\numberwithin{thm}{section}
\newtheorem{cor}[thm]{Corollary}
\newtheorem{lem}[thm]{Lemma}

\newtheorem{defn}[thm]{Definition}

\renewcommand{\thesection}{\arabic{section}}
\renewcommand{\thesubsection}{\thesection.\arabic{subsection}}

\makeatletter
\renewcommand{\p@subsection}{}
\renewcommand{\p@subsubsection}{}
\makeatother

\usepackage{xcolor}
\usepackage{mathtools}

\newcommand\bea{\begin{eqnarray}}
\newcommand\eea{\end{eqnarray}}
\newcommand\be{\begin{equation}}
\newcommand\ee{\end{equation}}
\newcommand\bes{\begin{subequations}}
\newcommand\ees{\end{subequations}}
\newcommand\bed{\begin{displaymath}}
\newcommand\eed{\end{displaymath}}
\newcommand\beal{\begin{aligned}}
\newcommand\eeal{\end{aligned}}
\newcommand\bew{\begin{widetext}}
\newcommand\eew{\end{widetext}}
\newcommand\beit{\begin{itemize}}
\newcommand\eeit{\end{itemize}}
\def\bea{\begin{array}}
\def\eea{\end{array}}
\newcommand\been{\begin{enumerate}}
\newcommand\eeen{\end{enumerate}}

\newcommand{\ident}[0]{\mathds{1}}

\newcommand{\ii}{\mathrm{i}}

\usepackage{verbatim}

\newcommand{\transpose}[0]{\mathsf{T}}

\usepackage{tikz}
\usetikzlibrary{quantikz2}


%% file: arxiv.bbl
\begin{thebibliography}{111}%
\makeatletter
\providecommand \@ifxundefined [1]{%
 \@ifx{#1\undefined}
}%
\providecommand \@ifnum [1]{%
 \ifnum #1\expandafter \@firstoftwo
 \else \expandafter \@secondoftwo
 \fi
}%
\providecommand \@ifx [1]{%
 \ifx #1\expandafter \@firstoftwo
 \else \expandafter \@secondoftwo
 \fi
}%
\providecommand \natexlab [1]{#1}%
\providecommand \enquote  [1]{``#1''}%
\providecommand \bibnamefont  [1]{#1}%
\providecommand \bibfnamefont [1]{#1}%
\providecommand \citenamefont [1]{#1}%
\providecommand \href@noop [0]{\@secondoftwo}%
\providecommand \href [0]{\begingroup \@sanitize@url \@href}%
\providecommand \@href[1]{\@@startlink{#1}\@@href}%
\providecommand \@@href[1]{\endgroup#1\@@endlink}%
\providecommand \@sanitize@url [0]{\catcode `\\12\catcode `\$12\catcode `\&12\catcode `\#12\catcode `\^12\catcode `\_12\catcode `\%12\relax}%
\providecommand \@@startlink[1]{}%
\providecommand \@@endlink[0]{}%
\providecommand \url  [0]{\begingroup\@sanitize@url \@url }%
\providecommand \@url [1]{\endgroup\@href {#1}{\urlprefix }}%
\providecommand \urlprefix  [0]{URL }%
\providecommand \Eprint [0]{\href }%
\providecommand \doibase [0]{http://dx.doi.org/}%
\providecommand \selectlanguage [0]{\@gobble}%
\providecommand \bibinfo  [0]{\@secondoftwo}%
\providecommand \bibfield  [0]{\@secondoftwo}%
\providecommand \translation [1]{[#1]}%
\providecommand \BibitemOpen [0]{}%
\providecommand \bibitemStop [0]{}%
\providecommand \bibitemNoStop [0]{.\EOS\space}%
\providecommand \EOS [0]{\spacefactor3000\relax}%
\providecommand \BibitemShut  [1]{\csname bibitem#1\endcsname}%
\let\auto@bib@innerbib\@empty
\bibitem [{\citenamefont {Thomas}(1989)}]{Thomas1989}%
  \BibitemOpen
  \bibfield  {author} {\bibinfo {author} {\bibfnamefont {Lawrence~E.}\ \bibnamefont {Thomas}},\ }\bibfield  {title} {\enquote {\bibinfo {title} {Bound on the mass gap for finite volume stochastic ising models at low temperature},}\ }\href {\doibase 10.1007/BF02124328} {\bibfield  {journal} {\bibinfo  {journal} {Communications in Mathematical Physics}\ }\textbf {\bibinfo {volume} {126}},\ \bibinfo {pages} {1--11} (\bibinfo {year} {1989})}\BibitemShut {NoStop}%
\bibitem [{\citenamefont {Dennis}\ \emph {et~al.}(2002)\citenamefont {Dennis}, \citenamefont {Kitaev}, \citenamefont {Landahl},\ and\ \citenamefont {Preskill}}]{Dennis_2002}%
  \BibitemOpen
  \bibfield  {author} {\bibinfo {author} {\bibfnamefont {Eric}\ \bibnamefont {Dennis}}, \bibinfo {author} {\bibfnamefont {Alexei}\ \bibnamefont {Kitaev}}, \bibinfo {author} {\bibfnamefont {Andrew}\ \bibnamefont {Landahl}}, \ and\ \bibinfo {author} {\bibfnamefont {John}\ \bibnamefont {Preskill}},\ }\bibfield  {title} {\enquote {\bibinfo {title} {Topological quantum memory},}\ }\href {\doibase 10.1063/1.1499754} {\bibfield  {journal} {\bibinfo  {journal} {Journal of Mathematical Physics}\ }\textbf {\bibinfo {volume} {43}},\ \bibinfo {pages} {4452–4505} (\bibinfo {year} {2002})}\BibitemShut {NoStop}%
\bibitem [{\citenamefont {Onsager}(1944)}]{Onsager_1944}%
  \BibitemOpen
  \bibfield  {author} {\bibinfo {author} {\bibfnamefont {Lars}\ \bibnamefont {Onsager}},\ }\bibfield  {title} {\enquote {\bibinfo {title} {Crystal statistics. i. a two-dimensional model with an order-disorder transition},}\ }\href {\doibase 10.1103/PhysRev.65.117} {\bibfield  {journal} {\bibinfo  {journal} {Phys. Rev.}\ }\textbf {\bibinfo {volume} {65}},\ \bibinfo {pages} {117--149} (\bibinfo {year} {1944})}\BibitemShut {NoStop}%
\bibitem [{\citenamefont {Alicki}\ \emph {et~al.}(2008)\citenamefont {Alicki}, \citenamefont {Horodecki}, \citenamefont {Horodecki},\ and\ \citenamefont {Horodecki}}]{alicki2008thermal}%
  \BibitemOpen
  \bibfield  {author} {\bibinfo {author} {\bibfnamefont {R.}~\bibnamefont {Alicki}}, \bibinfo {author} {\bibfnamefont {M.}~\bibnamefont {Horodecki}}, \bibinfo {author} {\bibfnamefont {P.}~\bibnamefont {Horodecki}}, \ and\ \bibinfo {author} {\bibfnamefont {R.}~\bibnamefont {Horodecki}},\ }\href@noop {} {\enquote {\bibinfo {title} {On thermal stability of topological qubit in kitaev's 4d model},}\ } (\bibinfo {year} {2008}),\ \Eprint {http://arxiv.org/abs/0811.0033} {arXiv:0811.0033 [quant-ph]} \BibitemShut {NoStop}%
\bibitem [{\citenamefont {Yoshida}(2011)}]{Yoshida_2011}%
  \BibitemOpen
  \bibfield  {author} {\bibinfo {author} {\bibfnamefont {Beni}\ \bibnamefont {Yoshida}},\ }\bibfield  {title} {\enquote {\bibinfo {title} {Feasibility of self-correcting quantum memory and thermal stability of topological order},}\ }\href {\doibase 10.1016/j.aop.2011.06.001} {\bibfield  {journal} {\bibinfo  {journal} {Annals of Physics}\ }\textbf {\bibinfo {volume} {326}},\ \bibinfo {pages} {2566–2633} (\bibinfo {year} {2011})}\BibitemShut {NoStop}%
\bibitem [{\citenamefont {Hastings}(2011)}]{Hastings_2011}%
  \BibitemOpen
  \bibfield  {author} {\bibinfo {author} {\bibfnamefont {Matthew~B.}\ \bibnamefont {Hastings}},\ }\bibfield  {title} {\enquote {\bibinfo {title} {Topological order at nonzero temperature},}\ }\href {\doibase 10.1103/physrevlett.107.210501} {\bibfield  {journal} {\bibinfo  {journal} {Physical Review Letters}\ }\textbf {\bibinfo {volume} {107}} (\bibinfo {year} {2011}),\ 10.1103/physrevlett.107.210501}\BibitemShut {NoStop}%
\bibitem [{\citenamefont {Brown}\ \emph {et~al.}(2016)\citenamefont {Brown}, \citenamefont {Loss}, \citenamefont {Pachos}, \citenamefont {Self},\ and\ \citenamefont {Wootton}}]{Brown_2016}%
  \BibitemOpen
  \bibfield  {author} {\bibinfo {author} {\bibfnamefont {Benjamin~J.}\ \bibnamefont {Brown}}, \bibinfo {author} {\bibfnamefont {Daniel}\ \bibnamefont {Loss}}, \bibinfo {author} {\bibfnamefont {Jiannis~K.}\ \bibnamefont {Pachos}}, \bibinfo {author} {\bibfnamefont {Chris~N.}\ \bibnamefont {Self}}, \ and\ \bibinfo {author} {\bibfnamefont {James~R.}\ \bibnamefont {Wootton}},\ }\bibfield  {title} {\enquote {\bibinfo {title} {Quantum memories at finite temperature},}\ }\href {\doibase 10.1103/revmodphys.88.045005} {\bibfield  {journal} {\bibinfo  {journal} {Reviews of Modern Physics}\ }\textbf {\bibinfo {volume} {88}} (\bibinfo {year} {2016}),\ 10.1103/revmodphys.88.045005}\BibitemShut {NoStop}%
\bibitem [{\citenamefont {Liu}\ and\ \citenamefont {Lieu}(2024)}]{lieu}%
  \BibitemOpen
  \bibfield  {author} {\bibinfo {author} {\bibfnamefont {Yu-Jie}\ \bibnamefont {Liu}}\ and\ \bibinfo {author} {\bibfnamefont {Simon}\ \bibnamefont {Lieu}},\ }\bibfield  {title} {\enquote {\bibinfo {title} {Dissipative phase transitions and passive error correction},}\ }\href {\doibase 10.1103/PhysRevA.109.022422} {\bibfield  {journal} {\bibinfo  {journal} {Phys. Rev. A}\ }\textbf {\bibinfo {volume} {109}},\ \bibinfo {pages} {022422} (\bibinfo {year} {2024})}\BibitemShut {NoStop}%
\bibitem [{\citenamefont {Xu}\ \emph {et~al.}(2023)\citenamefont {Xu}, \citenamefont {Ataides}, \citenamefont {Pattison}, \citenamefont {Raveendran}, \citenamefont {Bluvstein}, \citenamefont {Wurtz}, \citenamefont {Vasic}, \citenamefont {Lukin}, \citenamefont {Jiang},\ and\ \citenamefont {Zhou}}]{xu2023}%
  \BibitemOpen
  \bibfield  {author} {\bibinfo {author} {\bibfnamefont {Qian}\ \bibnamefont {Xu}}, \bibinfo {author} {\bibfnamefont {J.~Pablo~Bonilla}\ \bibnamefont {Ataides}}, \bibinfo {author} {\bibfnamefont {Christopher~A.}\ \bibnamefont {Pattison}}, \bibinfo {author} {\bibfnamefont {Nithin}\ \bibnamefont {Raveendran}}, \bibinfo {author} {\bibfnamefont {Dolev}\ \bibnamefont {Bluvstein}}, \bibinfo {author} {\bibfnamefont {Jonathan}\ \bibnamefont {Wurtz}}, \bibinfo {author} {\bibfnamefont {Bane}\ \bibnamefont {Vasic}}, \bibinfo {author} {\bibfnamefont {Mikhail~D.}\ \bibnamefont {Lukin}}, \bibinfo {author} {\bibfnamefont {Liang}\ \bibnamefont {Jiang}}, \ and\ \bibinfo {author} {\bibfnamefont {Hengyun}\ \bibnamefont {Zhou}},\ }\href@noop {} {\enquote {\bibinfo {title} {Constant-overhead fault-tolerant quantum computation with reconfigurable atom arrays},}\ } (\bibinfo {year} {2023}),\ \Eprint {http://arxiv.org/abs/2308.08648} {arXiv:2308.08648 [quant-ph]} \BibitemShut {NoStop}%
\bibitem [{\citenamefont {Hong}\ \emph {et~al.}(2023)\citenamefont {Hong}, \citenamefont {Marinelli}, \citenamefont {Kaufman},\ and\ \citenamefont {Lucas}}]{Hong:2023trf}%
  \BibitemOpen
  \bibfield  {author} {\bibinfo {author} {\bibfnamefont {Yifan}\ \bibnamefont {Hong}}, \bibinfo {author} {\bibfnamefont {Matteo}\ \bibnamefont {Marinelli}}, \bibinfo {author} {\bibfnamefont {Adam~M.}\ \bibnamefont {Kaufman}}, \ and\ \bibinfo {author} {\bibfnamefont {Andrew}\ \bibnamefont {Lucas}},\ }\bibfield  {title} {\enquote {\bibinfo {title} {{Long-range-enhanced surface codes}},}\ }\href@noop {} {\  (\bibinfo {year} {2023})},\ \Eprint {http://arxiv.org/abs/2309.11719} {arXiv:2309.11719 [quant-ph]} \BibitemShut {NoStop}%
\bibitem [{\citenamefont {Gallager}(1962)}]{Gallager_1962}%
  \BibitemOpen
  \bibfield  {author} {\bibinfo {author} {\bibfnamefont {R.}~\bibnamefont {Gallager}},\ }\bibfield  {title} {\enquote {\bibinfo {title} {Low-density parity-check codes},}\ }\href {\doibase 10.1109/TIT.1962.1057683} {\bibfield  {journal} {\bibinfo  {journal} {IRE Transactions on Information Theory}\ }\textbf {\bibinfo {volume} {8}},\ \bibinfo {pages} {21--28} (\bibinfo {year} {1962})}\BibitemShut {NoStop}%
\bibitem [{\citenamefont {Breuckmann}\ and\ \citenamefont {Eberhardt}(2021)}]{Breuckmann_2021}%
  \BibitemOpen
  \bibfield  {author} {\bibinfo {author} {\bibfnamefont {Nikolas~P.}\ \bibnamefont {Breuckmann}}\ and\ \bibinfo {author} {\bibfnamefont {Jens~Niklas}\ \bibnamefont {Eberhardt}},\ }\bibfield  {title} {\enquote {\bibinfo {title} {Quantum low-density parity-check codes},}\ }\href {\doibase 10.1103/prxquantum.2.040101} {\bibfield  {journal} {\bibinfo  {journal} {PRX Quantum}\ }\textbf {\bibinfo {volume} {2}} (\bibinfo {year} {2021}),\ 10.1103/prxquantum.2.040101}\BibitemShut {NoStop}%
\bibitem [{\citenamefont {Bravyi}\ \emph {et~al.}(2010)\citenamefont {Bravyi}, \citenamefont {Poulin},\ and\ \citenamefont {Terhal}}]{Bravyi_2010}%
  \BibitemOpen
  \bibfield  {author} {\bibinfo {author} {\bibfnamefont {Sergey}\ \bibnamefont {Bravyi}}, \bibinfo {author} {\bibfnamefont {David}\ \bibnamefont {Poulin}}, \ and\ \bibinfo {author} {\bibfnamefont {Barbara}\ \bibnamefont {Terhal}},\ }\bibfield  {title} {\enquote {\bibinfo {title} {Tradeoffs for reliable quantum information storage in 2d systems},}\ }\href {\doibase 10.1103/physrevlett.104.050503} {\bibfield  {journal} {\bibinfo  {journal} {Physical Review Letters}\ }\textbf {\bibinfo {volume} {104}} (\bibinfo {year} {2010}),\ 10.1103/physrevlett.104.050503}\BibitemShut {NoStop}%
\bibitem [{\citenamefont {Sipser}\ and\ \citenamefont {Spielman}(1996)}]{Sipser_1996}%
  \BibitemOpen
  \bibfield  {author} {\bibinfo {author} {\bibfnamefont {M.}~\bibnamefont {Sipser}}\ and\ \bibinfo {author} {\bibfnamefont {D.A.}\ \bibnamefont {Spielman}},\ }\bibfield  {title} {\enquote {\bibinfo {title} {Expander codes},}\ }\href {\doibase 10.1109/18.556667} {\bibfield  {journal} {\bibinfo  {journal} {IEEE Transactions on Information Theory}\ }\textbf {\bibinfo {volume} {42}},\ \bibinfo {pages} {1710--1722} (\bibinfo {year} {1996})}\BibitemShut {NoStop}%
\bibitem [{\citenamefont {Metropolis}\ \emph {et~al.}(1953)\citenamefont {Metropolis}, \citenamefont {Rosenbluth}, \citenamefont {Rosenbluth}, \citenamefont {Teller},\ and\ \citenamefont {Teller}}]{Metropolis_1953}%
  \BibitemOpen
  \bibfield  {author} {\bibinfo {author} {\bibfnamefont {Nicholas}\ \bibnamefont {Metropolis}}, \bibinfo {author} {\bibfnamefont {Arianna~W.}\ \bibnamefont {Rosenbluth}}, \bibinfo {author} {\bibfnamefont {Marshall~N.}\ \bibnamefont {Rosenbluth}}, \bibinfo {author} {\bibfnamefont {Augusta~H.}\ \bibnamefont {Teller}}, \ and\ \bibinfo {author} {\bibfnamefont {Edward}\ \bibnamefont {Teller}},\ }\bibfield  {title} {\enquote {\bibinfo {title} {{Equation of State Calculations by Fast Computing Machines}},}\ }\href {\doibase 10.1063/1.1699114} {\bibfield  {journal} {\bibinfo  {journal} {The Journal of Chemical Physics}\ }\textbf {\bibinfo {volume} {21}},\ \bibinfo {pages} {1087--1092} (\bibinfo {year} {1953})},\ \Eprint {http://arxiv.org/abs/https://pubs.aip.org/aip/jcp/article-pdf/21/6/1087/18802390/1087\_1\_online.pdf} {https://pubs.aip.org/aip/jcp/article-pdf/21/6/1087/18802390/1087\_1\_online.pdf} \BibitemShut {NoStop}%
\bibitem [{\citenamefont {Hastings}(1970)}]{Hastings_1970}%
  \BibitemOpen
  \bibfield  {author} {\bibinfo {author} {\bibfnamefont {W.~K.}\ \bibnamefont {Hastings}},\ }\bibfield  {title} {\enquote {\bibinfo {title} {{Monte Carlo sampling methods using Markov chains and their applications}},}\ }\href {\doibase 10.1093/biomet/57.1.97} {\bibfield  {journal} {\bibinfo  {journal} {Biometrika}\ }\textbf {\bibinfo {volume} {57}},\ \bibinfo {pages} {97--109} (\bibinfo {year} {1970})},\ \Eprint {http://arxiv.org/abs/https://academic.oup.com/biomet/article-pdf/57/1/97/23940249/57-1-97.pdf} {https://academic.oup.com/biomet/article-pdf/57/1/97/23940249/57-1-97.pdf} \BibitemShut {NoStop}%
\bibitem [{\citenamefont {Montanari}\ and\ \citenamefont {Semerjian}(2006{\natexlab{a}})}]{Montanari_2006}%
  \BibitemOpen
  \bibfield  {author} {\bibinfo {author} {\bibfnamefont {Andrea}\ \bibnamefont {Montanari}}\ and\ \bibinfo {author} {\bibfnamefont {Guilhem}\ \bibnamefont {Semerjian}},\ }\bibfield  {title} {\enquote {\bibinfo {title} {On the dynamics of the glass transition on bethe lattices},}\ }\href {\doibase 10.1007/s10955-006-9103-1} {\bibfield  {journal} {\bibinfo  {journal} {Journal of Statistical Physics}\ }\textbf {\bibinfo {volume} {124}},\ \bibinfo {pages} {103–189} (\bibinfo {year} {2006}{\natexlab{a}})}\BibitemShut {NoStop}%
\bibitem [{\citenamefont {Kirkpatrick}\ and\ \citenamefont {Thirumalai}(1987)}]{Kirkpatrick_1987}%
  \BibitemOpen
  \bibfield  {author} {\bibinfo {author} {\bibfnamefont {T.~R.}\ \bibnamefont {Kirkpatrick}}\ and\ \bibinfo {author} {\bibfnamefont {D.}~\bibnamefont {Thirumalai}},\ }\bibfield  {title} {\enquote {\bibinfo {title} {Dynamics of the structural glass transition and the $p$-spin---interaction spin-glass model},}\ }\href {\doibase 10.1103/PhysRevLett.58.2091} {\bibfield  {journal} {\bibinfo  {journal} {Phys. Rev. Lett.}\ }\textbf {\bibinfo {volume} {58}},\ \bibinfo {pages} {2091--2094} (\bibinfo {year} {1987})}\BibitemShut {NoStop}%
\bibitem [{\citenamefont {Montanari}\ and\ \citenamefont {Semerjian}(2006{\natexlab{b}})}]{Montanari_2006_2}%
  \BibitemOpen
  \bibfield  {author} {\bibinfo {author} {\bibfnamefont {Andrea}\ \bibnamefont {Montanari}}\ and\ \bibinfo {author} {\bibfnamefont {Guilhem}\ \bibnamefont {Semerjian}},\ }\bibfield  {title} {\enquote {\bibinfo {title} {Rigorous inequalities between length and time scales in glassy systems},}\ }\href {\doibase 10.1007/s10955-006-9175-y} {\bibfield  {journal} {\bibinfo  {journal} {Journal of Statistical Physics}\ }\textbf {\bibinfo {volume} {125}},\ \bibinfo {pages} {23–54} (\bibinfo {year} {2006}{\natexlab{b}})}\BibitemShut {NoStop}%
\bibitem [{\citenamefont {Richardson}\ and\ \citenamefont {Urbanke}(2008)}]{ModernCodingTheory}%
  \BibitemOpen
  \bibfield  {author} {\bibinfo {author} {\bibfnamefont {Tom}\ \bibnamefont {Richardson}}\ and\ \bibinfo {author} {\bibfnamefont {Ruediger}\ \bibnamefont {Urbanke}},\ }\href@noop {} {\emph {\bibinfo {title} {Modern coding theory}}}\ (\bibinfo  {publisher} {Cambridge University Press},\ \bibinfo {year} {2008})\BibitemShut {NoStop}%
\bibitem [{\citenamefont {Freedman}\ and\ \citenamefont {Hastings}(2013)}]{freedman2013}%
  \BibitemOpen
  \bibfield  {author} {\bibinfo {author} {\bibfnamefont {M.~H.}\ \bibnamefont {Freedman}}\ and\ \bibinfo {author} {\bibfnamefont {M.~B.}\ \bibnamefont {Hastings}},\ }\href@noop {} {\enquote {\bibinfo {title} {Quantum systems on non-$k$-hyperfinite complexes: A generalization of classical statistical mechanics on expander graphs},}\ } (\bibinfo {year} {2013}),\ \Eprint {http://arxiv.org/abs/1301.1363} {arXiv:1301.1363 [quant-ph]} \BibitemShut {NoStop}%
\bibitem [{\citenamefont {Weinstein}\ \emph {et~al.}(2019)\citenamefont {Weinstein}, \citenamefont {Ortiz},\ and\ \citenamefont {Nussinov}}]{Weinstein_2019}%
  \BibitemOpen
  \bibfield  {author} {\bibinfo {author} {\bibfnamefont {Zack}\ \bibnamefont {Weinstein}}, \bibinfo {author} {\bibfnamefont {Gerardo}\ \bibnamefont {Ortiz}}, \ and\ \bibinfo {author} {\bibfnamefont {Zohar}\ \bibnamefont {Nussinov}},\ }\bibfield  {title} {\enquote {\bibinfo {title} {Universality classes of stabilizer code hamiltonians},}\ }\href {\doibase 10.1103/physrevlett.123.230503} {\bibfield  {journal} {\bibinfo  {journal} {Physical Review Letters}\ }\textbf {\bibinfo {volume} {123}} (\bibinfo {year} {2019}),\ 10.1103/physrevlett.123.230503}\BibitemShut {NoStop}%
\bibitem [{\citenamefont {Rakovszky}\ and\ \citenamefont {Khemani}(2023)}]{Rakovszky:2023fng}%
  \BibitemOpen
  \bibfield  {author} {\bibinfo {author} {\bibfnamefont {Tibor}\ \bibnamefont {Rakovszky}}\ and\ \bibinfo {author} {\bibfnamefont {Vedika}\ \bibnamefont {Khemani}},\ }\bibfield  {title} {\enquote {\bibinfo {title} {{The Physics of (good) LDPC Codes I. Gauging and dualities}},}\ }\href@noop {} {\  (\bibinfo {year} {2023})},\ \Eprint {http://arxiv.org/abs/2310.16032} {arXiv:2310.16032 [quant-ph]} \BibitemShut {NoStop}%
\bibitem [{\citenamefont {Mézard}\ and\ \citenamefont {Montanari}(2009)}]{Montanari_book}%
  \BibitemOpen
  \bibfield  {author} {\bibinfo {author} {\bibfnamefont {Marc}\ \bibnamefont {Mézard}}\ and\ \bibinfo {author} {\bibfnamefont {Andrea}\ \bibnamefont {Montanari}},\ }\href {\doibase 10.1093/acprof:oso/9780198570837.001.0001} {\emph {\bibinfo {title} {{Information, Physics, and Computation}}}}\ (\bibinfo  {publisher} {Oxford University Press},\ \bibinfo {year} {2009})\BibitemShut {NoStop}%
\bibitem [{\citenamefont {Rakovszky}\ and\ \citenamefont {Khemani}(2024)}]{Rakovszky:2024iks}%
  \BibitemOpen
  \bibfield  {author} {\bibinfo {author} {\bibfnamefont {Tibor}\ \bibnamefont {Rakovszky}}\ and\ \bibinfo {author} {\bibfnamefont {Vedika}\ \bibnamefont {Khemani}},\ }\bibfield  {title} {\enquote {\bibinfo {title} {{The Physics of (good) LDPC Codes II. Product constructions}},}\ }\href@noop {} {\  (\bibinfo {year} {2024})},\ \Eprint {http://arxiv.org/abs/2402.16831} {arXiv:2402.16831 [quant-ph]} \BibitemShut {NoStop}%
\bibitem [{\citenamefont {Ben-Sasson}\ \emph {et~al.}(2009)\citenamefont {Ben-Sasson}, \citenamefont {Guruswami}, \citenamefont {Kaufman}, \citenamefont {Sudan},\ and\ \citenamefont {Viderman}}]{Sasson_2009}%
  \BibitemOpen
  \bibfield  {author} {\bibinfo {author} {\bibfnamefont {Eli}\ \bibnamefont {Ben-Sasson}}, \bibinfo {author} {\bibfnamefont {Venkatesan}\ \bibnamefont {Guruswami}}, \bibinfo {author} {\bibfnamefont {Tali}\ \bibnamefont {Kaufman}}, \bibinfo {author} {\bibfnamefont {Madhu}\ \bibnamefont {Sudan}}, \ and\ \bibinfo {author} {\bibfnamefont {Michael}\ \bibnamefont {Viderman}},\ }\bibfield  {title} {\enquote {\bibinfo {title} {Locally testable codes require redundant testers},}\ }in\ \href {\doibase 10.1109/CCC.2009.6} {\emph {\bibinfo {booktitle} {2009 24th Annual IEEE Conference on Computational Complexity}}}\ (\bibinfo {year} {2009})\ pp.\ \bibinfo {pages} {52--61}\BibitemShut {NoStop}%
\bibitem [{\citenamefont {Panteleev}\ and\ \citenamefont {Kalachev}(2022)}]{Panteleev_2022}%
  \BibitemOpen
  \bibfield  {author} {\bibinfo {author} {\bibfnamefont {Pavel}\ \bibnamefont {Panteleev}}\ and\ \bibinfo {author} {\bibfnamefont {Gleb}\ \bibnamefont {Kalachev}},\ }\bibfield  {title} {\enquote {\bibinfo {title} {Asymptotically good quantum and locally testable classical ldpc codes},}\ }in\ \href {\doibase 10.1145/3519935.3520017} {\emph {\bibinfo {booktitle} {Proceedings of the 54th Annual ACM SIGACT Symposium on Theory of Computing}}},\ \bibinfo {series and number} {STOC 2022}\ (\bibinfo  {publisher} {Association for Computing Machinery},\ \bibinfo {address} {New York, NY, USA},\ \bibinfo {year} {2022})\ p.\ \bibinfo {pages} {375–388}\BibitemShut {NoStop}%
\bibitem [{\citenamefont {Lin}\ and\ \citenamefont {Hsieh}(2022)}]{lin2022}%
  \BibitemOpen
  \bibfield  {author} {\bibinfo {author} {\bibfnamefont {Ting-Chun}\ \bibnamefont {Lin}}\ and\ \bibinfo {author} {\bibfnamefont {Min-Hsiu}\ \bibnamefont {Hsieh}},\ }\href@noop {} {\enquote {\bibinfo {title} {$c^3$-locally testable codes from lossless expanders},}\ } (\bibinfo {year} {2022}),\ \Eprint {http://arxiv.org/abs/2201.11369} {arXiv:2201.11369 [cs.IT]} \BibitemShut {NoStop}%
\bibitem [{\citenamefont {Dinur}\ \emph {et~al.}(2022{\natexlab{a}})\citenamefont {Dinur}, \citenamefont {Evra}, \citenamefont {Livne}, \citenamefont {Lubotzky},\ and\ \citenamefont {Mozes}}]{dinur2022LTC}%
  \BibitemOpen
  \bibfield  {author} {\bibinfo {author} {\bibfnamefont {Irit}\ \bibnamefont {Dinur}}, \bibinfo {author} {\bibfnamefont {Shai}\ \bibnamefont {Evra}}, \bibinfo {author} {\bibfnamefont {Ron}\ \bibnamefont {Livne}}, \bibinfo {author} {\bibfnamefont {Alexander}\ \bibnamefont {Lubotzky}}, \ and\ \bibinfo {author} {\bibfnamefont {Shahar}\ \bibnamefont {Mozes}},\ }\href@noop {} {\enquote {\bibinfo {title} {Good locally testable codes},}\ } (\bibinfo {year} {2022}{\natexlab{a}}),\ \Eprint {http://arxiv.org/abs/2207.11929} {arXiv:2207.11929 [cs.IT]} \BibitemShut {NoStop}%
\bibitem [{\citenamefont {Calderbank}\ and\ \citenamefont {Shor}(1996)}]{Calderbank_1996}%
  \BibitemOpen
  \bibfield  {author} {\bibinfo {author} {\bibfnamefont {A.~R.}\ \bibnamefont {Calderbank}}\ and\ \bibinfo {author} {\bibfnamefont {Peter~W.}\ \bibnamefont {Shor}},\ }\bibfield  {title} {\enquote {\bibinfo {title} {Good quantum error-correcting codes exist},}\ }\href {\doibase 10.1103/physreva.54.1098} {\bibfield  {journal} {\bibinfo  {journal} {Physical Review A}\ }\textbf {\bibinfo {volume} {54}},\ \bibinfo {pages} {1098--1105} (\bibinfo {year} {1996})}\BibitemShut {NoStop}%
\bibitem [{\citenamefont {Steane}(1996)}]{Steane_1996}%
  \BibitemOpen
  \bibfield  {author} {\bibinfo {author} {\bibfnamefont {Andrew}\ \bibnamefont {Steane}},\ }\bibfield  {title} {\enquote {\bibinfo {title} {Multiple-particle interference and quantum error correction},}\ }\href {\doibase 10.1098/rspa.1996.0136} {\bibfield  {journal} {\bibinfo  {journal} {Proceedings of the Royal Society of London. Series A: Mathematical, Physical and Engineering Sciences}\ }\textbf {\bibinfo {volume} {452}},\ \bibinfo {pages} {2551--2577} (\bibinfo {year} {1996})}\BibitemShut {NoStop}%
\bibitem [{\citenamefont {Leverrier}\ and\ \citenamefont {Zémor}(2022)}]{qTanner_codes}%
  \BibitemOpen
  \bibfield  {author} {\bibinfo {author} {\bibfnamefont {Anthony}\ \bibnamefont {Leverrier}}\ and\ \bibinfo {author} {\bibfnamefont {Gilles}\ \bibnamefont {Zémor}},\ }\href@noop {} {\enquote {\bibinfo {title} {Quantum tanner codes},}\ } (\bibinfo {year} {2022}),\ \Eprint {http://arxiv.org/abs/2202.13641} {arXiv:2202.13641 [quant-ph]} \BibitemShut {NoStop}%
\bibitem [{\citenamefont {Dinur}\ \emph {et~al.}(2022{\natexlab{b}})\citenamefont {Dinur}, \citenamefont {Hsieh}, \citenamefont {Lin},\ and\ \citenamefont {Vidick}}]{dinur2022good}%
  \BibitemOpen
  \bibfield  {author} {\bibinfo {author} {\bibfnamefont {Irit}\ \bibnamefont {Dinur}}, \bibinfo {author} {\bibfnamefont {Min-Hsiu}\ \bibnamefont {Hsieh}}, \bibinfo {author} {\bibfnamefont {Ting-Chun}\ \bibnamefont {Lin}}, \ and\ \bibinfo {author} {\bibfnamefont {Thomas}\ \bibnamefont {Vidick}},\ }\href@noop {} {\enquote {\bibinfo {title} {Good quantum ldpc codes with linear time decoders},}\ } (\bibinfo {year} {2022}{\natexlab{b}}),\ \Eprint {http://arxiv.org/abs/2206.07750} {arXiv:2206.07750 [quant-ph]} \BibitemShut {NoStop}%
\bibitem [{\citenamefont {Tillich}\ and\ \citenamefont {Zemor}(2014)}]{HGP}%
  \BibitemOpen
  \bibfield  {author} {\bibinfo {author} {\bibfnamefont {Jean-Pierre}\ \bibnamefont {Tillich}}\ and\ \bibinfo {author} {\bibfnamefont {Gilles}\ \bibnamefont {Zemor}},\ }\bibfield  {title} {\enquote {\bibinfo {title} {Quantum {LDPC} codes with positive rate and minimum distance proportional to the square root of the blocklength},}\ }\href {\doibase 10.1109/tit.2013.2292061} {\bibfield  {journal} {\bibinfo  {journal} {{IEEE} Transactions on Information Theory}\ }\textbf {\bibinfo {volume} {60}},\ \bibinfo {pages} {1193--1202} (\bibinfo {year} {2014})}\BibitemShut {NoStop}%
\bibitem [{\citenamefont {Leverrier}\ \emph {et~al.}(2015)\citenamefont {Leverrier}, \citenamefont {Tillich},\ and\ \citenamefont {Zémor}}]{Leverrier_2015}%
  \BibitemOpen
  \bibfield  {author} {\bibinfo {author} {\bibfnamefont {Anthony}\ \bibnamefont {Leverrier}}, \bibinfo {author} {\bibfnamefont {Jean-Pierre}\ \bibnamefont {Tillich}}, \ and\ \bibinfo {author} {\bibfnamefont {Gilles}\ \bibnamefont {Zémor}},\ }\bibfield  {title} {\enquote {\bibinfo {title} {Quantum expander codes},}\ }in\ \href {\doibase 10.1109/focs.2015.55} {\emph {\bibinfo {booktitle} {2015 IEEE 56th Annual Symposium on Foundations of Computer Science}}}\ (\bibinfo  {publisher} {IEEE},\ \bibinfo {year} {2015})\BibitemShut {NoStop}%
\bibitem [{\citenamefont {Peierls}(1936)}]{Peierls_1936}%
  \BibitemOpen
  \bibfield  {author} {\bibinfo {author} {\bibfnamefont {R.}~\bibnamefont {Peierls}},\ }\bibfield  {title} {\enquote {\bibinfo {title} {On ising’s model of ferromagnetism},}\ }\href {\doibase 10.1017/S0305004100019174} {\bibfield  {journal} {\bibinfo  {journal} {Mathematical Proceedings of the Cambridge Philosophical Society}\ }\textbf {\bibinfo {volume} {32}},\ \bibinfo {pages} {477–481} (\bibinfo {year} {1936})}\BibitemShut {NoStop}%
\bibitem [{\citenamefont {Hastings}\ \emph {et~al.}(2014)\citenamefont {Hastings}, \citenamefont {Watson},\ and\ \citenamefont {Melko}}]{melko}%
  \BibitemOpen
  \bibfield  {author} {\bibinfo {author} {\bibfnamefont {Matthew~B.}\ \bibnamefont {Hastings}}, \bibinfo {author} {\bibfnamefont {Grant~H.}\ \bibnamefont {Watson}}, \ and\ \bibinfo {author} {\bibfnamefont {Roger~G.}\ \bibnamefont {Melko}},\ }\bibfield  {title} {\enquote {\bibinfo {title} {Self-correcting quantum memories beyond the percolation threshold},}\ }\href {\doibase 10.1103/PhysRevLett.112.070501} {\bibfield  {journal} {\bibinfo  {journal} {Phys. Rev. Lett.}\ }\textbf {\bibinfo {volume} {112}},\ \bibinfo {pages} {070501} (\bibinfo {year} {2014})}\BibitemShut {NoStop}%
\bibitem [{\citenamefont {Haah}(2011)}]{haah1}%
  \BibitemOpen
  \bibfield  {author} {\bibinfo {author} {\bibfnamefont {Jeongwan}\ \bibnamefont {Haah}},\ }\bibfield  {title} {\enquote {\bibinfo {title} {Local stabilizer codes in three dimensions without string logical operators},}\ }\href {\doibase 10.1103/PhysRevA.83.042330} {\bibfield  {journal} {\bibinfo  {journal} {Phys. Rev. A}\ }\textbf {\bibinfo {volume} {83}},\ \bibinfo {pages} {042330} (\bibinfo {year} {2011})}\BibitemShut {NoStop}%
\bibitem [{\citenamefont {Michnicki}(2014)}]{Michnicki_2014}%
  \BibitemOpen
  \bibfield  {author} {\bibinfo {author} {\bibfnamefont {Kamil~P.}\ \bibnamefont {Michnicki}},\ }\bibfield  {title} {\enquote {\bibinfo {title} {3d topological quantum memory with a power-law energy barrier},}\ }\href {\doibase 10.1103/PhysRevLett.113.130501} {\bibfield  {journal} {\bibinfo  {journal} {Phys. Rev. Lett.}\ }\textbf {\bibinfo {volume} {113}},\ \bibinfo {pages} {130501} (\bibinfo {year} {2014})}\BibitemShut {NoStop}%
\bibitem [{\citenamefont {Bravyi}\ and\ \citenamefont {Haah}(2013)}]{haah2}%
  \BibitemOpen
  \bibfield  {author} {\bibinfo {author} {\bibfnamefont {Sergey}\ \bibnamefont {Bravyi}}\ and\ \bibinfo {author} {\bibfnamefont {Jeongwan}\ \bibnamefont {Haah}},\ }\bibfield  {title} {\enquote {\bibinfo {title} {Quantum self-correction in the 3d cubic code model},}\ }\href {\doibase 10.1103/PhysRevLett.111.200501} {\bibfield  {journal} {\bibinfo  {journal} {Phys. Rev. Lett.}\ }\textbf {\bibinfo {volume} {111}},\ \bibinfo {pages} {200501} (\bibinfo {year} {2013})}\BibitemShut {NoStop}%
\bibitem [{\citenamefont {Siva}\ and\ \citenamefont {Yoshida}(2017)}]{Siva_2017}%
  \BibitemOpen
  \bibfield  {author} {\bibinfo {author} {\bibfnamefont {Karthik}\ \bibnamefont {Siva}}\ and\ \bibinfo {author} {\bibfnamefont {Beni}\ \bibnamefont {Yoshida}},\ }\bibfield  {title} {\enquote {\bibinfo {title} {Topological order and memory time in marginally-self-correcting quantum memory},}\ }\href {\doibase 10.1103/physreva.95.032324} {\bibfield  {journal} {\bibinfo  {journal} {Physical Review A}\ }\textbf {\bibinfo {volume} {95}} (\bibinfo {year} {2017}),\ 10.1103/physreva.95.032324}\BibitemShut {NoStop}%
\bibitem [{\citenamefont {Ahn}\ \emph {et~al.}(2002)\citenamefont {Ahn}, \citenamefont {Doherty},\ and\ \citenamefont {Landahl}}]{Ahn_2002}%
  \BibitemOpen
  \bibfield  {author} {\bibinfo {author} {\bibfnamefont {Charlene}\ \bibnamefont {Ahn}}, \bibinfo {author} {\bibfnamefont {Andrew~C.}\ \bibnamefont {Doherty}}, \ and\ \bibinfo {author} {\bibfnamefont {Andrew~J.}\ \bibnamefont {Landahl}},\ }\bibfield  {title} {\enquote {\bibinfo {title} {Continuous quantum error correction via quantum feedback control},}\ }\href {\doibase 10.1103/PhysRevA.65.042301} {\bibfield  {journal} {\bibinfo  {journal} {Phys. Rev. A}\ }\textbf {\bibinfo {volume} {65}},\ \bibinfo {pages} {042301} (\bibinfo {year} {2002})}\BibitemShut {NoStop}%
\bibitem [{\citenamefont {Sarovar}\ and\ \citenamefont {Milburn}(2005)}]{Sarovar_2005}%
  \BibitemOpen
  \bibfield  {author} {\bibinfo {author} {\bibfnamefont {Mohan}\ \bibnamefont {Sarovar}}\ and\ \bibinfo {author} {\bibfnamefont {G.~J.}\ \bibnamefont {Milburn}},\ }\bibfield  {title} {\enquote {\bibinfo {title} {Continuous quantum error correction by cooling},}\ }\href {\doibase 10.1103/PhysRevA.72.012306} {\bibfield  {journal} {\bibinfo  {journal} {Phys. Rev. A}\ }\textbf {\bibinfo {volume} {72}},\ \bibinfo {pages} {012306} (\bibinfo {year} {2005})}\BibitemShut {NoStop}%
\bibitem [{\citenamefont {Pastawski}\ \emph {et~al.}(2011)\citenamefont {Pastawski}, \citenamefont {Clemente},\ and\ \citenamefont {Cirac}}]{Pastawski_2011}%
  \BibitemOpen
  \bibfield  {author} {\bibinfo {author} {\bibfnamefont {Fernando}\ \bibnamefont {Pastawski}}, \bibinfo {author} {\bibfnamefont {Lucas}\ \bibnamefont {Clemente}}, \ and\ \bibinfo {author} {\bibfnamefont {Juan~Ignacio}\ \bibnamefont {Cirac}},\ }\bibfield  {title} {\enquote {\bibinfo {title} {Quantum memories based on engineered dissipation},}\ }\href {\doibase 10.1103/physreva.83.012304} {\bibfield  {journal} {\bibinfo  {journal} {Physical Review A}\ }\textbf {\bibinfo {volume} {83}} (\bibinfo {year} {2011}),\ 10.1103/physreva.83.012304}\BibitemShut {NoStop}%
\bibitem [{\citenamefont {Stinespring}(1955)}]{Stinespring}%
  \BibitemOpen
  \bibfield  {author} {\bibinfo {author} {\bibfnamefont {W.~Forrest}\ \bibnamefont {Stinespring}},\ }\bibfield  {title} {\enquote {\bibinfo {title} {Positive functions on c*-algebras},}\ }\href@noop {} {\bibfield  {journal} {\bibinfo  {journal} {Proceedings of the American Mathematical Society}\ }\textbf {\bibinfo {volume} {6}},\ \bibinfo {pages} {211--216} (\bibinfo {year} {1955})}\BibitemShut {NoStop}%
\bibitem [{\citenamefont {Friedman}\ \emph {et~al.}(2022)\citenamefont {Friedman}, \citenamefont {Yin}, \citenamefont {Hong},\ and\ \citenamefont {Lucas}}]{Friedman:2022vqb}%
  \BibitemOpen
  \bibfield  {author} {\bibinfo {author} {\bibfnamefont {Aaron~J.}\ \bibnamefont {Friedman}}, \bibinfo {author} {\bibfnamefont {Chao}\ \bibnamefont {Yin}}, \bibinfo {author} {\bibfnamefont {Yifan}\ \bibnamefont {Hong}}, \ and\ \bibinfo {author} {\bibfnamefont {Andrew}\ \bibnamefont {Lucas}},\ }\bibfield  {title} {\enquote {\bibinfo {title} {{Locality and error correction in quantum dynamics with measurement}},}\ }\href@noop {} {\  (\bibinfo {year} {2022})},\ \Eprint {http://arxiv.org/abs/2206.09929} {arXiv:2206.09929 [quant-ph]} \BibitemShut {NoStop}%
\bibitem [{\citenamefont {Bombin}(2015)}]{Bombin_2015}%
  \BibitemOpen
  \bibfield  {author} {\bibinfo {author} {\bibfnamefont {H\'ector}\ \bibnamefont {Bombin}},\ }\bibfield  {title} {\enquote {\bibinfo {title} {Single-shot fault-tolerant quantum error correction},}\ }\href {\doibase 10.1103/physrevx.5.031043} {\bibfield  {journal} {\bibinfo  {journal} {Physical Review X}\ }\textbf {\bibinfo {volume} {5}} (\bibinfo {year} {2015}),\ 10.1103/physrevx.5.031043}\BibitemShut {NoStop}%
\bibitem [{\citenamefont {Campbell}(2019)}]{Campbell_2019}%
  \BibitemOpen
  \bibfield  {author} {\bibinfo {author} {\bibfnamefont {Earl~T}\ \bibnamefont {Campbell}},\ }\bibfield  {title} {\enquote {\bibinfo {title} {A theory of single-shot error correction for adversarial noise},}\ }\href {\doibase 10.1088/2058-9565/aafc8f} {\bibfield  {journal} {\bibinfo  {journal} {Quantum Science and Technology}\ }\textbf {\bibinfo {volume} {4}},\ \bibinfo {pages} {025006} (\bibinfo {year} {2019})}\BibitemShut {NoStop}%
\bibitem [{\citenamefont {Quintavalle}\ \emph {et~al.}(2021)\citenamefont {Quintavalle}, \citenamefont {Vasmer}, \citenamefont {Roffe},\ and\ \citenamefont {Campbell}}]{Quintavalle_2021}%
  \BibitemOpen
  \bibfield  {author} {\bibinfo {author} {\bibfnamefont {Armanda~O.}\ \bibnamefont {Quintavalle}}, \bibinfo {author} {\bibfnamefont {Michael}\ \bibnamefont {Vasmer}}, \bibinfo {author} {\bibfnamefont {Joschka}\ \bibnamefont {Roffe}}, \ and\ \bibinfo {author} {\bibfnamefont {Earl~T.}\ \bibnamefont {Campbell}},\ }\bibfield  {title} {\enquote {\bibinfo {title} {Single-shot error correction of three-dimensional homological product codes},}\ }\href {\doibase 10.1103/PRXQuantum.2.020340} {\bibfield  {journal} {\bibinfo  {journal} {PRX Quantum}\ }\textbf {\bibinfo {volume} {2}},\ \bibinfo {pages} {020340} (\bibinfo {year} {2021})}\BibitemShut {NoStop}%
\bibitem [{\citenamefont {Higgott}\ and\ \citenamefont {Breuckmann}(2023)}]{Higgott_2023}%
  \BibitemOpen
  \bibfield  {author} {\bibinfo {author} {\bibfnamefont {Oscar}\ \bibnamefont {Higgott}}\ and\ \bibinfo {author} {\bibfnamefont {Nikolas~P.}\ \bibnamefont {Breuckmann}},\ }\bibfield  {title} {\enquote {\bibinfo {title} {Improved single-shot decoding of higher-dimensional hypergraph-product codes},}\ }\href {\doibase 10.1103/PRXQuantum.4.020332} {\bibfield  {journal} {\bibinfo  {journal} {PRX Quantum}\ }\textbf {\bibinfo {volume} {4}},\ \bibinfo {pages} {020332} (\bibinfo {year} {2023})}\BibitemShut {NoStop}%
\bibitem [{\citenamefont {Fawzi}\ \emph {et~al.}(2018)\citenamefont {Fawzi}, \citenamefont {Grospellier},\ and\ \citenamefont {Leverrier}}]{Fawzi_2018}%
  \BibitemOpen
  \bibfield  {author} {\bibinfo {author} {\bibfnamefont {Omar}\ \bibnamefont {Fawzi}}, \bibinfo {author} {\bibfnamefont {Antoine}\ \bibnamefont {Grospellier}}, \ and\ \bibinfo {author} {\bibfnamefont {Anthony}\ \bibnamefont {Leverrier}},\ }\bibfield  {title} {\enquote {\bibinfo {title} {Constant overhead quantum fault-tolerance with quantum expander codes},}\ }in\ \href {\doibase 10.1109/focs.2018.00076} {\emph {\bibinfo {booktitle} {2018 IEEE 59th Annual Symposium on Foundations of Computer Science (FOCS)}}}\ (\bibinfo  {publisher} {IEEE},\ \bibinfo {year} {2018})\BibitemShut {NoStop}%
\bibitem [{\citenamefont {Gu}\ \emph {et~al.}(2023)\citenamefont {Gu}, \citenamefont {Tang}, \citenamefont {Caha}, \citenamefont {Choe}, \citenamefont {He},\ and\ \citenamefont {Kubica}}]{Gu:2023pzw}%
  \BibitemOpen
  \bibfield  {author} {\bibinfo {author} {\bibfnamefont {Shouzhen}\ \bibnamefont {Gu}}, \bibinfo {author} {\bibfnamefont {Eugene}\ \bibnamefont {Tang}}, \bibinfo {author} {\bibfnamefont {Libor}\ \bibnamefont {Caha}}, \bibinfo {author} {\bibfnamefont {Shin~Ho}\ \bibnamefont {Choe}}, \bibinfo {author} {\bibfnamefont {Zhiyang}\ \bibnamefont {He}}, \ and\ \bibinfo {author} {\bibfnamefont {Aleksander}\ \bibnamefont {Kubica}},\ }\bibfield  {title} {\enquote {\bibinfo {title} {{Single-shot decoding of good quantum LDPC codes}},}\ }\href@noop {} {\  (\bibinfo {year} {2023})},\ \Eprint {http://arxiv.org/abs/2306.12470} {arXiv:2306.12470 [quant-ph]} \BibitemShut {NoStop}%
\bibitem [{\citenamefont {Saffman}\ \emph {et~al.}(2010)\citenamefont {Saffman}, \citenamefont {Walker},\ and\ \citenamefont {M{\o}lmer}}]{Saffman2010}%
  \BibitemOpen
  \bibfield  {author} {\bibinfo {author} {\bibfnamefont {M.}~\bibnamefont {Saffman}}, \bibinfo {author} {\bibfnamefont {T.~G.}\ \bibnamefont {Walker}}, \ and\ \bibinfo {author} {\bibfnamefont {K.}~\bibnamefont {M{\o}lmer}},\ }\bibfield  {title} {\enquote {\bibinfo {title} {{Quantum information with Rydberg atoms}},}\ }\href {\doibase 10.1103/RevModPhys.82.2313} {\bibfield  {journal} {\bibinfo  {journal} {Rev. Mod. Phys.}\ }\textbf {\bibinfo {volume} {82}},\ \bibinfo {pages} {2313--2363} (\bibinfo {year} {2010})}\BibitemShut {NoStop}%
\bibitem [{\citenamefont {Kaufman}\ and\ \citenamefont {Ni}(2021)}]{Kaufman_2021}%
  \BibitemOpen
  \bibfield  {author} {\bibinfo {author} {\bibfnamefont {Adam~M}\ \bibnamefont {Kaufman}}\ and\ \bibinfo {author} {\bibfnamefont {Kang-Kuen}\ \bibnamefont {Ni}},\ }\bibfield  {title} {\enquote {\bibinfo {title} {Quantum science with optical tweezer arrays of ultracold atoms and molecules},}\ }\href@noop {} {\bibfield  {journal} {\bibinfo  {journal} {Nature Physics}\ }\textbf {\bibinfo {volume} {17}},\ \bibinfo {pages} {1324--1333} (\bibinfo {year} {2021})}\BibitemShut {NoStop}%
\bibitem [{\citenamefont {Cong}\ \emph {et~al.}(2022)\citenamefont {Cong}, \citenamefont {Levine}, \citenamefont {Keesling}, \citenamefont {Bluvstein}, \citenamefont {Wang},\ and\ \citenamefont {Lukin}}]{Cong_2022}%
  \BibitemOpen
  \bibfield  {author} {\bibinfo {author} {\bibfnamefont {Iris}\ \bibnamefont {Cong}}, \bibinfo {author} {\bibfnamefont {Harry}\ \bibnamefont {Levine}}, \bibinfo {author} {\bibfnamefont {Alexander}\ \bibnamefont {Keesling}}, \bibinfo {author} {\bibfnamefont {Dolev}\ \bibnamefont {Bluvstein}}, \bibinfo {author} {\bibfnamefont {Sheng-Tao}\ \bibnamefont {Wang}}, \ and\ \bibinfo {author} {\bibfnamefont {Mikhail~D}\ \bibnamefont {Lukin}},\ }\bibfield  {title} {\enquote {\bibinfo {title} {Hardware-efficient, fault-tolerant quantum computation with rydberg atoms},}\ }\href@noop {} {\bibfield  {journal} {\bibinfo  {journal} {Physical Review X}\ }\textbf {\bibinfo {volume} {12}},\ \bibinfo {pages} {021049} (\bibinfo {year} {2022})}\BibitemShut {NoStop}%
\bibitem [{\citenamefont {Bluvstein}\ \emph {et~al.}(2022)\citenamefont {Bluvstein}, \citenamefont {Levine}, \citenamefont {Semeghini}, \citenamefont {Wang}, \citenamefont {Ebadi}, \citenamefont {Kalinowski}, \citenamefont {Keesling}, \citenamefont {Maskara}, \citenamefont {Pichler}, \citenamefont {Greiner} \emph {et~al.}}]{Bluvstein_2022}%
  \BibitemOpen
  \bibfield  {author} {\bibinfo {author} {\bibfnamefont {Dolev}\ \bibnamefont {Bluvstein}}, \bibinfo {author} {\bibfnamefont {Harry}\ \bibnamefont {Levine}}, \bibinfo {author} {\bibfnamefont {Giulia}\ \bibnamefont {Semeghini}}, \bibinfo {author} {\bibfnamefont {Tout~T}\ \bibnamefont {Wang}}, \bibinfo {author} {\bibfnamefont {Sepehr}\ \bibnamefont {Ebadi}}, \bibinfo {author} {\bibfnamefont {Marcin}\ \bibnamefont {Kalinowski}}, \bibinfo {author} {\bibfnamefont {Alexander}\ \bibnamefont {Keesling}}, \bibinfo {author} {\bibfnamefont {Nishad}\ \bibnamefont {Maskara}}, \bibinfo {author} {\bibfnamefont {Hannes}\ \bibnamefont {Pichler}}, \bibinfo {author} {\bibfnamefont {Markus}\ \bibnamefont {Greiner}},  \emph {et~al.},\ }\bibfield  {title} {\enquote {\bibinfo {title} {A quantum processor based on coherent transport of entangled atom arrays},}\ }\href@noop {} {\bibfield  {journal} {\bibinfo  {journal} {Nature}\ }\textbf {\bibinfo {volume} {604}},\ \bibinfo {pages} {451--456} (\bibinfo {year} {2022})}\BibitemShut
  {NoStop}%
\bibitem [{\citenamefont {Jenkins}\ \emph {et~al.}(2022)\citenamefont {Jenkins}, \citenamefont {Lis}, \citenamefont {Senoo}, \citenamefont {McGrew},\ and\ \citenamefont {Kaufman}}]{Jenkins_2022}%
  \BibitemOpen
  \bibfield  {author} {\bibinfo {author} {\bibfnamefont {Alec}\ \bibnamefont {Jenkins}}, \bibinfo {author} {\bibfnamefont {Joanna~W.}\ \bibnamefont {Lis}}, \bibinfo {author} {\bibfnamefont {Aruku}\ \bibnamefont {Senoo}}, \bibinfo {author} {\bibfnamefont {William~F.}\ \bibnamefont {McGrew}}, \ and\ \bibinfo {author} {\bibfnamefont {Adam~M.}\ \bibnamefont {Kaufman}},\ }\bibfield  {title} {\enquote {\bibinfo {title} {Ytterbium nuclear-spin qubits in an optical tweezer array},}\ }\href {\doibase 10.1103/PhysRevX.12.021027} {\bibfield  {journal} {\bibinfo  {journal} {Phys. Rev. X}\ }\textbf {\bibinfo {volume} {12}},\ \bibinfo {pages} {021027} (\bibinfo {year} {2022})}\BibitemShut {NoStop}%
\bibitem [{\citenamefont {Bluvstein}\ \emph {et~al.}(2023)\citenamefont {Bluvstein}, \citenamefont {Evered}, \citenamefont {Geim}, \citenamefont {Li}, \citenamefont {Zhou}, \citenamefont {Manovitz}, \citenamefont {Ebadi}, \citenamefont {Cain}, \citenamefont {Kalinowski}, \citenamefont {Hangleiter}, \citenamefont {Bonilla~Ataides}, \citenamefont {Maskara}, \citenamefont {Cong}, \citenamefont {Gao}, \citenamefont {Sales~Rodriguez}, \citenamefont {Karolyshyn}, \citenamefont {Semeghini}, \citenamefont {Gullans}, \citenamefont {Greiner}, \citenamefont {Vuletić},\ and\ \citenamefont {Lukin}}]{Bluvstein_2023}%
  \BibitemOpen
  \bibfield  {author} {\bibinfo {author} {\bibfnamefont {Dolev}\ \bibnamefont {Bluvstein}}, \bibinfo {author} {\bibfnamefont {Simon~J.}\ \bibnamefont {Evered}}, \bibinfo {author} {\bibfnamefont {Alexandra~A.}\ \bibnamefont {Geim}}, \bibinfo {author} {\bibfnamefont {Sophie~H.}\ \bibnamefont {Li}}, \bibinfo {author} {\bibfnamefont {Hengyun}\ \bibnamefont {Zhou}}, \bibinfo {author} {\bibfnamefont {Tom}\ \bibnamefont {Manovitz}}, \bibinfo {author} {\bibfnamefont {Sepehr}\ \bibnamefont {Ebadi}}, \bibinfo {author} {\bibfnamefont {Madelyn}\ \bibnamefont {Cain}}, \bibinfo {author} {\bibfnamefont {Marcin}\ \bibnamefont {Kalinowski}}, \bibinfo {author} {\bibfnamefont {Dominik}\ \bibnamefont {Hangleiter}}, \bibinfo {author} {\bibfnamefont {J.~Pablo}\ \bibnamefont {Bonilla~Ataides}}, \bibinfo {author} {\bibfnamefont {Nishad}\ \bibnamefont {Maskara}}, \bibinfo {author} {\bibfnamefont {Iris}\ \bibnamefont {Cong}}, \bibinfo {author} {\bibfnamefont {Xun}\ \bibnamefont {Gao}}, \bibinfo {author} {\bibfnamefont {Pedro}\
  \bibnamefont {Sales~Rodriguez}}, \bibinfo {author} {\bibfnamefont {Thomas}\ \bibnamefont {Karolyshyn}}, \bibinfo {author} {\bibfnamefont {Giulia}\ \bibnamefont {Semeghini}}, \bibinfo {author} {\bibfnamefont {Michael~J.}\ \bibnamefont {Gullans}}, \bibinfo {author} {\bibfnamefont {Markus}\ \bibnamefont {Greiner}}, \bibinfo {author} {\bibfnamefont {Vladan}\ \bibnamefont {Vuletić}}, \ and\ \bibinfo {author} {\bibfnamefont {Mikhail~D.}\ \bibnamefont {Lukin}},\ }\bibfield  {title} {\enquote {\bibinfo {title} {Logical quantum processor based on reconfigurable atom arrays},}\ }\href {\doibase 10.1038/s41586-023-06927-3} {\bibfield  {journal} {\bibinfo  {journal} {Nature}\ }\textbf {\bibinfo {volume} {626}},\ \bibinfo {pages} {58–65} (\bibinfo {year} {2023})}\BibitemShut {NoStop}%
\bibitem [{\citenamefont {Lis}\ \emph {et~al.}(2023)\citenamefont {Lis}, \citenamefont {Senoo}, \citenamefont {McGrew}, \citenamefont {R\"onchen}, \citenamefont {Jenkins},\ and\ \citenamefont {Kaufman}}]{Lis_2023}%
  \BibitemOpen
  \bibfield  {author} {\bibinfo {author} {\bibfnamefont {Joanna~W.}\ \bibnamefont {Lis}}, \bibinfo {author} {\bibfnamefont {Aruku}\ \bibnamefont {Senoo}}, \bibinfo {author} {\bibfnamefont {William~F.}\ \bibnamefont {McGrew}}, \bibinfo {author} {\bibfnamefont {Felix}\ \bibnamefont {R\"onchen}}, \bibinfo {author} {\bibfnamefont {Alec}\ \bibnamefont {Jenkins}}, \ and\ \bibinfo {author} {\bibfnamefont {Adam~M.}\ \bibnamefont {Kaufman}},\ }\bibfield  {title} {\enquote {\bibinfo {title} {Midcircuit operations using the omg architecture in neutral atom arrays},}\ }\href {\doibase 10.1103/PhysRevX.13.041035} {\bibfield  {journal} {\bibinfo  {journal} {Phys. Rev. X}\ }\textbf {\bibinfo {volume} {13}},\ \bibinfo {pages} {041035} (\bibinfo {year} {2023})}\BibitemShut {NoStop}%
\bibitem [{\citenamefont {Norcia}\ \emph {et~al.}(2023)\citenamefont {Norcia}, \citenamefont {Cairncross}, \citenamefont {Barnes}, \citenamefont {Battaglino}, \citenamefont {Brown}, \citenamefont {Brown}, \citenamefont {Cassella}, \citenamefont {Chen}, \citenamefont {Coxe}, \citenamefont {Crow}, \citenamefont {Epstein}, \citenamefont {Griger}, \citenamefont {Jones}, \citenamefont {Kim}, \citenamefont {Kindem}, \citenamefont {King}, \citenamefont {Kondov}, \citenamefont {Kotru}, \citenamefont {Lauigan}, \citenamefont {Li}, \citenamefont {Lu}, \citenamefont {Megidish}, \citenamefont {Marjanovic}, \citenamefont {McDonald}, \citenamefont {Mittiga}, \citenamefont {Muniz}, \citenamefont {Narayanaswami}, \citenamefont {Nishiguchi}, \citenamefont {Notermans}, \citenamefont {Paule}, \citenamefont {Pawlak}, \citenamefont {Peng}, \citenamefont {Ryou}, \citenamefont {Smull}, \citenamefont {Stack}, \citenamefont {Stone}, \citenamefont {Sucich}, \citenamefont {Urbanek}, \citenamefont {van~de Veerdonk}, \citenamefont
  {Vendeiro}, \citenamefont {Wilkason}, \citenamefont {Wu}, \citenamefont {Xie}, \citenamefont {Zhang},\ and\ \citenamefont {Bloom}}]{Norcia_2023}%
  \BibitemOpen
  \bibfield  {author} {\bibinfo {author} {\bibfnamefont {M.~A.}\ \bibnamefont {Norcia}}, \bibinfo {author} {\bibfnamefont {W.~B.}\ \bibnamefont {Cairncross}}, \bibinfo {author} {\bibfnamefont {K.}~\bibnamefont {Barnes}}, \bibinfo {author} {\bibfnamefont {P.}~\bibnamefont {Battaglino}}, \bibinfo {author} {\bibfnamefont {A.}~\bibnamefont {Brown}}, \bibinfo {author} {\bibfnamefont {M.~O.}\ \bibnamefont {Brown}}, \bibinfo {author} {\bibfnamefont {K.}~\bibnamefont {Cassella}}, \bibinfo {author} {\bibfnamefont {C.-A.}\ \bibnamefont {Chen}}, \bibinfo {author} {\bibfnamefont {R.}~\bibnamefont {Coxe}}, \bibinfo {author} {\bibfnamefont {D.}~\bibnamefont {Crow}}, \bibinfo {author} {\bibfnamefont {J.}~\bibnamefont {Epstein}}, \bibinfo {author} {\bibfnamefont {C.}~\bibnamefont {Griger}}, \bibinfo {author} {\bibfnamefont {A.~M.~W.}\ \bibnamefont {Jones}}, \bibinfo {author} {\bibfnamefont {H.}~\bibnamefont {Kim}}, \bibinfo {author} {\bibfnamefont {J.~M.}\ \bibnamefont {Kindem}}, \bibinfo {author} {\bibfnamefont
  {J.}~\bibnamefont {King}}, \bibinfo {author} {\bibfnamefont {S.~S.}\ \bibnamefont {Kondov}}, \bibinfo {author} {\bibfnamefont {K.}~\bibnamefont {Kotru}}, \bibinfo {author} {\bibfnamefont {J.}~\bibnamefont {Lauigan}}, \bibinfo {author} {\bibfnamefont {M.}~\bibnamefont {Li}}, \bibinfo {author} {\bibfnamefont {M.}~\bibnamefont {Lu}}, \bibinfo {author} {\bibfnamefont {E.}~\bibnamefont {Megidish}}, \bibinfo {author} {\bibfnamefont {J.}~\bibnamefont {Marjanovic}}, \bibinfo {author} {\bibfnamefont {M.}~\bibnamefont {McDonald}}, \bibinfo {author} {\bibfnamefont {T.}~\bibnamefont {Mittiga}}, \bibinfo {author} {\bibfnamefont {J.~A.}\ \bibnamefont {Muniz}}, \bibinfo {author} {\bibfnamefont {S.}~\bibnamefont {Narayanaswami}}, \bibinfo {author} {\bibfnamefont {C.}~\bibnamefont {Nishiguchi}}, \bibinfo {author} {\bibfnamefont {R.}~\bibnamefont {Notermans}}, \bibinfo {author} {\bibfnamefont {T.}~\bibnamefont {Paule}}, \bibinfo {author} {\bibfnamefont {K.~A.}\ \bibnamefont {Pawlak}}, \bibinfo {author} {\bibfnamefont
  {L.~S.}\ \bibnamefont {Peng}}, \bibinfo {author} {\bibfnamefont {A.}~\bibnamefont {Ryou}}, \bibinfo {author} {\bibfnamefont {A.}~\bibnamefont {Smull}}, \bibinfo {author} {\bibfnamefont {D.}~\bibnamefont {Stack}}, \bibinfo {author} {\bibfnamefont {M.}~\bibnamefont {Stone}}, \bibinfo {author} {\bibfnamefont {A.}~\bibnamefont {Sucich}}, \bibinfo {author} {\bibfnamefont {M.}~\bibnamefont {Urbanek}}, \bibinfo {author} {\bibfnamefont {R.~J.~M.}\ \bibnamefont {van~de Veerdonk}}, \bibinfo {author} {\bibfnamefont {Z.}~\bibnamefont {Vendeiro}}, \bibinfo {author} {\bibfnamefont {T.}~\bibnamefont {Wilkason}}, \bibinfo {author} {\bibfnamefont {T.-Y.}\ \bibnamefont {Wu}}, \bibinfo {author} {\bibfnamefont {X.}~\bibnamefont {Xie}}, \bibinfo {author} {\bibfnamefont {X.}~\bibnamefont {Zhang}}, \ and\ \bibinfo {author} {\bibfnamefont {B.~J.}\ \bibnamefont {Bloom}},\ }\bibfield  {title} {\enquote {\bibinfo {title} {Midcircuit qubit measurement and rearrangement in a $^{171}\mathrm{Yb}$ atomic array},}\ }\href {\doibase
  10.1103/PhysRevX.13.041034} {\bibfield  {journal} {\bibinfo  {journal} {Phys. Rev. X}\ }\textbf {\bibinfo {volume} {13}},\ \bibinfo {pages} {041034} (\bibinfo {year} {2023})}\BibitemShut {NoStop}%
\bibitem [{\citenamefont {Huie}\ \emph {et~al.}(2023)\citenamefont {Huie}, \citenamefont {Li}, \citenamefont {Chen}, \citenamefont {Hu}, \citenamefont {Jia}, \citenamefont {Sun},\ and\ \citenamefont {Covey}}]{Huie_2023}%
  \BibitemOpen
  \bibfield  {author} {\bibinfo {author} {\bibfnamefont {William}\ \bibnamefont {Huie}}, \bibinfo {author} {\bibfnamefont {Lintao}\ \bibnamefont {Li}}, \bibinfo {author} {\bibfnamefont {Neville}\ \bibnamefont {Chen}}, \bibinfo {author} {\bibfnamefont {Xiye}\ \bibnamefont {Hu}}, \bibinfo {author} {\bibfnamefont {Zhubing}\ \bibnamefont {Jia}}, \bibinfo {author} {\bibfnamefont {Won Kyu~Calvin}\ \bibnamefont {Sun}}, \ and\ \bibinfo {author} {\bibfnamefont {Jacob~P.}\ \bibnamefont {Covey}},\ }\bibfield  {title} {\enquote {\bibinfo {title} {Repetitive readout and real-time control of nuclear spin qubits in ${}^{171}\mathrm{Yb}$ atoms},}\ }\href {\doibase 10.1103/PRXQuantum.4.030337} {\bibfield  {journal} {\bibinfo  {journal} {PRX Quantum}\ }\textbf {\bibinfo {volume} {4}},\ \bibinfo {pages} {030337} (\bibinfo {year} {2023})}\BibitemShut {NoStop}%
\bibitem [{\citenamefont {Fowler}\ \emph {et~al.}(2010)\citenamefont {Fowler}, \citenamefont {Wang}, \citenamefont {Hill}, \citenamefont {Ladd}, \citenamefont {Meter},\ and\ \citenamefont {Hollenberg}}]{Fowler_2010}%
  \BibitemOpen
  \bibfield  {author} {\bibinfo {author} {\bibfnamefont {Austin~G.}\ \bibnamefont {Fowler}}, \bibinfo {author} {\bibfnamefont {David~S.}\ \bibnamefont {Wang}}, \bibinfo {author} {\bibfnamefont {Charles~D.}\ \bibnamefont {Hill}}, \bibinfo {author} {\bibfnamefont {Thaddeus~D.}\ \bibnamefont {Ladd}}, \bibinfo {author} {\bibfnamefont {Rodney~Van}\ \bibnamefont {Meter}}, \ and\ \bibinfo {author} {\bibfnamefont {Lloyd C.~L.}\ \bibnamefont {Hollenberg}},\ }\bibfield  {title} {\enquote {\bibinfo {title} {Surface code quantum communication},}\ }\href {\doibase 10.1103/physrevlett.104.180503} {\bibfield  {journal} {\bibinfo  {journal} {Physical Review Letters}\ }\textbf {\bibinfo {volume} {104}} (\bibinfo {year} {2010}),\ 10.1103/physrevlett.104.180503}\BibitemShut {NoStop}%
\bibitem [{\citenamefont {Azuma}\ \emph {et~al.}(2023)\citenamefont {Azuma}, \citenamefont {Economou}, \citenamefont {Elkouss}, \citenamefont {Hilaire}, \citenamefont {Jiang}, \citenamefont {Lo},\ and\ \citenamefont {Tzitrin}}]{Azuma_2023}%
  \BibitemOpen
  \bibfield  {author} {\bibinfo {author} {\bibfnamefont {Koji}\ \bibnamefont {Azuma}}, \bibinfo {author} {\bibfnamefont {Sophia~E.}\ \bibnamefont {Economou}}, \bibinfo {author} {\bibfnamefont {David}\ \bibnamefont {Elkouss}}, \bibinfo {author} {\bibfnamefont {Paul}\ \bibnamefont {Hilaire}}, \bibinfo {author} {\bibfnamefont {Liang}\ \bibnamefont {Jiang}}, \bibinfo {author} {\bibfnamefont {Hoi-Kwong}\ \bibnamefont {Lo}}, \ and\ \bibinfo {author} {\bibfnamefont {Ilan}\ \bibnamefont {Tzitrin}},\ }\bibfield  {title} {\enquote {\bibinfo {title} {Quantum repeaters: From quantum networks to the quantum internet},}\ }\href {\doibase 10.1103/revmodphys.95.045006} {\bibfield  {journal} {\bibinfo  {journal} {Reviews of Modern Physics}\ }\textbf {\bibinfo {volume} {95}} (\bibinfo {year} {2023}),\ 10.1103/revmodphys.95.045006}\BibitemShut {NoStop}%
\bibitem [{\citenamefont {Capalbo}\ \emph {et~al.}(2002)\citenamefont {Capalbo}, \citenamefont {Reingold}, \citenamefont {Vadhan},\ and\ \citenamefont {Wigderson}}]{Capalbo_2002}%
  \BibitemOpen
  \bibfield  {author} {\bibinfo {author} {\bibfnamefont {Michael}\ \bibnamefont {Capalbo}}, \bibinfo {author} {\bibfnamefont {Omer}\ \bibnamefont {Reingold}}, \bibinfo {author} {\bibfnamefont {Salil}\ \bibnamefont {Vadhan}}, \ and\ \bibinfo {author} {\bibfnamefont {Avi}\ \bibnamefont {Wigderson}},\ }\bibfield  {title} {\enquote {\bibinfo {title} {Randomness conductors and constant-degree lossless expanders},}\ }in\ \href {\doibase 10.1145/509907.510003} {\emph {\bibinfo {booktitle} {Proceedings of the Thiry-Fourth Annual ACM Symposium on Theory of Computing}}},\ \bibinfo {series and number} {STOC '02}\ (\bibinfo  {publisher} {Association for Computing Machinery},\ \bibinfo {address} {New York, NY, USA},\ \bibinfo {year} {2002})\ p.\ \bibinfo {pages} {659–668}\BibitemShut {NoStop}%
\bibitem [{\citenamefont {Golowich}(2024)}]{Golowich_2023}%
  \BibitemOpen
  \bibfield  {author} {\bibinfo {author} {\bibfnamefont {Louis}\ \bibnamefont {Golowich}},\ }\bibfield  {title} {\enquote {\bibinfo {title} {New explicit constant-degree lossless expanders},}\ }in\ \href {\doibase 10.1137/1.9781611977912.177} {\emph {\bibinfo {booktitle} {Proceedings of the 2024 Annual ACM-SIAM Symposium on Discrete Algorithms (SODA)}}}\ (\bibinfo {year} {2024})\ pp.\ \bibinfo {pages} {4963--4971}\BibitemShut {NoStop}%
\bibitem [{\citenamefont {Panteleev}\ and\ \citenamefont {Kalachev}(2021)}]{Panteleev_2021}%
  \BibitemOpen
  \bibfield  {author} {\bibinfo {author} {\bibfnamefont {Pavel}\ \bibnamefont {Panteleev}}\ and\ \bibinfo {author} {\bibfnamefont {Gleb}\ \bibnamefont {Kalachev}},\ }\bibfield  {title} {\enquote {\bibinfo {title} {Degenerate quantum ldpc codes with good finite length performance},}\ }\href {\doibase 10.22331/q-2021-11-22-585} {\bibfield  {journal} {\bibinfo  {journal} {Quantum}\ }\textbf {\bibinfo {volume} {5}},\ \bibinfo {pages} {585} (\bibinfo {year} {2021})}\BibitemShut {NoStop}%
\bibitem [{\citenamefont {Quintavalle}\ and\ \citenamefont {Campbell}(2022)}]{Quintavalle_2022}%
  \BibitemOpen
  \bibfield  {author} {\bibinfo {author} {\bibfnamefont {Armanda~O.}\ \bibnamefont {Quintavalle}}\ and\ \bibinfo {author} {\bibfnamefont {Earl~T.}\ \bibnamefont {Campbell}},\ }\bibfield  {title} {\enquote {\bibinfo {title} {Reshape: A decoder for hypergraph product codes},}\ }\href {\doibase 10.1109/TIT.2022.3184108} {\bibfield  {journal} {\bibinfo  {journal} {IEEE Transactions on Information Theory}\ }\textbf {\bibinfo {volume} {68}},\ \bibinfo {pages} {6569--6584} (\bibinfo {year} {2022})}\BibitemShut {NoStop}%
\bibitem [{\citenamefont {Delfosse}\ \emph {et~al.}(2022)\citenamefont {Delfosse}, \citenamefont {Londe},\ and\ \citenamefont {Beverland}}]{Delfosse_2022}%
  \BibitemOpen
  \bibfield  {author} {\bibinfo {author} {\bibfnamefont {Nicolas}\ \bibnamefont {Delfosse}}, \bibinfo {author} {\bibfnamefont {Vivien}\ \bibnamefont {Londe}}, \ and\ \bibinfo {author} {\bibfnamefont {Michael~E.}\ \bibnamefont {Beverland}},\ }\bibfield  {title} {\enquote {\bibinfo {title} {Toward a union-find decoder for quantum ldpc codes},}\ }\href {\doibase 10.1109/TIT.2022.3143452} {\bibfield  {journal} {\bibinfo  {journal} {IEEE Transactions on Information Theory}\ }\textbf {\bibinfo {volume} {68}},\ \bibinfo {pages} {3187--3199} (\bibinfo {year} {2022})}\BibitemShut {NoStop}%
\bibitem [{\citenamefont {Sala}\ \emph {et~al.}(2020)\citenamefont {Sala}, \citenamefont {Rakovszky}, \citenamefont {Verresen}, \citenamefont {Knap},\ and\ \citenamefont {Pollmann}}]{sala}%
  \BibitemOpen
  \bibfield  {author} {\bibinfo {author} {\bibfnamefont {Pablo}\ \bibnamefont {Sala}}, \bibinfo {author} {\bibfnamefont {Tibor}\ \bibnamefont {Rakovszky}}, \bibinfo {author} {\bibfnamefont {Ruben}\ \bibnamefont {Verresen}}, \bibinfo {author} {\bibfnamefont {Michael}\ \bibnamefont {Knap}}, \ and\ \bibinfo {author} {\bibfnamefont {Frank}\ \bibnamefont {Pollmann}},\ }\bibfield  {title} {\enquote {\bibinfo {title} {Ergodicity breaking arising from hilbert space fragmentation in dipole-conserving hamiltonians},}\ }\href {\doibase 10.1103/PhysRevX.10.011047} {\bibfield  {journal} {\bibinfo  {journal} {Phys. Rev. X}\ }\textbf {\bibinfo {volume} {10}},\ \bibinfo {pages} {011047} (\bibinfo {year} {2020})}\BibitemShut {NoStop}%
\bibitem [{\citenamefont {Khemani}\ \emph {et~al.}(2020)\citenamefont {Khemani}, \citenamefont {Hermele},\ and\ \citenamefont {Nandkishore}}]{Khemani_2020}%
  \BibitemOpen
  \bibfield  {author} {\bibinfo {author} {\bibfnamefont {Vedika}\ \bibnamefont {Khemani}}, \bibinfo {author} {\bibfnamefont {Michael}\ \bibnamefont {Hermele}}, \ and\ \bibinfo {author} {\bibfnamefont {Rahul}\ \bibnamefont {Nandkishore}},\ }\bibfield  {title} {\enquote {\bibinfo {title} {Localization from hilbert space shattering: From theory to physical realizations},}\ }\href {\doibase 10.1103/PhysRevB.101.174204} {\bibfield  {journal} {\bibinfo  {journal} {Phys. Rev. B}\ }\textbf {\bibinfo {volume} {101}},\ \bibinfo {pages} {174204} (\bibinfo {year} {2020})}\BibitemShut {NoStop}%
\bibitem [{\citenamefont {Hart}\ and\ \citenamefont {Nandkishore}(2022)}]{Hart_2022}%
  \BibitemOpen
  \bibfield  {author} {\bibinfo {author} {\bibfnamefont {Oliver}\ \bibnamefont {Hart}}\ and\ \bibinfo {author} {\bibfnamefont {Rahul}\ \bibnamefont {Nandkishore}},\ }\bibfield  {title} {\enquote {\bibinfo {title} {Hilbert space shattering and dynamical freezing in the quantum ising model},}\ }\href {\doibase 10.1103/physrevb.106.214426} {\bibfield  {journal} {\bibinfo  {journal} {Physical Review B}\ }\textbf {\bibinfo {volume} {106}} (\bibinfo {year} {2022}),\ 10.1103/physrevb.106.214426}\BibitemShut {NoStop}%
\bibitem [{\citenamefont {Stephen}\ \emph {et~al.}(2024)\citenamefont {Stephen}, \citenamefont {Hart},\ and\ \citenamefont {Nandkishore}}]{Stephen_2024}%
  \BibitemOpen
  \bibfield  {author} {\bibinfo {author} {\bibfnamefont {David~T.}\ \bibnamefont {Stephen}}, \bibinfo {author} {\bibfnamefont {Oliver}\ \bibnamefont {Hart}}, \ and\ \bibinfo {author} {\bibfnamefont {Rahul~M.}\ \bibnamefont {Nandkishore}},\ }\bibfield  {title} {\enquote {\bibinfo {title} {Ergodicity breaking provably robust to arbitrary perturbations},}\ }\href {\doibase 10.1103/physrevlett.132.040401} {\bibfield  {journal} {\bibinfo  {journal} {Physical Review Letters}\ }\textbf {\bibinfo {volume} {132}} (\bibinfo {year} {2024}),\ 10.1103/physrevlett.132.040401}\BibitemShut {NoStop}%
\bibitem [{\citenamefont {Stahl}\ \emph {et~al.}(2023)\citenamefont {Stahl}, \citenamefont {Nandkishore},\ and\ \citenamefont {Hart}}]{stahl2023}%
  \BibitemOpen
  \bibfield  {author} {\bibinfo {author} {\bibfnamefont {Charles}\ \bibnamefont {Stahl}}, \bibinfo {author} {\bibfnamefont {Rahul}\ \bibnamefont {Nandkishore}}, \ and\ \bibinfo {author} {\bibfnamefont {Oliver}\ \bibnamefont {Hart}},\ }\href@noop {} {\enquote {\bibinfo {title} {Topologically stable ergodicity breaking from emergent higher-form symmetries in generalized quantum loop models},}\ } (\bibinfo {year} {2023}),\ \Eprint {http://arxiv.org/abs/2304.04792} {arXiv:2304.04792 [cond-mat.stat-mech]} \BibitemShut {NoStop}%
\bibitem [{\citenamefont {Han}\ \emph {et~al.}(2024)\citenamefont {Han}, \citenamefont {Chen},\ and\ \citenamefont {Lake}}]{Han:2024yfm}%
  \BibitemOpen
  \bibfield  {author} {\bibinfo {author} {\bibfnamefont {Yiqiu}\ \bibnamefont {Han}}, \bibinfo {author} {\bibfnamefont {Xiao}\ \bibnamefont {Chen}}, \ and\ \bibinfo {author} {\bibfnamefont {Ethan}\ \bibnamefont {Lake}},\ }\bibfield  {title} {\enquote {\bibinfo {title} {{Exponentially slow thermalization and the robustness of Hilbert space fragmentation}},}\ }\href@noop {} {\  (\bibinfo {year} {2024})},\ \Eprint {http://arxiv.org/abs/2401.11294} {arXiv:2401.11294 [quant-ph]} \BibitemShut {NoStop}%
\bibitem [{\citenamefont {Yin}\ \emph {et~al.}(2024)\citenamefont {Yin}, \citenamefont {Nandkishore},\ and\ \citenamefont {Lucas}}]{LDPC_MBL}%
  \BibitemOpen
  \bibfield  {author} {\bibinfo {author} {\bibfnamefont {Chao}\ \bibnamefont {Yin}}, \bibinfo {author} {\bibfnamefont {Rahul}\ \bibnamefont {Nandkishore}}, \ and\ \bibinfo {author} {\bibfnamefont {Andrew}\ \bibnamefont {Lucas}},\ }\href {https://arxiv.org/abs/2405.12279} {\enquote {\bibinfo {title} {Eigenstate localization in a many-body quantum system},}\ } (\bibinfo {year} {2024}),\ \Eprint {http://arxiv.org/abs/2405.12279} {arXiv:2405.12279 [cond-mat.stat-mech]} \BibitemShut {NoStop}%
\bibitem [{\citenamefont {Edwards}\ and\ \citenamefont {Anderson}(1975)}]{Edwards_1975}%
  \BibitemOpen
  \bibfield  {author} {\bibinfo {author} {\bibfnamefont {S~F}\ \bibnamefont {Edwards}}\ and\ \bibinfo {author} {\bibfnamefont {P~W}\ \bibnamefont {Anderson}},\ }\bibfield  {title} {\enquote {\bibinfo {title} {Theory of spin glasses},}\ }\href {\doibase 10.1088/0305-4608/5/5/017} {\bibfield  {journal} {\bibinfo  {journal} {Journal of Physics F: Metal Physics}\ }\textbf {\bibinfo {volume} {5}},\ \bibinfo {pages} {965} (\bibinfo {year} {1975})}\BibitemShut {NoStop}%
\bibitem [{\citenamefont {Sherrington}\ and\ \citenamefont {Kirkpatrick}(1975)}]{Sherrington_1975}%
  \BibitemOpen
  \bibfield  {author} {\bibinfo {author} {\bibfnamefont {David}\ \bibnamefont {Sherrington}}\ and\ \bibinfo {author} {\bibfnamefont {Scott}\ \bibnamefont {Kirkpatrick}},\ }\bibfield  {title} {\enquote {\bibinfo {title} {Solvable model of a spin-glass},}\ }\href {\doibase 10.1103/PhysRevLett.35.1792} {\bibfield  {journal} {\bibinfo  {journal} {Phys. Rev. Lett.}\ }\textbf {\bibinfo {volume} {35}},\ \bibinfo {pages} {1792--1796} (\bibinfo {year} {1975})}\BibitemShut {NoStop}%
\bibitem [{\citenamefont {Mezard}\ \emph {et~al.}(1987)\citenamefont {Mezard}, \citenamefont {Parisi},\ and\ \citenamefont {Virasoro}}]{parisi}%
  \BibitemOpen
  \bibfield  {author} {\bibinfo {author} {\bibfnamefont {M.}~\bibnamefont {Mezard}}, \bibinfo {author} {\bibfnamefont {G.}~\bibnamefont {Parisi}}, \ and\ \bibinfo {author} {\bibfnamefont {M.}~\bibnamefont {Virasoro}},\ }\href@noop {} {\emph {\bibinfo {title} {Spin Glass Theory and Beyond}}}\ (\bibinfo  {publisher} {World Scientific Press},\ \bibinfo {year} {1987})\BibitemShut {NoStop}%
\bibitem [{\citenamefont {Krzakala}\ \emph {et~al.}(2007)\citenamefont {Krzakala}, \citenamefont {Montanari}, \citenamefont {Ricci-Tersenghi},\ and\ \citenamefont {Zdeborova}}]{krzakala}%
  \BibitemOpen
  \bibfield  {author} {\bibinfo {author} {\bibfnamefont {F.}~\bibnamefont {Krzakala}}, \bibinfo {author} {\bibfnamefont {A.}~\bibnamefont {Montanari}}, \bibinfo {author} {\bibfnamefont {F.}~\bibnamefont {Ricci-Tersenghi}}, \ and\ \bibinfo {author} {\bibfnamefont {L.}~\bibnamefont {Zdeborova}},\ }\bibfield  {title} {\enquote {\bibinfo {title} {Gibbs states and the set of solutions of random constraint satisfaction problems},}\ }\href@noop {} {\bibfield  {journal} {\bibinfo  {journal} {Proc. Nat. Acad. Sci.}\ }\textbf {\bibinfo {volume} {104}},\ \bibinfo {pages} {10318} (\bibinfo {year} {2007})}\BibitemShut {NoStop}%
\bibitem [{\citenamefont {Aharonov}\ and\ \citenamefont {Eldar}(2015)}]{Aharonov_2015}%
  \BibitemOpen
  \bibfield  {author} {\bibinfo {author} {\bibfnamefont {Dorit}\ \bibnamefont {Aharonov}}\ and\ \bibinfo {author} {\bibfnamefont {Lior}\ \bibnamefont {Eldar}},\ }\bibfield  {title} {\enquote {\bibinfo {title} {Quantum locally testable codes},}\ }\href {\doibase 10.1137/140975498} {\bibfield  {journal} {\bibinfo  {journal} {SIAM Journal on Computing}\ }\textbf {\bibinfo {volume} {44}},\ \bibinfo {pages} {1230--1262} (\bibinfo {year} {2015})}\BibitemShut {NoStop}%
\bibitem [{\citenamefont {Anshu}\ \emph {et~al.}(2023)\citenamefont {Anshu}, \citenamefont {Breuckmann},\ and\ \citenamefont {Nirkhe}}]{Anshu_2023}%
  \BibitemOpen
  \bibfield  {author} {\bibinfo {author} {\bibfnamefont {Anurag}\ \bibnamefont {Anshu}}, \bibinfo {author} {\bibfnamefont {Nikolas~P.}\ \bibnamefont {Breuckmann}}, \ and\ \bibinfo {author} {\bibfnamefont {Chinmay}\ \bibnamefont {Nirkhe}},\ }\bibfield  {title} {\enquote {\bibinfo {title} {Nlts hamiltonians from good quantum codes},}\ }in\ \href {\doibase 10.1145/3564246.3585114} {\emph {\bibinfo {booktitle} {Proceedings of the 55th Annual ACM Symposium on Theory of Computing}}},\ \bibinfo {series and number} {STOC ’23}\ (\bibinfo  {publisher} {ACM},\ \bibinfo {year} {2023})\BibitemShut {NoStop}%
\bibitem [{\citenamefont {Eldar}\ and\ \citenamefont {Harrow}(2017)}]{Eldar_2017}%
  \BibitemOpen
  \bibfield  {author} {\bibinfo {author} {\bibfnamefont {Lior}\ \bibnamefont {Eldar}}\ and\ \bibinfo {author} {\bibfnamefont {Aram~W.}\ \bibnamefont {Harrow}},\ }\bibfield  {title} {\enquote {\bibinfo {title} {Local hamiltonians whose ground states are hard to approximate},}\ }in\ \href {\doibase 10.1109/focs.2017.46} {\emph {\bibinfo {booktitle} {2017 IEEE 58th Annual Symposium on Foundations of Computer Science (FOCS)}}}\ (\bibinfo  {publisher} {IEEE},\ \bibinfo {year} {2017})\BibitemShut {NoStop}%
\bibitem [{\citenamefont {Chen}\ \emph {et~al.}(2010)\citenamefont {Chen}, \citenamefont {Gu},\ and\ \citenamefont {Wen}}]{Chen_2010}%
  \BibitemOpen
  \bibfield  {author} {\bibinfo {author} {\bibfnamefont {Xie}\ \bibnamefont {Chen}}, \bibinfo {author} {\bibfnamefont {Zheng-Cheng}\ \bibnamefont {Gu}}, \ and\ \bibinfo {author} {\bibfnamefont {Xiao-Gang}\ \bibnamefont {Wen}},\ }\bibfield  {title} {\enquote {\bibinfo {title} {Local unitary transformation, long-range quantum entanglement, wave function renormalization, and topological order},}\ }\href {\doibase 10.1103/PhysRevB.82.155138} {\bibfield  {journal} {\bibinfo  {journal} {Phys. Rev. B}\ }\textbf {\bibinfo {volume} {82}},\ \bibinfo {pages} {155138} (\bibinfo {year} {2010})}\BibitemShut {NoStop}%
\bibitem [{\citenamefont {Rakovszky}\ \emph {et~al.}(2024)\citenamefont {Rakovszky}, \citenamefont {Gopalakrishnan},\ and\ \citenamefont {von Keyserlingk}}]{rakovszky_stable}%
  \BibitemOpen
  \bibfield  {author} {\bibinfo {author} {\bibfnamefont {Tibor}\ \bibnamefont {Rakovszky}}, \bibinfo {author} {\bibfnamefont {Sarang}\ \bibnamefont {Gopalakrishnan}}, \ and\ \bibinfo {author} {\bibfnamefont {Curt}\ \bibnamefont {von Keyserlingk}},\ }\href {https://arxiv.org/abs/2308.15495} {\enquote {\bibinfo {title} {Defining stable phases of open quantum systems},}\ } (\bibinfo {year} {2024}),\ \Eprint {http://arxiv.org/abs/2308.15495} {arXiv:2308.15495 [quant-ph]} \BibitemShut {NoStop}%
\bibitem [{\citenamefont {Sang}\ \emph {et~al.}(2023)\citenamefont {Sang}, \citenamefont {Zou},\ and\ \citenamefont {Hsieh}}]{sang2023}%
  \BibitemOpen
  \bibfield  {author} {\bibinfo {author} {\bibfnamefont {Shengqi}\ \bibnamefont {Sang}}, \bibinfo {author} {\bibfnamefont {Yijian}\ \bibnamefont {Zou}}, \ and\ \bibinfo {author} {\bibfnamefont {Timothy~H.}\ \bibnamefont {Hsieh}},\ }\href@noop {} {\enquote {\bibinfo {title} {Mixed-state quantum phases: Renormalization and quantum error correction},}\ } (\bibinfo {year} {2023}),\ \Eprint {http://arxiv.org/abs/2310.08639} {arXiv:2310.08639 [quant-ph]} \BibitemShut {NoStop}%
\bibitem [{\citenamefont {Chen}\ and\ \citenamefont {Grover}(2023)}]{Chen:2023auj}%
  \BibitemOpen
  \bibfield  {author} {\bibinfo {author} {\bibfnamefont {Yu-Hsueh}\ \bibnamefont {Chen}}\ and\ \bibinfo {author} {\bibfnamefont {Tarun}\ \bibnamefont {Grover}},\ }\bibfield  {title} {\enquote {\bibinfo {title} {{Symmetry-enforced many-body separability transitions}},}\ }\href@noop {} {\  (\bibinfo {year} {2023})},\ \Eprint {http://arxiv.org/abs/2310.07286} {arXiv:2310.07286 [quant-ph]} \BibitemShut {NoStop}%
\bibitem [{\citenamefont {Placke}\ \emph {et~al.}(to appear)\citenamefont {Placke}, \citenamefont {Rakovszky}, \citenamefont {Breuckmann}, \citenamefont {Sommers},\ and\ \citenamefont {Khemani}}]{Breuckmann_2024}%
  \BibitemOpen
  \bibfield  {author} {\bibinfo {author} {\bibfnamefont {Benedikt}\ \bibnamefont {Placke}}, \bibinfo {author} {\bibfnamefont {Tibor}\ \bibnamefont {Rakovszky}}, \bibinfo {author} {\bibfnamefont {Nikolas~P.}\ \bibnamefont {Breuckmann}}, \bibinfo {author} {\bibfnamefont {Grace}\ \bibnamefont {Sommers}}, \ and\ \bibinfo {author} {\bibfnamefont {Vedika}\ \bibnamefont {Khemani}},\ }\href@noop {} {\enquote {\bibinfo {title} {Spin glass order in classical and quantum ldpc codes},}\ } (\bibinfo {year} {to appear})\BibitemShut {NoStop}%
\bibitem [{\citenamefont {Kramers}\ and\ \citenamefont {Wannier}(1941)}]{KW_duality}%
  \BibitemOpen
  \bibfield  {author} {\bibinfo {author} {\bibfnamefont {H.~A.}\ \bibnamefont {Kramers}}\ and\ \bibinfo {author} {\bibfnamefont {G.~H.}\ \bibnamefont {Wannier}},\ }\bibfield  {title} {\enquote {\bibinfo {title} {Statistics of the two-dimensional ferromagnet. part i},}\ }\href {\doibase 10.1103/PhysRev.60.252} {\bibfield  {journal} {\bibinfo  {journal} {Phys. Rev.}\ }\textbf {\bibinfo {volume} {60}},\ \bibinfo {pages} {252--262} (\bibinfo {year} {1941})}\BibitemShut {NoStop}%
\bibitem [{\citenamefont {Levin}\ \emph {et~al.}()\citenamefont {Levin}, \citenamefont {Peres},\ and\ \citenamefont {Wilmer}}]{levin_markovchains}%
  \BibitemOpen
  \bibfield  {author} {\bibinfo {author} {\bibfnamefont {D.A.}\ \bibnamefont {Levin}}, \bibinfo {author} {\bibfnamefont {Y.}~\bibnamefont {Peres}}, \ and\ \bibinfo {author} {\bibfnamefont {E.L.}\ \bibnamefont {Wilmer}},\ }\href {https://books.google.com/books?id=6Cg5Nq5sSv4C} {\emph {\bibinfo {title} {Markov Chains and Mixing Times}}}\ (\bibinfo  {publisher} {American Mathematical Soc.})\BibitemShut {NoStop}%
\bibitem [{\citenamefont {Kschischang}\ \emph {et~al.}(2001)\citenamefont {Kschischang}, \citenamefont {Frey},\ and\ \citenamefont {Loeliger}}]{BP_decoding}%
  \BibitemOpen
  \bibfield  {author} {\bibinfo {author} {\bibfnamefont {F.R.}\ \bibnamefont {Kschischang}}, \bibinfo {author} {\bibfnamefont {B.J.}\ \bibnamefont {Frey}}, \ and\ \bibinfo {author} {\bibfnamefont {H.-A.}\ \bibnamefont {Loeliger}},\ }\bibfield  {title} {\enquote {\bibinfo {title} {Factor graphs and the sum-product algorithm},}\ }\href {\doibase 10.1109/18.910572} {\bibfield  {journal} {\bibinfo  {journal} {IEEE Transactions on Information Theory}\ }\textbf {\bibinfo {volume} {47}},\ \bibinfo {pages} {498--519} (\bibinfo {year} {2001})}\BibitemShut {NoStop}%
\bibitem [{\citenamefont {Gottesman}(1997)}]{gottesman1997}%
  \BibitemOpen
  \bibfield  {author} {\bibinfo {author} {\bibfnamefont {Daniel}\ \bibnamefont {Gottesman}},\ }\href@noop {} {\enquote {\bibinfo {title} {Stabilizer codes and quantum error correction},}\ } (\bibinfo {year} {1997}),\ \Eprint {http://arxiv.org/abs/quant-ph/9705052} {arXiv:quant-ph/9705052 [quant-ph]} \BibitemShut {NoStop}%
\bibitem [{\citenamefont {Bravyi}\ and\ \citenamefont {Terhal}(2009)}]{Bravyi_2009}%
  \BibitemOpen
  \bibfield  {author} {\bibinfo {author} {\bibfnamefont {Sergey}\ \bibnamefont {Bravyi}}\ and\ \bibinfo {author} {\bibfnamefont {Barbara}\ \bibnamefont {Terhal}},\ }\bibfield  {title} {\enquote {\bibinfo {title} {A no-go theorem for a two-dimensional self-correcting quantum memory based on stabilizer codes},}\ }\href {\doibase 10.1088/1367-2630/11/4/043029} {\bibfield  {journal} {\bibinfo  {journal} {New Journal of Physics}\ }\textbf {\bibinfo {volume} {11}},\ \bibinfo {pages} {043029} (\bibinfo {year} {2009})}\BibitemShut {NoStop}%
\bibitem [{\citenamefont {Baspin}\ and\ \citenamefont {Krishna}(2022)}]{Baspin_2022}%
  \BibitemOpen
  \bibfield  {author} {\bibinfo {author} {\bibfnamefont {Nouédyn}\ \bibnamefont {Baspin}}\ and\ \bibinfo {author} {\bibfnamefont {Anirudh}\ \bibnamefont {Krishna}},\ }\bibfield  {title} {\enquote {\bibinfo {title} {Connectivity constrains quantum codes},}\ }\href {\doibase 10.22331/q-2022-05-13-711} {\bibfield  {journal} {\bibinfo  {journal} {Quantum}\ }\textbf {\bibinfo {volume} {6}},\ \bibinfo {pages} {711} (\bibinfo {year} {2022})}\BibitemShut {NoStop}%
\bibitem [{\citenamefont {Kovalev}\ and\ \citenamefont {Pryadko}(2013)}]{Kovalev_2013}%
  \BibitemOpen
  \bibfield  {author} {\bibinfo {author} {\bibfnamefont {Alexey~A.}\ \bibnamefont {Kovalev}}\ and\ \bibinfo {author} {\bibfnamefont {Leonid~P.}\ \bibnamefont {Pryadko}},\ }\bibfield  {title} {\enquote {\bibinfo {title} {Fault tolerance of quantum low-density parity check codes with sublinear distance scaling},}\ }\href {\doibase 10.1103/physreva.87.020304} {\bibfield  {journal} {\bibinfo  {journal} {Physical Review A}\ }\textbf {\bibinfo {volume} {87}} (\bibinfo {year} {2013}),\ 10.1103/physreva.87.020304}\BibitemShut {NoStop}%
\bibitem [{\citenamefont {Aliferis}\ \emph {et~al.}(2008)\citenamefont {Aliferis}, \citenamefont {Gottesman},\ and\ \citenamefont {Preskill}}]{Aliferis_2008}%
  \BibitemOpen
  \bibfield  {author} {\bibinfo {author} {\bibfnamefont {Panos}\ \bibnamefont {Aliferis}}, \bibinfo {author} {\bibfnamefont {Daniel}\ \bibnamefont {Gottesman}}, \ and\ \bibinfo {author} {\bibfnamefont {John}\ \bibnamefont {Preskill}},\ }\bibfield  {title} {\enquote {\bibinfo {title} {Accuracy threshold for postselected quantum computation},}\ }\href@noop {} {\bibfield  {journal} {\bibinfo  {journal} {Quantum Info. Comput.}\ }\textbf {\bibinfo {volume} {8}},\ \bibinfo {pages} {181–244} (\bibinfo {year} {2008})}\BibitemShut {NoStop}%
\bibitem [{\citenamefont {Roffe}(2022)}]{Roffe_LDPC_Python_tools_2022}%
  \BibitemOpen
  \bibfield  {author} {\bibinfo {author} {\bibfnamefont {Joschka}\ \bibnamefont {Roffe}},\ }\href {https://pypi.org/project/ldpc/} {\enquote {\bibinfo {title} {{LDPC: Python tools for low density parity check codes}},}\ } (\bibinfo {year} {2022})\BibitemShut {NoStop}%
\bibitem [{\citenamefont {Roffe}\ \emph {et~al.}(2020)\citenamefont {Roffe}, \citenamefont {White}, \citenamefont {Burton},\ and\ \citenamefont {Campbell}}]{roffe_decoding_2020}%
  \BibitemOpen
  \bibfield  {author} {\bibinfo {author} {\bibfnamefont {Joschka}\ \bibnamefont {Roffe}}, \bibinfo {author} {\bibfnamefont {David~R.}\ \bibnamefont {White}}, \bibinfo {author} {\bibfnamefont {Simon}\ \bibnamefont {Burton}}, \ and\ \bibinfo {author} {\bibfnamefont {Earl}\ \bibnamefont {Campbell}},\ }\bibfield  {title} {\enquote {\bibinfo {title} {Decoding across the quantum low-density parity-check code landscape},}\ }\href {\doibase 10.1103/physrevresearch.2.043423} {\bibfield  {journal} {\bibinfo  {journal} {Physical Review Research}\ }\textbf {\bibinfo {volume} {2}} (\bibinfo {year} {2020}),\ 10.1103/physrevresearch.2.043423}\BibitemShut {NoStop}%
\bibitem [{\citenamefont {Higgott}\ and\ \citenamefont {Gidney}(2023)}]{pymatching}%
  \BibitemOpen
  \bibfield  {author} {\bibinfo {author} {\bibfnamefont {Oscar}\ \bibnamefont {Higgott}}\ and\ \bibinfo {author} {\bibfnamefont {Craig}\ \bibnamefont {Gidney}},\ }\bibfield  {title} {\enquote {\bibinfo {title} {Sparse blossom: correcting a million errors per core second with minimum-weight matching},}\ }\href@noop {} {\bibfield  {journal} {\bibinfo  {journal} {arXiv preprint arXiv:2303.15933}\ } (\bibinfo {year} {2023})}\BibitemShut {NoStop}%
\bibitem [{\citenamefont {Berger}\ \emph {et~al.}(2005)\citenamefont {Berger}, \citenamefont {Kenyon}, \citenamefont {Mossel},\ and\ \citenamefont {Peres}}]{berger2005glauber}%
  \BibitemOpen
  \bibfield  {author} {\bibinfo {author} {\bibfnamefont {Noam}\ \bibnamefont {Berger}}, \bibinfo {author} {\bibfnamefont {Claire}\ \bibnamefont {Kenyon}}, \bibinfo {author} {\bibfnamefont {Elchanan}\ \bibnamefont {Mossel}}, \ and\ \bibinfo {author} {\bibfnamefont {Yuval}\ \bibnamefont {Peres}},\ }\bibfield  {title} {\enquote {\bibinfo {title} {Glauber dynamics on trees and hyperbolic graphs},}\ }\href@noop {} {\bibfield  {journal} {\bibinfo  {journal} {Probability Theory and Related Fields}\ }\textbf {\bibinfo {volume} {131}},\ \bibinfo {pages} {311--340} (\bibinfo {year} {2005})}\BibitemShut {NoStop}%
\bibitem [{\citenamefont {Zhao}\ \emph {et~al.}(2024)\citenamefont {Zhao}, \citenamefont {Doherty},\ and\ \citenamefont {Kim}}]{zhao2024}%
  \BibitemOpen
  \bibfield  {author} {\bibinfo {author} {\bibfnamefont {Guangqi}\ \bibnamefont {Zhao}}, \bibinfo {author} {\bibfnamefont {Andrew~C.}\ \bibnamefont {Doherty}}, \ and\ \bibinfo {author} {\bibfnamefont {Isaac~H.}\ \bibnamefont {Kim}},\ }\href {https://arxiv.org/abs/2407.20526} {\enquote {\bibinfo {title} {On the energy barrier of hypergraph product codes},}\ } (\bibinfo {year} {2024}),\ \Eprint {http://arxiv.org/abs/2407.20526} {arXiv:2407.20526 [quant-ph]} \BibitemShut {NoStop}%
\bibitem [{\citenamefont {Davies}(1979)}]{Davies_1979}%
  \BibitemOpen
  \bibfield  {author} {\bibinfo {author} {\bibfnamefont {E.B.}\ \bibnamefont {Davies}},\ }\bibfield  {title} {\enquote {\bibinfo {title} {Generators of dynamical semigroups},}\ }\href {\doibase https://doi.org/10.1016/0022-1236(79)90085-5} {\bibfield  {journal} {\bibinfo  {journal} {Journal of Functional Analysis}\ }\textbf {\bibinfo {volume} {34}},\ \bibinfo {pages} {421--432} (\bibinfo {year} {1979})}\BibitemShut {NoStop}%
\bibitem [{\citenamefont {Guo}\ \emph {et~al.}(2024)\citenamefont {Guo}, \citenamefont {Hart}, \citenamefont {Chen}, \citenamefont {Friedman},\ and\ \citenamefont {Lucas}}]{aqm}%
  \BibitemOpen
  \bibfield  {author} {\bibinfo {author} {\bibfnamefont {Jinkang}\ \bibnamefont {Guo}}, \bibinfo {author} {\bibfnamefont {Oliver}\ \bibnamefont {Hart}}, \bibinfo {author} {\bibfnamefont {Chi-Fang}\ \bibnamefont {Chen}}, \bibinfo {author} {\bibfnamefont {Aaron~J.}\ \bibnamefont {Friedman}}, \ and\ \bibinfo {author} {\bibfnamefont {Andrew}\ \bibnamefont {Lucas}},\ }\href@noop {} {\enquote {\bibinfo {title} {Designing open quantum systems with known steady states: Davies generators and beyond},}\ } (\bibinfo {year} {2024}),\ \Eprint {http://arxiv.org/abs/2404.14538} {arXiv:2404.14538 [quant-ph]} \BibitemShut {NoStop}%
\bibitem [{\citenamefont {Knuth}(1997)}]{Art_of_CS}%
  \BibitemOpen
  \bibfield  {author} {\bibinfo {author} {\bibfnamefont {Donald~E.}\ \bibnamefont {Knuth}},\ }\href@noop {} {\emph {\bibinfo {title} {The art of computer programming, volume 1 (3rd ed.): fundamental algorithms}}}\ (\bibinfo  {publisher} {Addison Wesley Longman Publishing Co., Inc.},\ \bibinfo {address} {USA},\ \bibinfo {year} {1997})\BibitemShut {NoStop}%
\bibitem [{\citenamefont {Bundala}\ and\ \citenamefont {Závodný}(2013)}]{bundala2013}%
  \BibitemOpen
  \bibfield  {author} {\bibinfo {author} {\bibfnamefont {Daniel}\ \bibnamefont {Bundala}}\ and\ \bibinfo {author} {\bibfnamefont {Jakub}\ \bibnamefont {Závodný}},\ }\href@noop {} {\enquote {\bibinfo {title} {Optimal sorting networks},}\ } (\bibinfo {year} {2013}),\ \Eprint {http://arxiv.org/abs/1310.6271} {arXiv:1310.6271 [cs.DM]} \BibitemShut {NoStop}%
\bibitem [{\citenamefont {Codish}\ \emph {et~al.}(2014)\citenamefont {Codish}, \citenamefont {Cruz-Filipe}, \citenamefont {Frank},\ and\ \citenamefont {Schneider-Kamp}}]{Codish_2014}%
  \BibitemOpen
  \bibfield  {author} {\bibinfo {author} {\bibfnamefont {Michael}\ \bibnamefont {Codish}}, \bibinfo {author} {\bibfnamefont {Luis}\ \bibnamefont {Cruz-Filipe}}, \bibinfo {author} {\bibfnamefont {Michael}\ \bibnamefont {Frank}}, \ and\ \bibinfo {author} {\bibfnamefont {Peter}\ \bibnamefont {Schneider-Kamp}},\ }\bibfield  {title} {\enquote {\bibinfo {title} {Twenty-five comparators is optimal when sorting nine inputs (and twenty-nine for ten)},}\ }in\ \href {\doibase 10.1109/ictai.2014.36} {\emph {\bibinfo {booktitle} {2014 IEEE 26th International Conference on Tools with Artificial Intelligence}}}\ (\bibinfo  {publisher} {IEEE},\ \bibinfo {year} {2014})\BibitemShut {NoStop}%
\bibitem [{\citenamefont {Harder}(2022)}]{harder2022}%
  \BibitemOpen
  \bibfield  {author} {\bibinfo {author} {\bibfnamefont {Jannis}\ \bibnamefont {Harder}},\ }\href@noop {} {\enquote {\bibinfo {title} {An answer to the bose-nelson sorting problem for 11 and 12 channels},}\ } (\bibinfo {year} {2022}),\ \Eprint {http://arxiv.org/abs/2012.04400} {arXiv:2012.04400 [cs.DS]} \BibitemShut {NoStop}%
\bibitem [{\citenamefont {Shi}(2018)}]{Shi_2018}%
  \BibitemOpen
  \bibfield  {author} {\bibinfo {author} {\bibfnamefont {Xiao-Feng}\ \bibnamefont {Shi}},\ }\bibfield  {title} {\enquote {\bibinfo {title} {Deutsch, toffoli, and cnot gates via rydberg blockade of neutral atoms},}\ }\href {\doibase 10.1103/PhysRevApplied.9.051001} {\bibfield  {journal} {\bibinfo  {journal} {Phys. Rev. Appl.}\ }\textbf {\bibinfo {volume} {9}},\ \bibinfo {pages} {051001} (\bibinfo {year} {2018})}\BibitemShut {NoStop}%
\bibitem [{\citenamefont {Yin}\ \emph {et~al.}(2020)\citenamefont {Yin}, \citenamefont {Li}, \citenamefont {Wang},\ and\ \citenamefont {Shao}}]{Yin_2020}%
  \BibitemOpen
  \bibfield  {author} {\bibinfo {author} {\bibfnamefont {Hong-Da}\ \bibnamefont {Yin}}, \bibinfo {author} {\bibfnamefont {Xiao-Xuan}\ \bibnamefont {Li}}, \bibinfo {author} {\bibfnamefont {Gang-Cheng}\ \bibnamefont {Wang}}, \ and\ \bibinfo {author} {\bibfnamefont {Xiao-Qiang}\ \bibnamefont {Shao}},\ }\bibfield  {title} {\enquote {\bibinfo {title} {One-step implementation of toffoli gate for neutral atoms based on unconventional rydberg pumping},}\ }\href {\doibase 10.1364/oe.410158} {\bibfield  {journal} {\bibinfo  {journal} {Optics Express}\ }\textbf {\bibinfo {volume} {28}},\ \bibinfo {pages} {35576} (\bibinfo {year} {2020})}\BibitemShut {NoStop}%
\bibitem [{\citenamefont {Batcher}(1968)}]{Batcher_1968}%
  \BibitemOpen
  \bibfield  {author} {\bibinfo {author} {\bibfnamefont {K.~E.}\ \bibnamefont {Batcher}},\ }\bibfield  {title} {\enquote {\bibinfo {title} {Sorting networks and their applications},}\ }in\ \href {\doibase 10.1145/1468075.1468121} {\emph {\bibinfo {booktitle} {Proceedings of the April 30--May 2, 1968, Spring Joint Computer Conference}}},\ \bibinfo {series and number} {AFIPS '68 (Spring)}\ (\bibinfo  {publisher} {Association for Computing Machinery},\ \bibinfo {address} {New York, NY, USA},\ \bibinfo {year} {1968})\ p.\ \bibinfo {pages} {307–314}\BibitemShut {NoStop}%
\bibitem [{\citenamefont {M.~Q.~Cruz}\ and\ \citenamefont {Murta}(2024)}]{Cruz_2024}%
  \BibitemOpen
  \bibfield  {author} {\bibinfo {author} {\bibfnamefont {Pedro}\ \bibnamefont {M.~Q.~Cruz}}\ and\ \bibinfo {author} {\bibfnamefont {Bruno}\ \bibnamefont {Murta}},\ }\bibfield  {title} {\enquote {\bibinfo {title} {Shallow unitary decompositions of quantum fredkin and toffoli gates for connectivity-aware equivalent circuit averaging},}\ }\href {\doibase 10.1063/5.0187026} {\bibfield  {journal} {\bibinfo  {journal} {APL Quantum}\ }\textbf {\bibinfo {volume} {1}} (\bibinfo {year} {2024}),\ 10.1063/5.0187026}\BibitemShut {NoStop}%
\bibitem [{\citenamefont {Kheirandish}\ \emph {et~al.}(2020)\citenamefont {Kheirandish}, \citenamefont {Haghparast}, \citenamefont {Reshadi},\ and\ \citenamefont {Hosseinzadeh}}]{Kheirandish_2020}%
  \BibitemOpen
  \bibfield  {author} {\bibinfo {author} {\bibfnamefont {Davar}\ \bibnamefont {Kheirandish}}, \bibinfo {author} {\bibfnamefont {Majid}\ \bibnamefont {Haghparast}}, \bibinfo {author} {\bibfnamefont {Midia}\ \bibnamefont {Reshadi}}, \ and\ \bibinfo {author} {\bibfnamefont {Mehdi}\ \bibnamefont {Hosseinzadeh}},\ }\bibfield  {title} {\enquote {\bibinfo {title} {Efficient designs of reversible majority voters},}\ }\href {https://api.semanticscholar.org/CorpusID:230642857} {\bibfield  {journal} {\bibinfo  {journal} {Journal of Electronic Testing}\ }\textbf {\bibinfo {volume} {36}},\ \bibinfo {pages} {757 -- 770} (\bibinfo {year} {2020})}\BibitemShut {NoStop}%
\end{thebibliography}%
